\documentclass[10pt,a4paper]{article}

\usepackage[pdftex]{graphicx}
\usepackage{multirow}
\usepackage{morefloats}     
\usepackage{url}    
\usepackage{algorithm}
\usepackage{amsthm}         

\include{Definitions}     

\graphicspath{{./Figs/}}

\theoremstyle{definition}
\newtheorem{defn}{Definition}[section]

\begin{document}

\title{Tiling rectangles with holey polyominoes}

\author{Dmitry Kamenetsky and Tristrom Cooke \\
dkamenetsky@gmail.com, tcooke@internode.on.net \\
Adelaide, Australia
}

\maketitle

\abstract{We present a new type of polyominoes that can have transparent squares (holes).
We show how these polyominoes can tile rectangles and we categorise them according to their tiling ability.
We were able to categorise all but 6 polyominoes with 5 or fewer visible squares.
}

\section{Introduction}

Polyominoes are geometric shapes made from squares.
We introduce \emph{holey} polyominoes as polyominoes that contain transparent squares, making them disjoint.
We now give formal definitions.

\begin{defn}
A \emph{polyomino} is a geometric shape formed by the union of non-overlapping squares edge to edge such that at least two
corners of the squares touch.
\end{defn}

\begin{defn}
A \emph{holey polyomino} of order $(n,k)$ is a polyomino with $n$ visible squares and $k$ transparent squares.
$k$ must be the least number of transparent squares required to connect all the visible squares.
\end{defn}

From now on we refer to holey polyominoes of order $(n,k)$ as $(n,k)$-polyominoes.
For example, Figure~\ref{fig:numHoles} shows a $(4,1)$-polyomino.
Visible squares are shown in blue, while the
transparent squares are in white. 
Although we can make all visible squares connected
via squares a and c, it is sufficient to just use square b, making $k=1$.
Figure~\ref{fig:31} shows all four possible $(3,1)$-polyominoes.
Note that regular polyominoes are a subset of holey polyominoes; a regular polyomino of order $n$ is a
$(n,0)$-polyomino.

\begin{figure}[!htpb]
\centering
\includegraphics[width=0.25\linewidth]{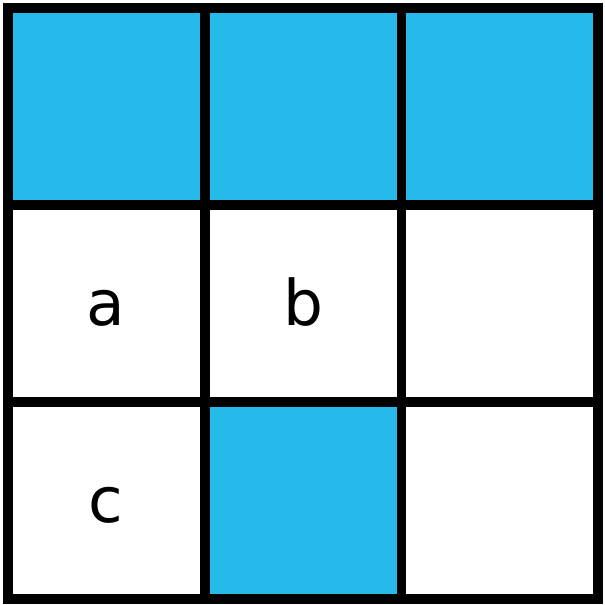}
\caption{A $(4,1)$-polyomino. We can make all visible squares (in blue) connected by making square b visible.}
\label{fig:numHoles}
\end{figure}

\begin{figure}[!htpb]
\centering
\begin{tabular}{cccc}
\includegraphics[width=0.207\linewidth]{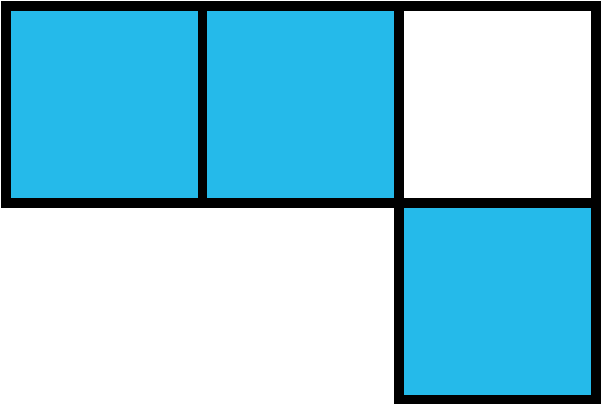} & \includegraphics[width=0.207\linewidth]{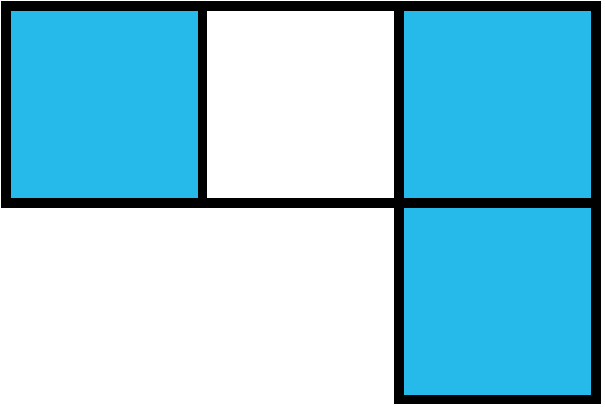} &
\includegraphics[width=0.207\linewidth]{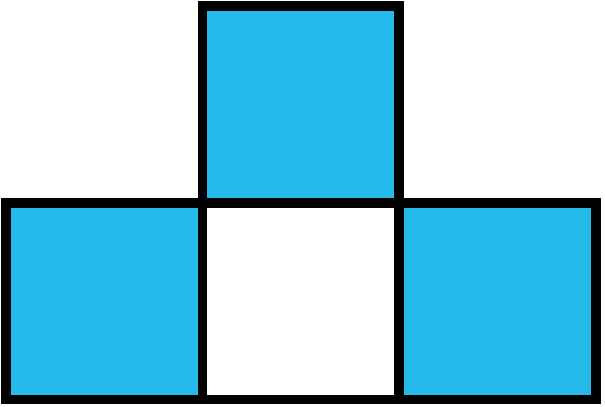} & \includegraphics[width=0.27\linewidth]{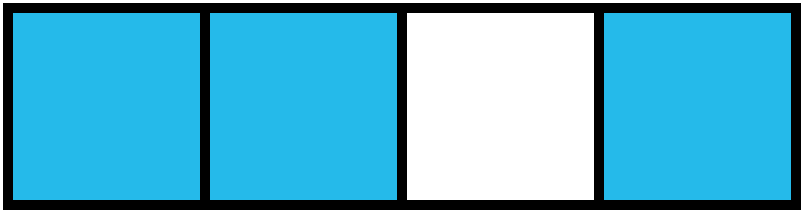} 
\end{tabular}
\caption{All $(3,1)$-polyominoes.}
\label{fig:31}
\end{figure}

We assume that polyominoes are free, meaning that they can be flipped, mirrored and rotated at will. 
Figure~\ref{fig:G} shows all 8 congruent versions of a single $(2,2)$-polyomino.

\begin{figure}[!htpb]
\centering
\includegraphics[width=0.8\linewidth]{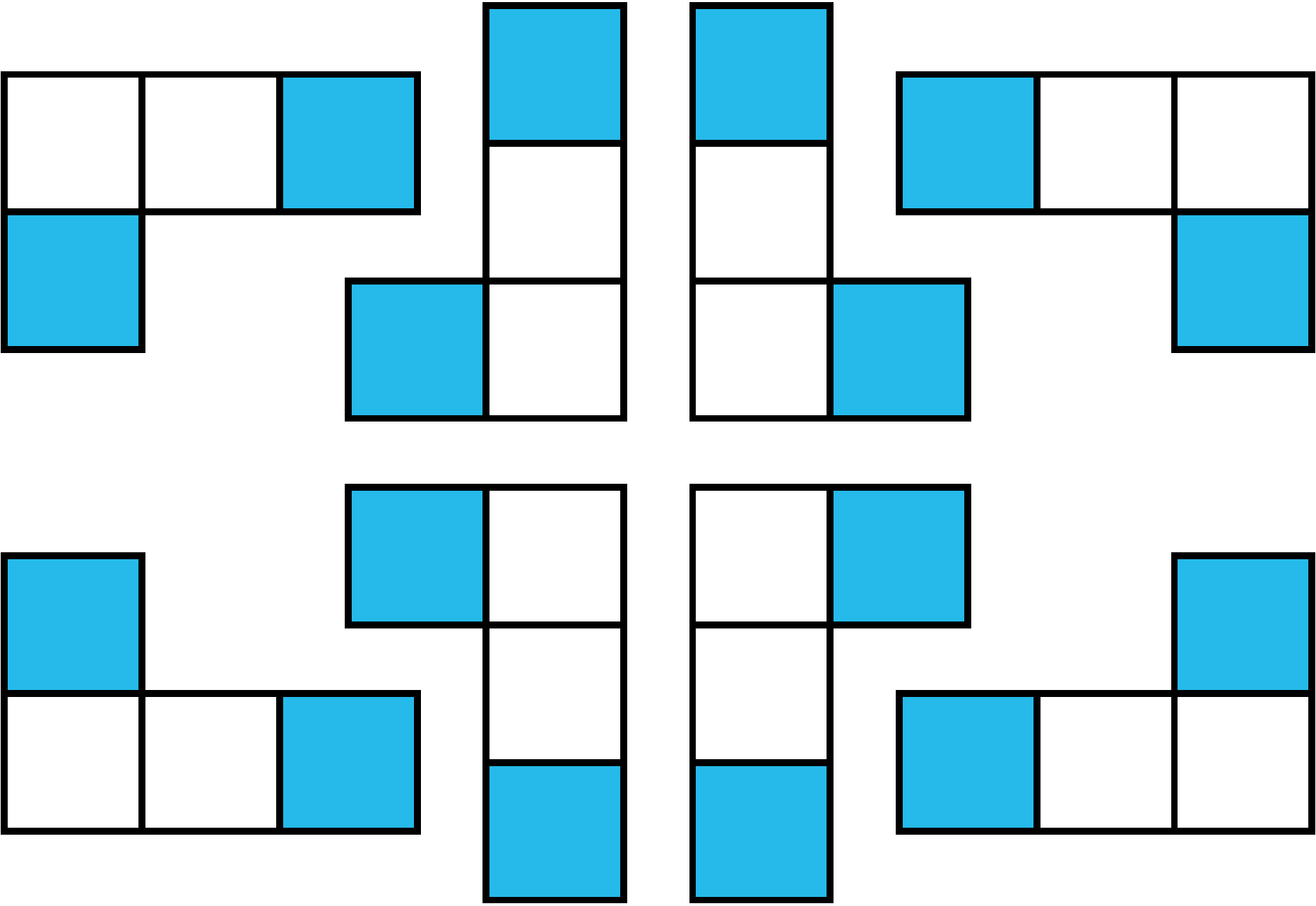}
\caption{All 8 congruent versions of a $(2,2)$-polyomino.}
\label{fig:G}
\end{figure}

\section{Related work}

The concept of polyominoes can be traced back to ancient times. The famous British puzzle expert Henry Dudeney and the Fairy Chess Review
created problems involving n-ominoes, which they represented as figures cut from checkerboards \cite{klarner65}.
The term \emph{polyominoes} was coined by Solomon Golomb in 1953 and later popularised by Martin Gardner \cite{golomb96}.
Polyominoes gained a great deal of popularity through tiling puzzles and games such as Tetris and Blokus.

Numerous variations to polyominoes have been suggested over the years, such as: polyiamonds (from equilateral triangles),
polyhexes (from regular hexagons), polycubes (from cubes). All of these however, do not involve any transparent squares.
Perhaps the closest variation is polyplets: polyomino-like objects made by attaching squares joined either at sides
or corners (see Figure~\ref{fig:polyplets} top row). Note that polyplets are a subset of holey polyominoes. Another related idea is rounded
polyominoes \cite{Harsh} (see Figure~\ref{fig:polyplets} bottom row). These are polyplets with rounded corners and bridges connecting diagonally adjacent squares.
Unlike polyplets and holey polyominoes, rounded polyominoes can be made in the physical world.

Polyomino tiling problems ask whether copies of a single polyomino can tile (cover) a given region, such as a plane or a rectangle.
In 1960's Golomb~\cite{golomb66} and Klarner~\cite{klarner65,klarner69} were the first to study these problems for particular polyominoes.
For the problem of rectangular tiling, results have been found for various polyominoes~\cite{Reid:results,Friedman:rectifiable}. Mark Reid
provides extensive literature on this subject on his site \cite{Reid:literature}.

Tiling with disjoint polyominoes has been studied for certain classes of polyominoes.
Chvatal et al.,~\cite[problem 8]{Chvatal72} showed that any single-row $(3,k)$-polyomino can tile a $1 \times n$ rectangle.
Gordon~\cite{Gordon80} showed that any n-dimensional $(3,k)$-polyomino can tile the $\mathbb{Z}^n$ lattice. In their problem 9, 
Chvatal et al.,~\cite{Chvatal72} ask whether every $(4,k)$-polyomino tiles the plane? In 1985 this question was answered with
affirmative by Coppersmith~\cite{coppersmith85}. Friedman~\cite{Friedman:disjoint} has collected results about tiling rectangles
with single-row $(n,k)$-polyominoes for $n+k \leq 9$ and $k \geq 1$.

\begin{figure}[!htpb]
\centering
\begin{tabular}{ccccc}
\includegraphics[width=0.195\linewidth]{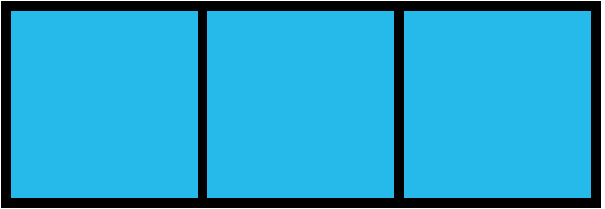} & 
\includegraphics[width=0.195\linewidth]{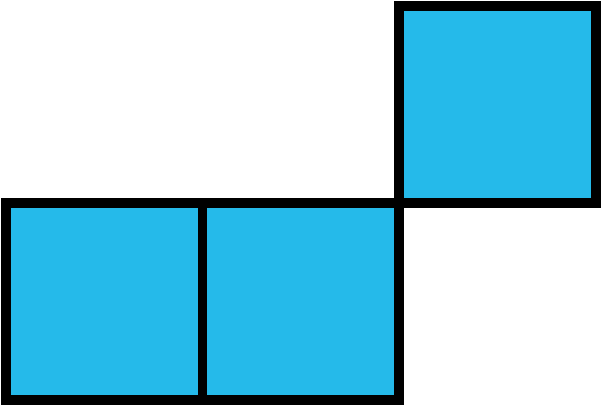} & 
\includegraphics[width=0.195\linewidth]{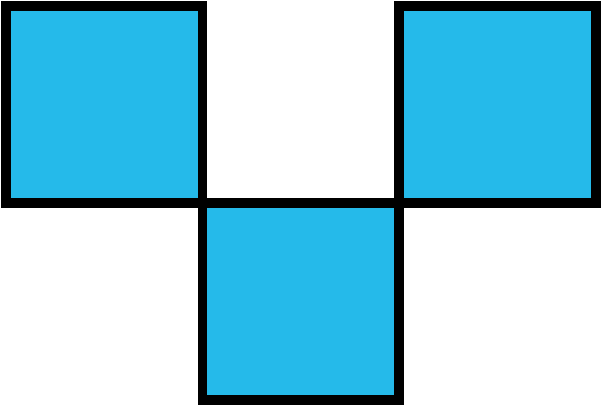} & 
\includegraphics[width=0.195\linewidth]{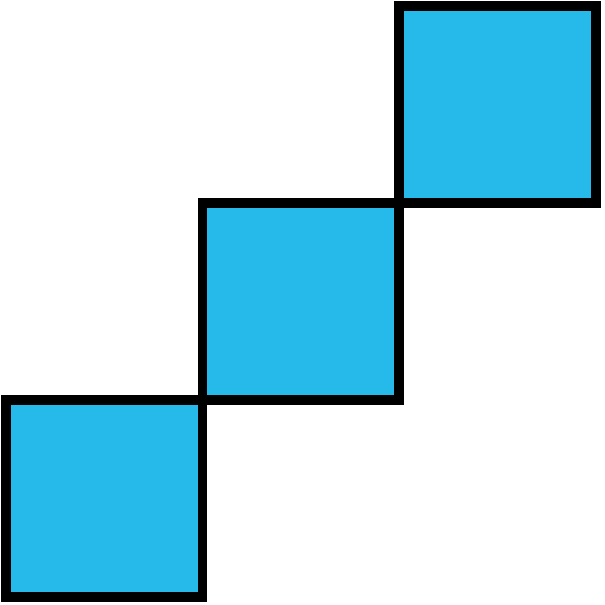} & 
\includegraphics[width=0.13\linewidth]{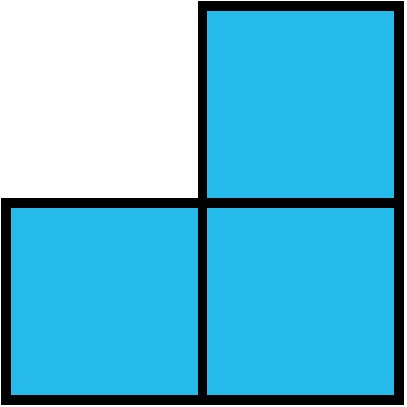}
\end{tabular}
\includegraphics[width=0.5\linewidth]{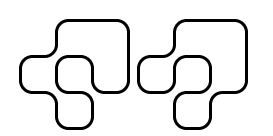}
\caption{Top row: all polyplets of order 3. Bottom row: two rounded pentominoes with the same squares, but different connections (bridges).
Image courtesy of http://www.ericharshbarger.org/pentominoes/article\_09.html}
\label{fig:polyplets}
\end{figure}


\section{Rectangular tilings}

We now consider the problem of tiling rectangles with holey polyominoes. During tiling, a visible square may lie
on top of a transparent square (or vice-versa), but it cannot lie on top of another visible square. Naturally, transparent squares can lie
on top of other transparent squares. If possible, we are interested in finding the smallest rectangle (in area)
that can be tiled by a single $(n,k)$-polyomino.

We have investigated rectangular tilings of the following classes of polyominoes:
$(2,1), (2,2), (3,1), (3,2), (4,1)$ and $(5,1)$. We call a holey polyomino \emph{rectifiable} (or \emph{solved})
if it can tile a rectangle. We call it \emph{unrectifiable} (or \emph{impossible}) if we can show
that it cannot tile a rectangle. Otherwise, we call it \emph{unknown}.
Table~\ref{tab:summary} summarises our results.

\begin{table}[!htpb]
\centering
\begin{tabular}{|c|c|c|c|c|c|}
\hline
n & k & Solved & Impossible & Unknown & Total \\ \hline
2 & 1 & 2 & 0 & 0 & 2 \\ \hline
2 & 2 & 2 & 0 & 0 & 2 \\ \hline
3 & 1 & 4 & 0 & 0 & 4 \\ \hline
3 & 2 & 9 & 2 & 0 & 11 \\ \hline
4 & 1 & 14 & 6 & 0 & 20 \\ \hline
5 & 1 & 28 & 34 & 6 & 68 \\ \hline
\end{tabular}
\caption{Summary of results.}
\label{tab:summary}
\end{table}

\subsection{Algorithm}

We use the classical Depth First Search (DFS) algorithm to find the status of each polyomino (see Algorithm~\ref{alg:solver}).
The root node of the search tree is the empty $N \times N$ grid. A move consists of placing a single polyomino in the current
grid. The search terminates when one of the following conditions is met:
\begin{enumerate}
\item We completely fill a rectangular area with tiles. 
If this happens then we have found a solution and we can terminate (see lines 4-8).
The tiled rectangle may be smaller than $N \times N$.
\item We have visited every possible node in the tree without finding a solution.
This means that there is no solution with a rectangle
whose maximum side is $N$ or less. This however, does not exclude solutions whose maximum side is greater than $N$.
\end{enumerate}

A move consists of placing a tile in an empty grid location. We try all possible orientations and shifts
of the tile, as given by the $Rot(\cdot)$ function on line 19.
For any given node there can be a great number of possible moves that can be made. 
Trying all moves exhaustively may be too slow or even infeasible.
Hence we would like to make moves that bring us closer to search termination. 

For this reason we order our moves from most to least constrained. An empty location is considered to be more constrained if there are
less unique ways of filling it with tiles. Line 10 of the algorithm computes how constrained each empty location is. 
Line 11 computes the most constrained location $(r^*,c^*)$.

If a polyomino is rectifiable then it must be able to tile the corner of a rectangle.
Since corners are the most constrained locations in the grid, we begin our search there.
As soon as we find a location that cannot be filled by a tile, we can backtrack (see lines 13-16).

\emph{Solve} is a recursive function that takes the following parameters:
\begin{itemize}
\item \emph{grid}: a working array of placed tiles. Initially all squares are empty.
\item \emph{tile}: the tile that we are trying to place.
\item \emph{moves}: a list of moves that we have already made.
\item \emph{R}: the index of the first row that is completely empty.
\item \emph{C}: the index of the first column that is completely empty.
\end{itemize}

The algorithm begins by initialising an empty $N \times N$ \emph{grid} and an empty list of \emph{moves}.
We then call Solve(grid, tile, moves, 0, 0). If a solution exists, then this algorithm will find one,
although it may not be the smallest.

\begin{algorithm}[!htpb]
\renewcommand{\arraystretch}{1.15}
\caption{: Algorithm for finding solutions.}
\vspace{0.5ex}
\bf function Solve(grid, tile, moves, R, C) \\
\begin{tabular}{rl}
\\
1 & $//set~of~empty~locations~within~bounds$ \\
2 & $E:=\{(r,c):~grid(r,c)~is~empty,~r<R,~c<C\}$ \\
3 & \\
4 & $//no~more~empty~locations,~so~we~found~a~R \times C~solution$ \\
5 & if $E=\emptyset$\\
6 & ~~$grid.print()$ \\
7 & ~~terminate \\
8 & end \\
9 & \\
10 & $\forall (r,c) \in E: G(r,c):=~\#~ways~to~fill~grid(r,c)$ \\
11 & $(r^*,c^*):=\argmin_{(r,c)} G(r,c)~~~~//most~constrained~location$ \\
12 & \\
13 & $//no~solution, so~backtrack$ \\
14 & if $G(r^*,c^*)=0$ \\
15 & ~~return \\
16 & end \\
17 & \\
18 & $//for~each~rotated~and~shifted~version~of~the~tile$ \\
19 & $\forall t \in Rot(tile)$ \\
20 & ~~$move=~$new $Move(t,r^*,c^*)$ \\
21 & ~~if $grid.isValidMove(move)$ \\
22 & ~~~~$grid.makeMove(move)$ \\
23 & ~~~~$moves.add(move)$ \\
24 & \\
25 & ~~~~$R':=max(R,r^*+t.height)~~~~//compute~new~bounds$\\
26 & ~~~~$C':=max(C,c^*+t.width)$\\
27 & \\
28 & ~~~~$Solve(grid,tile,moves,R',C')$ ~~~~$//recurse$\\
29 & \\
30 & ~~~~$grid.undoMove(move)$ \\
31 & ~~~~$moves.remove(move)$ \\
32 & ~~end \\
33 & end \\
\end{tabular}
\label{alg:solver}
\end{algorithm}


\pagebreak
\subsection{Rectifiable polyominoes}

For certain classes of polyominoes (such as the $(2,k)$ class) we were able to show that all polyominoes within that class are rectifiable.

\begin{theorem}
Every $(2,k)$-polyomino is rectifiable.
\end{theorem}
\begin{proof}
Consider a large finite grid with the top-left corner being at row$=0$ and column$=0$. We can describe any 
$(2,k)$-polyomino based on the location of its two visible squares in the grid. So $(r_1,c_1,r_2,c_2)$ represents
a polyomino with its first visible square located at $(r_1,c_1)$ and second square at $(r_2,c_2)$.
We now show how any $(2,k)$-polyomino can tile a rectangle. Without loss of generality, we place the first
polyomino at $(0,0,r,c)$ as shown in Figure~\ref{fig:2N}(a). We then place $(c-1)$ polyominoes into locations
$(0,1,r,c+1), \ldots, (0,c-1,r,2c-1)$ as shown in Figure~\ref{fig:2N}(b). We now place $c$ tiles into locations
$(0,c,r,0), \ldots, (0,2c-1,r,c-1)$ as shown in Figure~\ref{fig:2N}(c). We now have 2 rows of squares of length $2c$ at
rows $0$ and $r$. Now replicate this structure and shift it one row down. After $(r-1)$ repetitions of this process
we will get a filled rectangle with $2r$ rows and $2c$ columns.
\end{proof}

\begin{figure}[!htpb]
\centering
\begin{tabular}{ccc}
\includegraphics[width=0.228\linewidth]{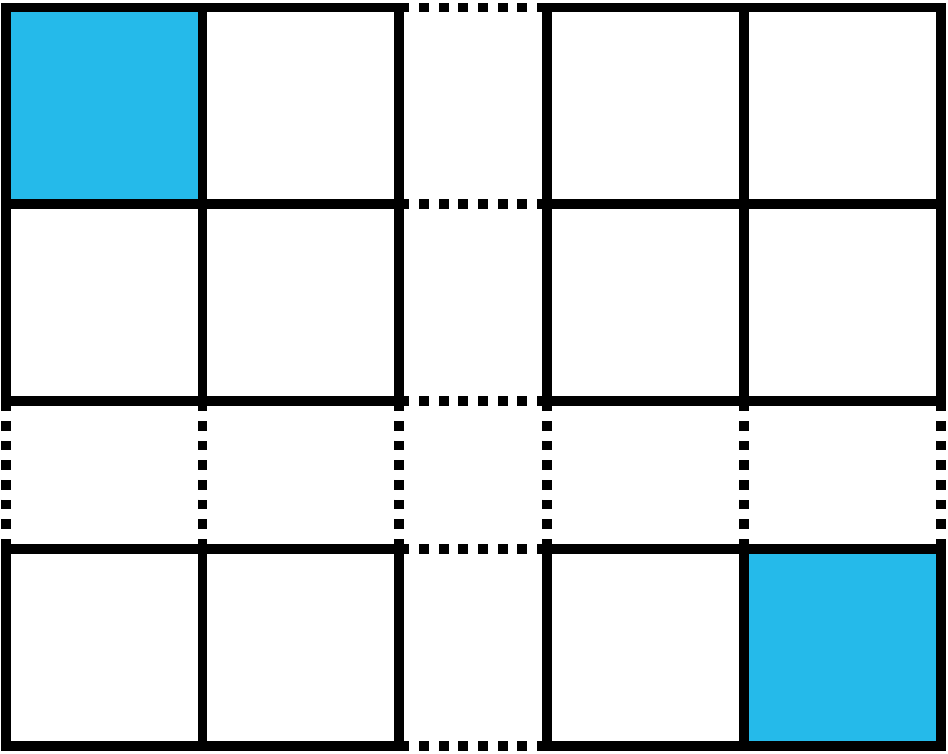} & \includegraphics[width=0.36\linewidth]{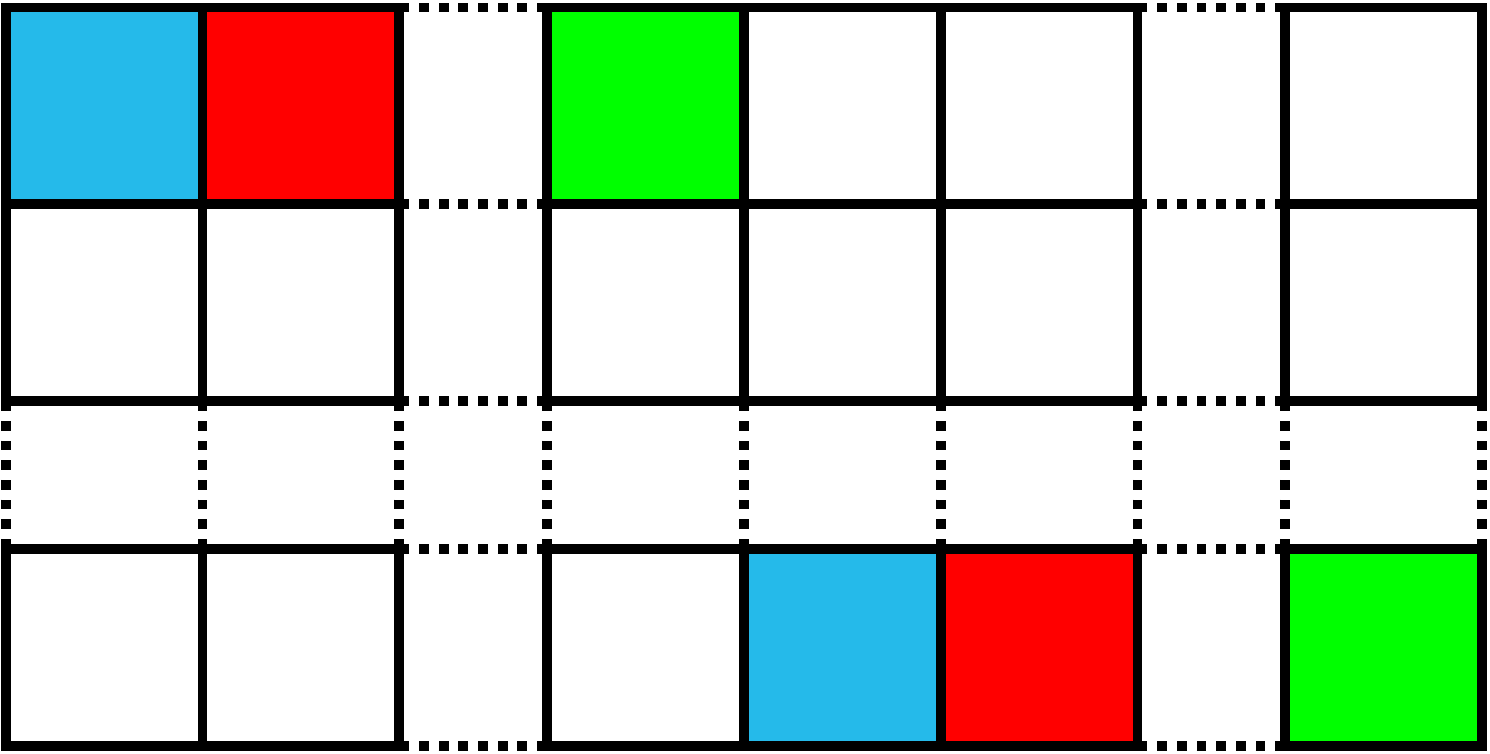} & \includegraphics[width=0.36\linewidth]{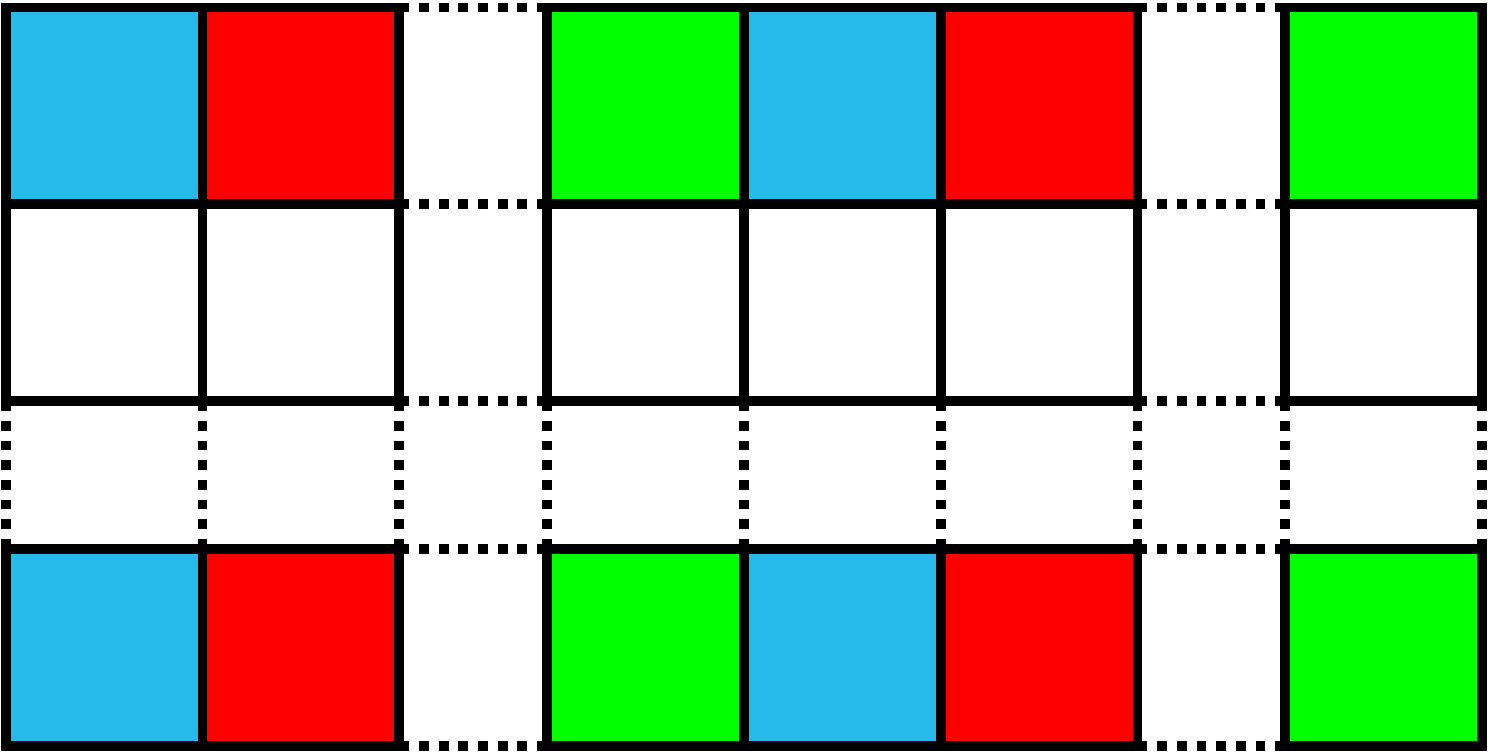} \\
(a) & (b) & (c)
\end{tabular}
\caption{Steps involved in showing that every $(2,k)$-polyomino is rectifiable.}
\label{fig:2N}
\end{figure}

Tables~\ref{tab:trivial-21}-\ref{tab:trivial-51b} show the smallest known tiling sizes.
Some of the smaller solutions are simple enough to find by hand and are optimal - 
tilings with the smallest area.
Tables~\ref{tab:cool1}-\ref{tab:cool2} show the actual tilings for
some non-trivial pieces. These solutions are not guaranteed to be optimal.

\begin{table}[!htpb]
\centering
\begin{tabular}{|c|c|c|c|}
\hline
Piece & Smallest Tiling & Piece & Smallest Tiling\\ \hline
& & & \\ 
\includegraphics[width=0.120000\linewidth]{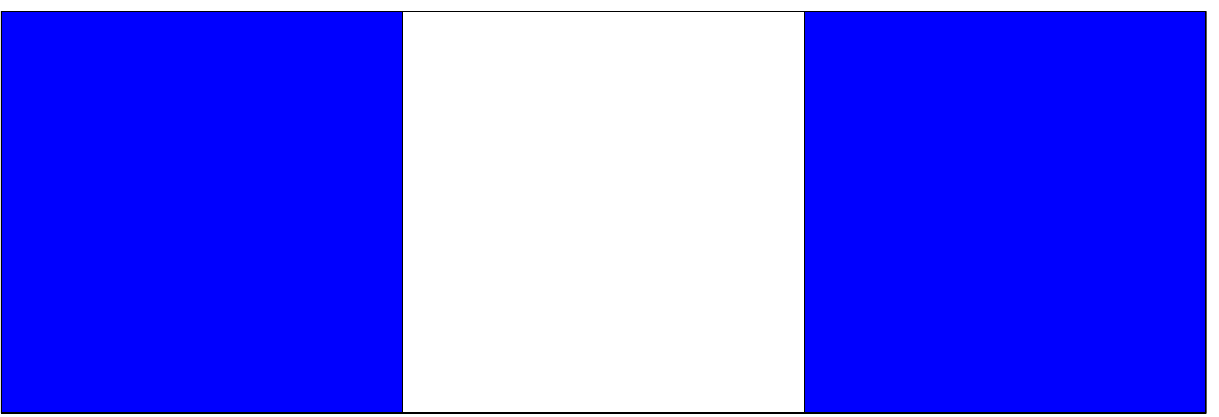} & 1x4 & \includegraphics[width=0.080000\linewidth]{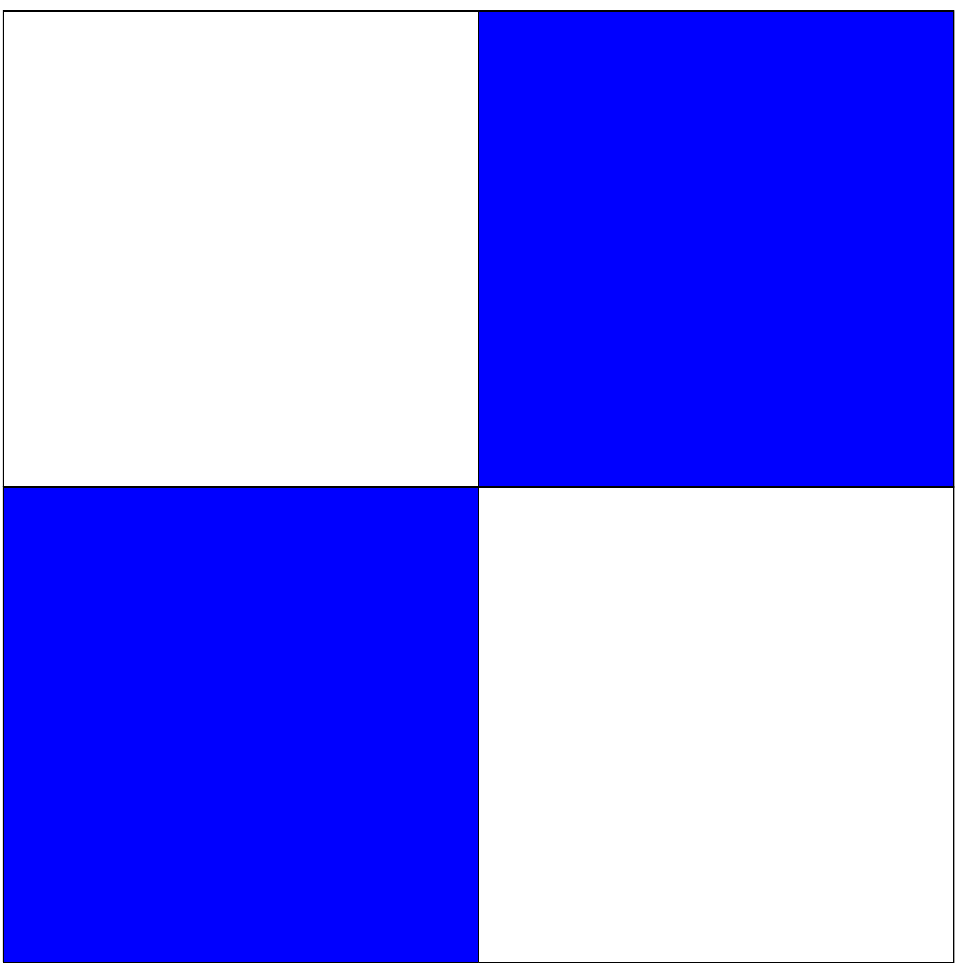} & 2x2 \\ \hline
\end{tabular}
\caption{Smallest known solutions for n=2 and k=1.}
\label{tab:trivial-21}
\end{table}

\begin{table}[!htpb]
\centering
\begin{tabular}{|c|c|c|c|}
\hline
Piece & Smallest Tiling & Piece & Smallest Tiling\\ \hline
& & & \\ 
\includegraphics[width=0.160000\linewidth]{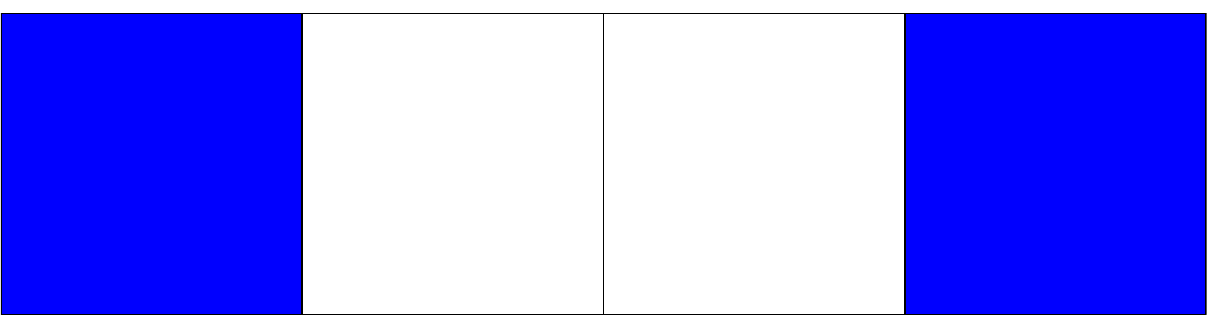} & 1x6 & \includegraphics[width=0.120000\linewidth]{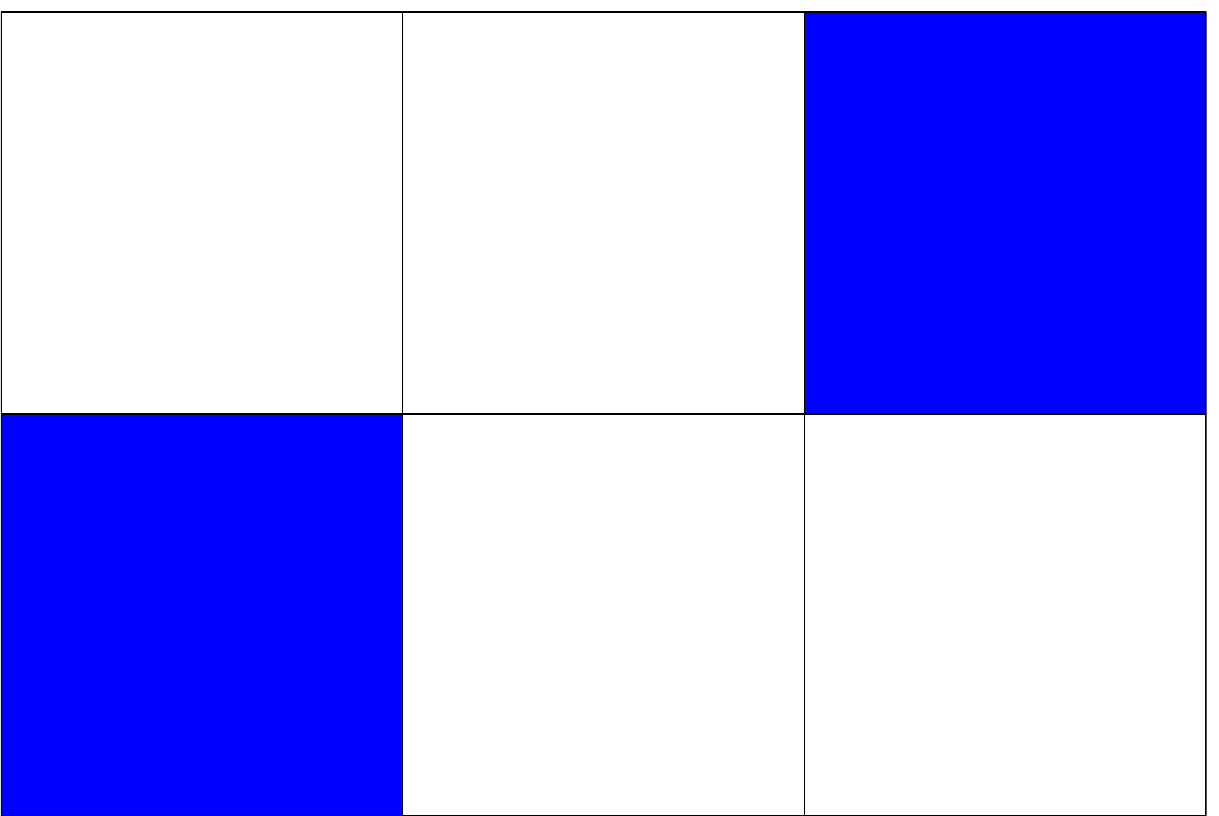} & 2x4 \\ \hline
\end{tabular}
\caption{Smallest known solutions for n=2 and k=2.}
\label{tab:trivial-22}
\end{table}

\begin{table}[!htpb]
\centering
\begin{tabular}{|c|c|c|c|}
\hline
Piece & Smallest Tiling & Piece & Smallest Tiling\\ \hline
& & & \\ 
\includegraphics[width=0.160000\linewidth]{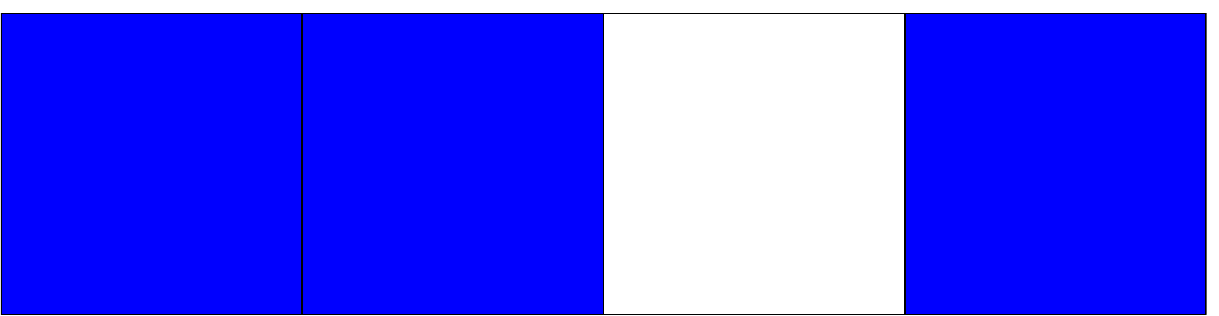} & 1x6 & \includegraphics[width=0.120000\linewidth]{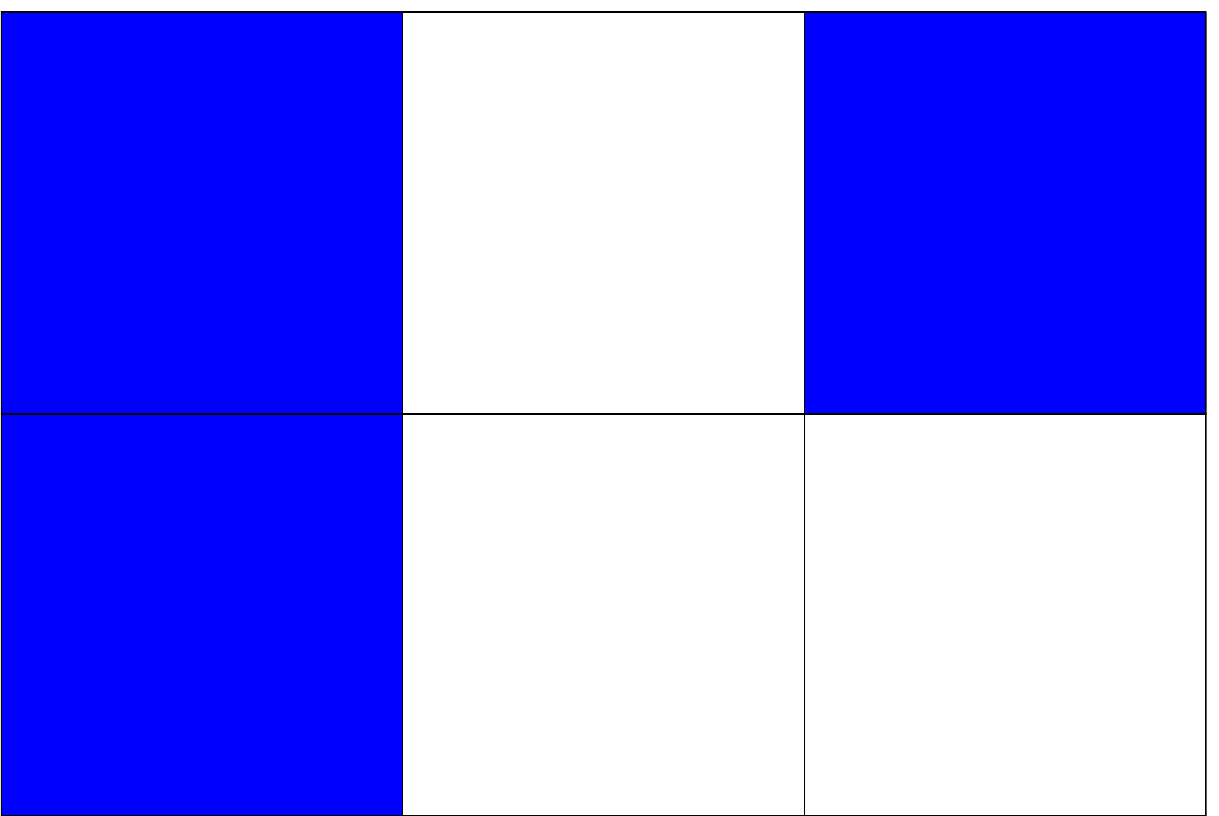} & 3x4 \\ \hline
& & & \\ 
\includegraphics[width=0.120000\linewidth]{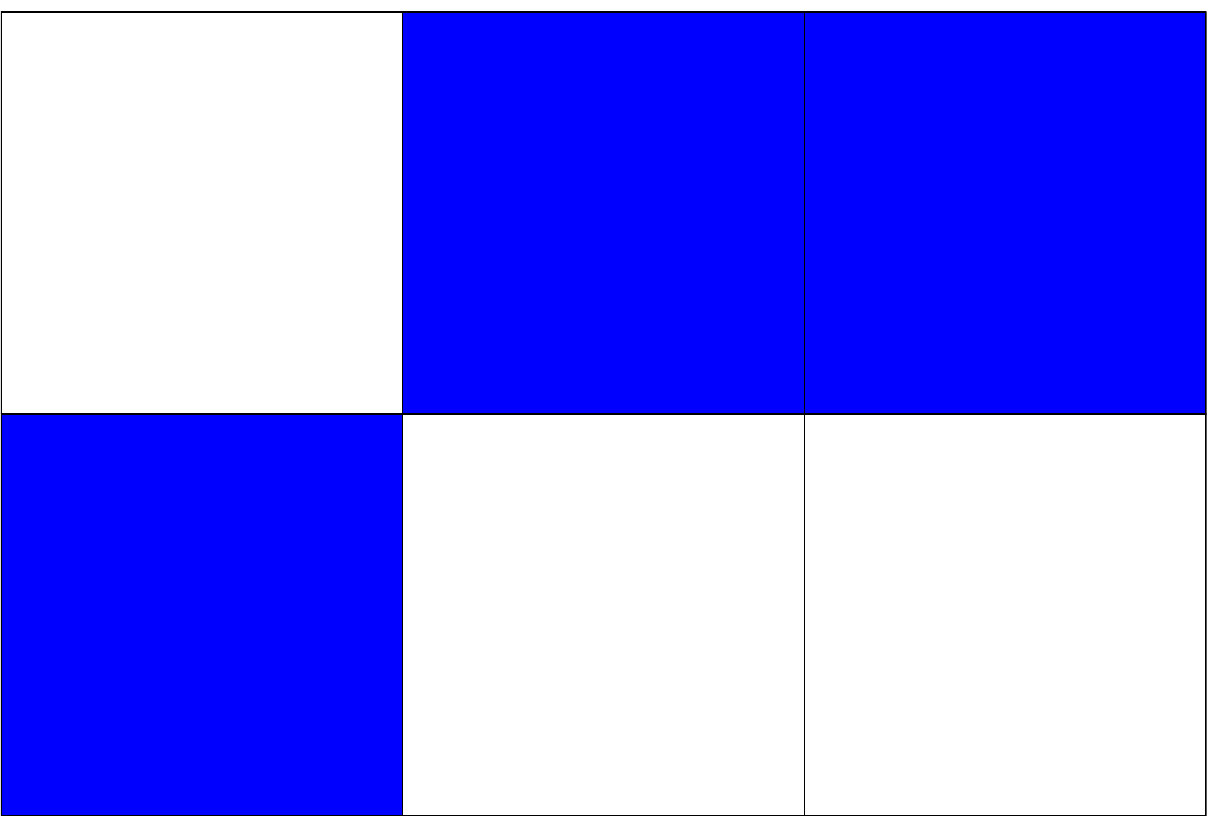} & 2x3 & \includegraphics[width=0.120000\linewidth]{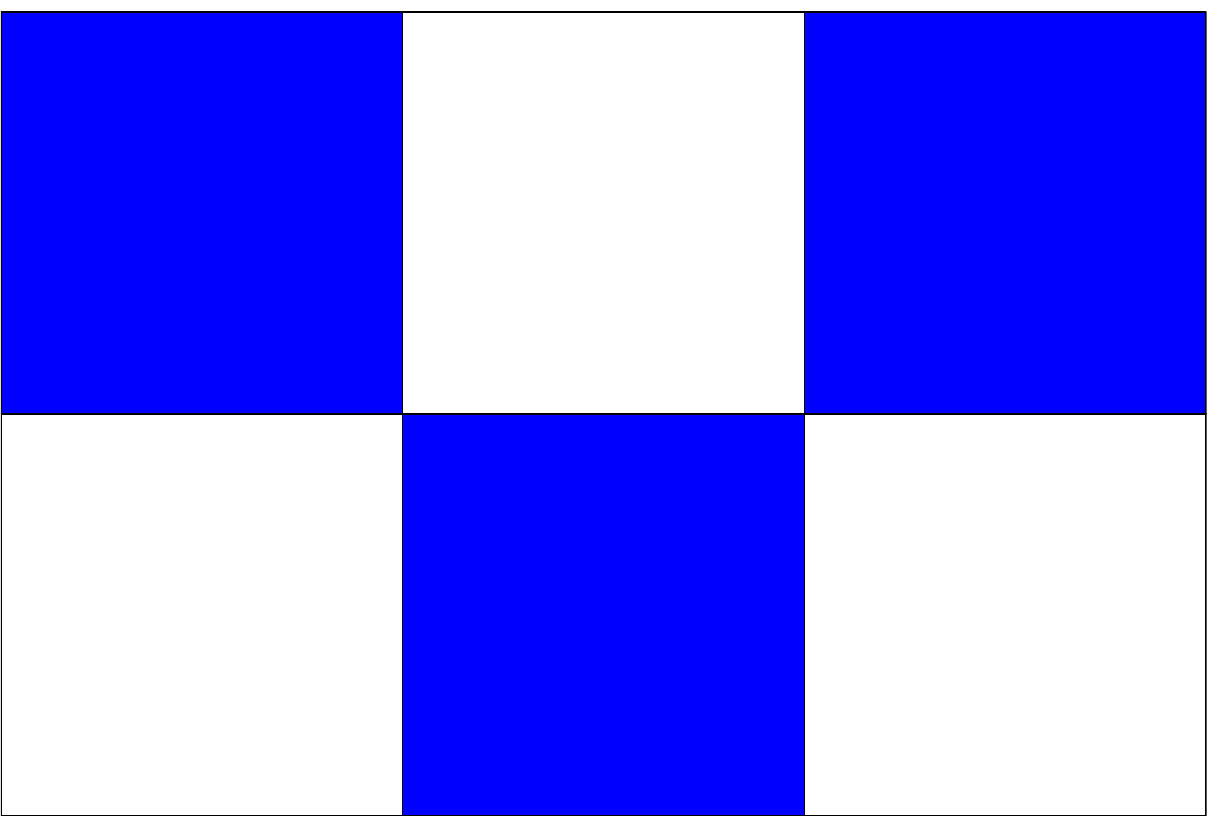} & 2x3 \\ \hline
\end{tabular}
\caption{Smallest known solutions for n=3 and k=1.}
\label{tab:trivial-31}
\end{table}

\begin{table}[!htpb]
\centering
\begin{tabular}{|c|c|c|c|}
\hline
Piece & Smallest Tiling & Piece & Smallest Tiling\\ \hline
& & & \\ 
\includegraphics[width=0.200000\linewidth]{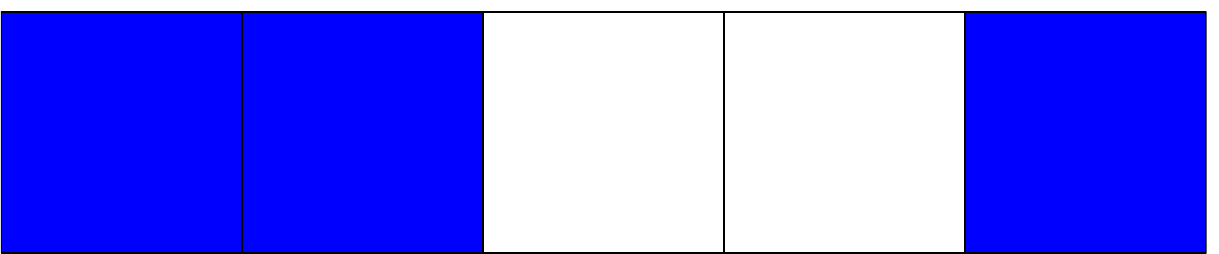} & 1x6 & \includegraphics[width=0.200000\linewidth]{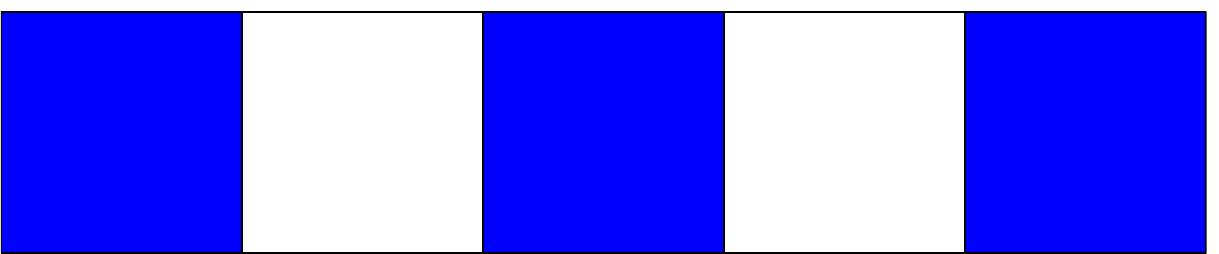} & 1x6 \\ \hline
& & & \\ 
\includegraphics[width=0.160000\linewidth]{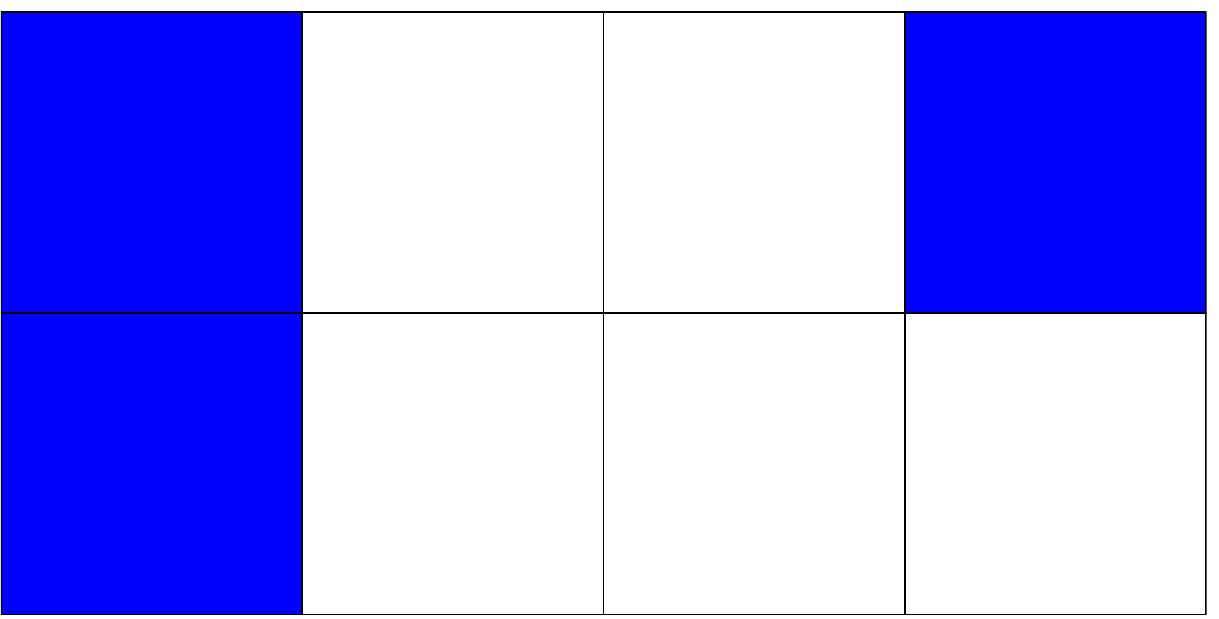} & 2x9 & \includegraphics[width=0.160000\linewidth]{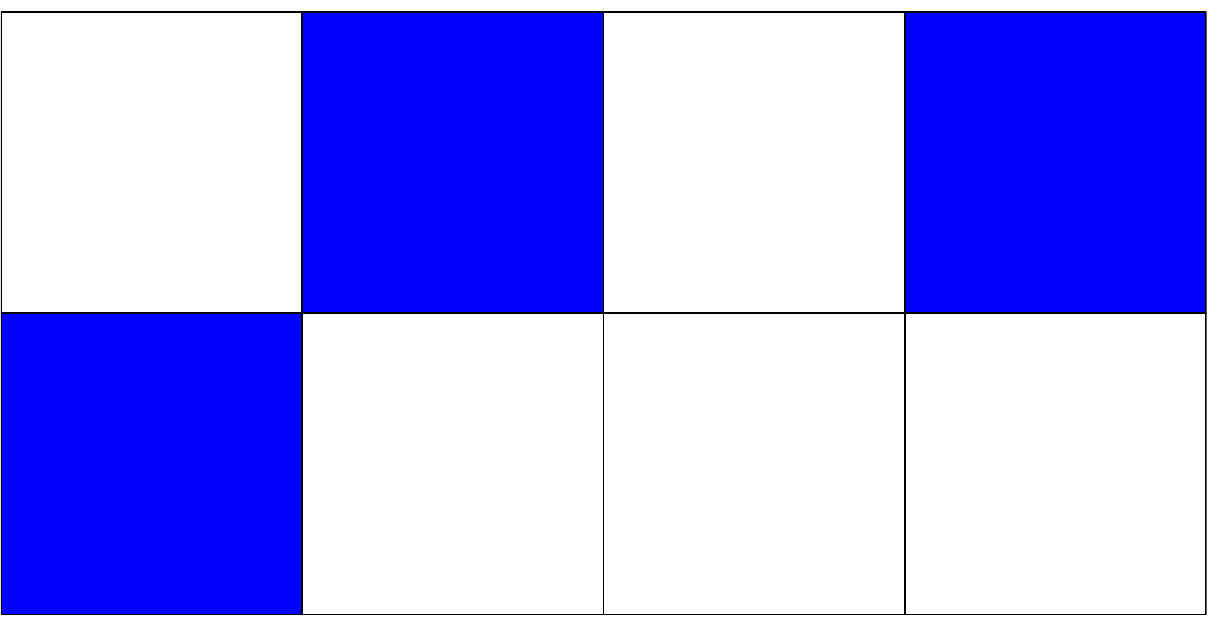} & 2x6 \\ \hline
& & & \\ 
\includegraphics[width=0.160000\linewidth]{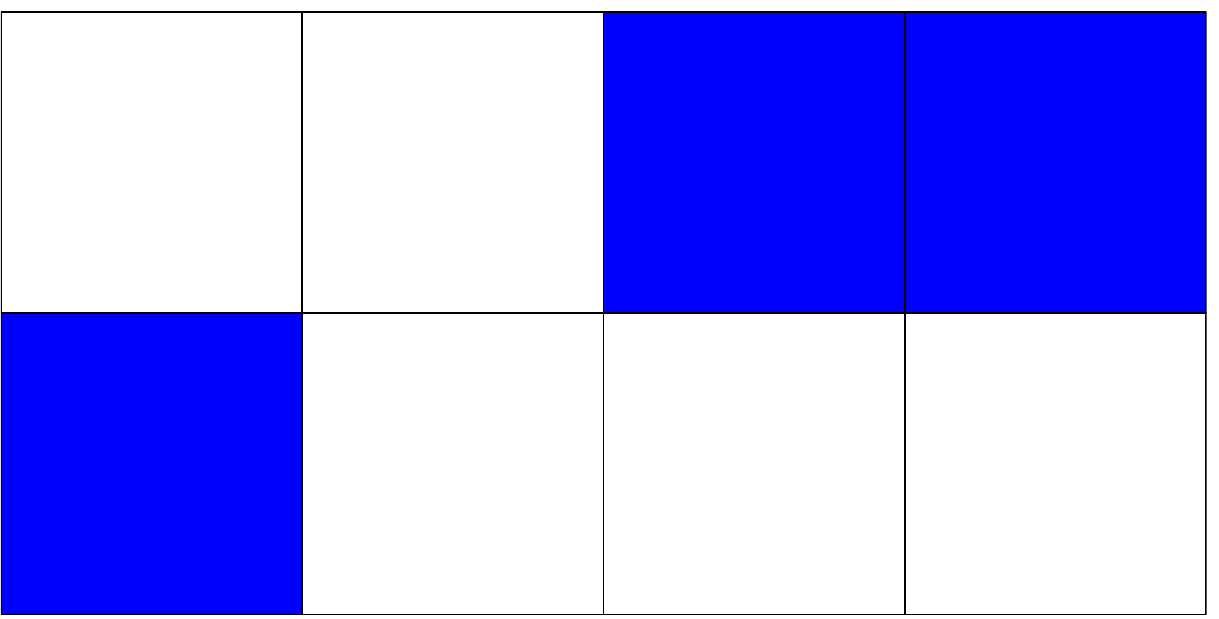} & 2x6 & \includegraphics[width=0.160000\linewidth]{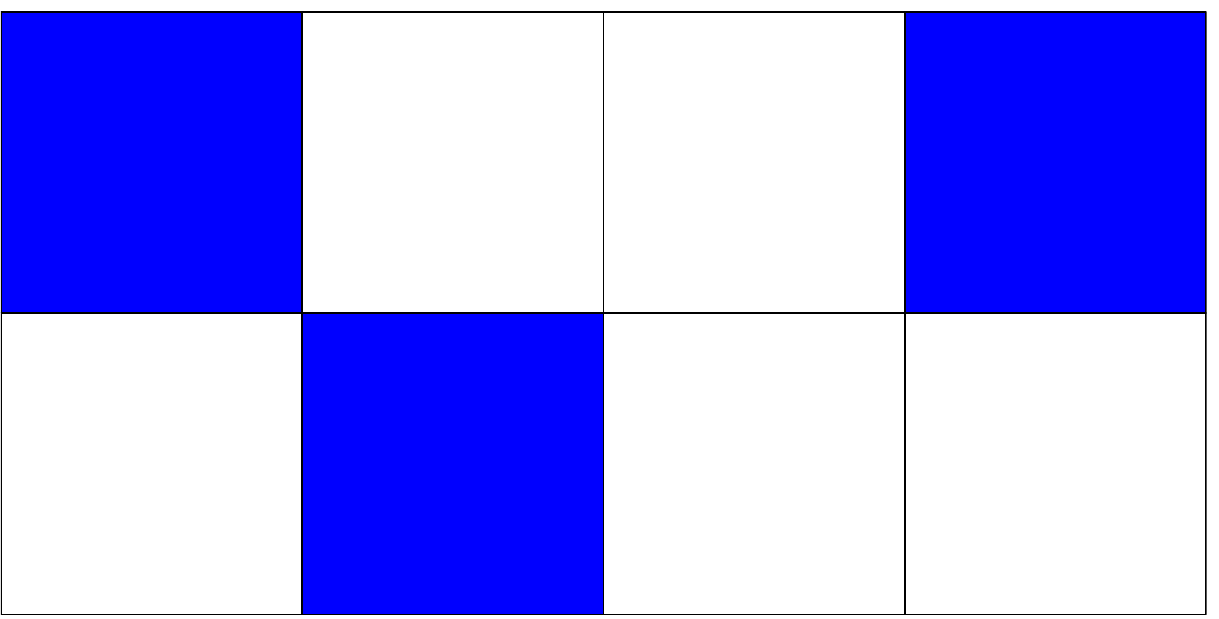} & 2x6 \\ \hline
& & & \\ 
\includegraphics[width=0.120000\linewidth]{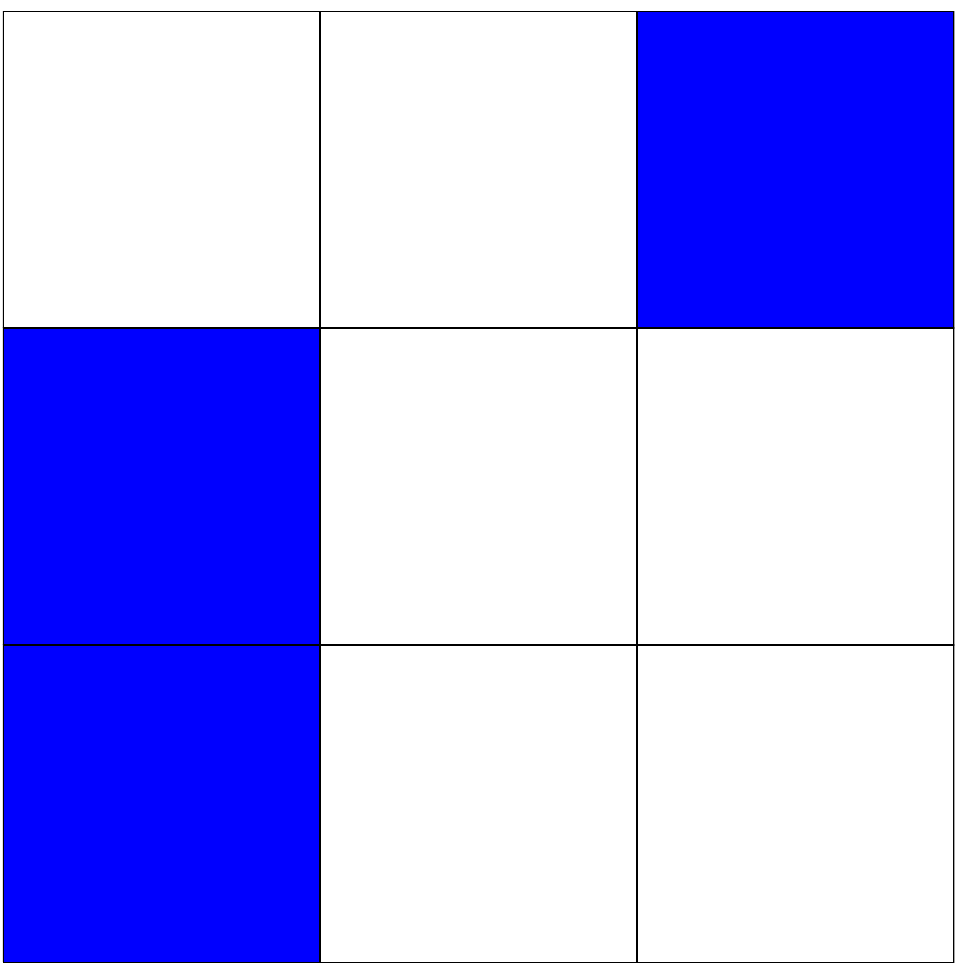} & 3x4 & \includegraphics[width=0.120000\linewidth]{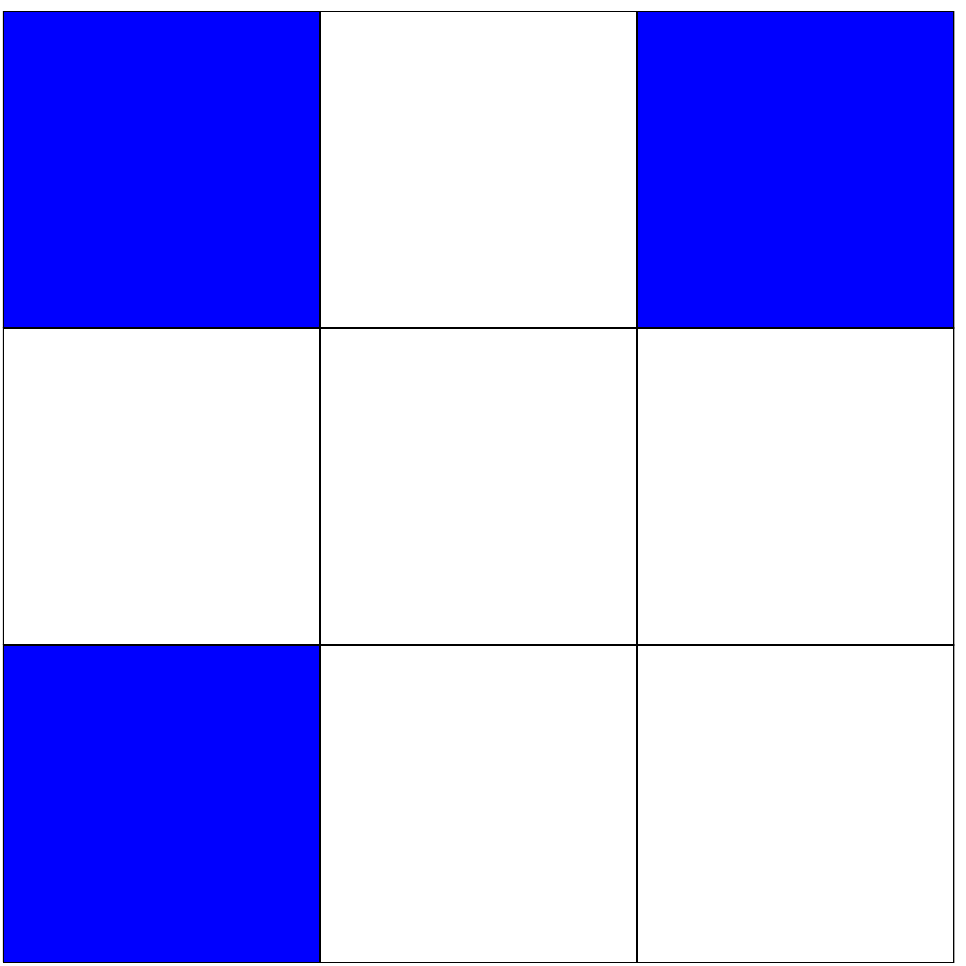} & 4x6 \\ \hline
& & & \\ 
\includegraphics[width=0.120000\linewidth]{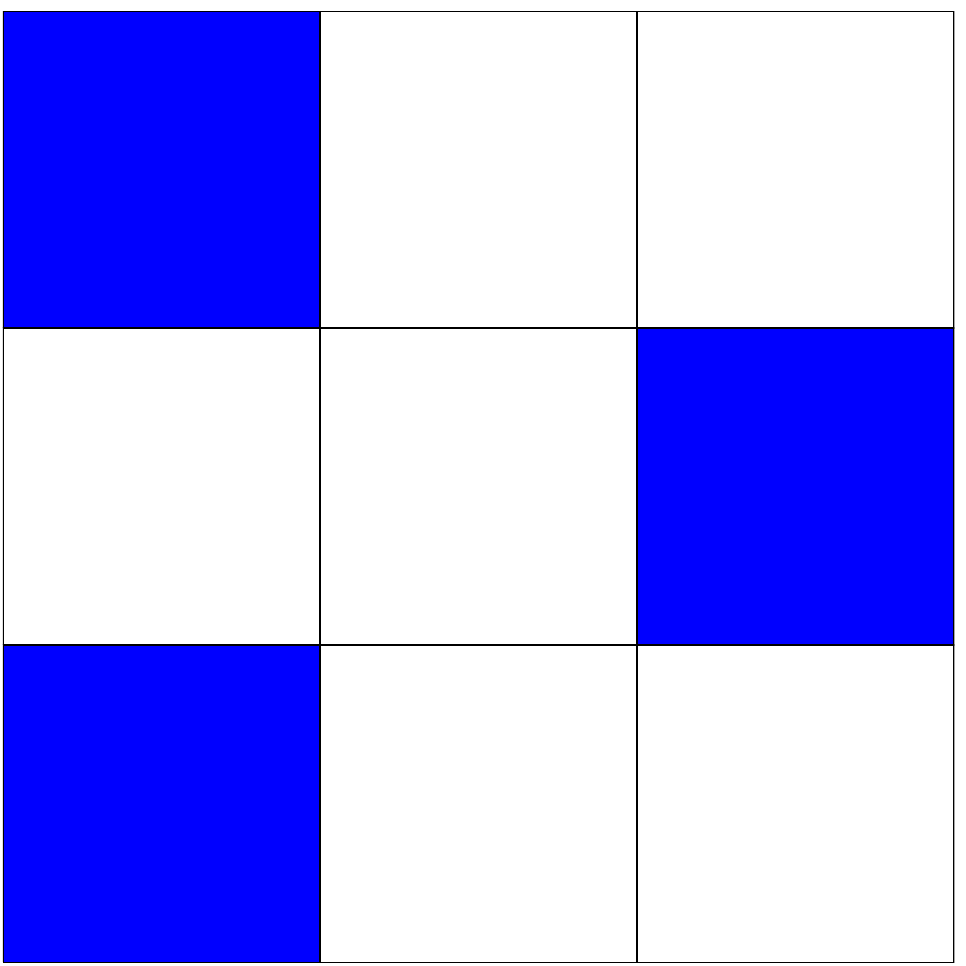} & 3x4  & & \\ \hline
\end{tabular}
\caption{Smallest known solutions for n=3 and k=2.}
\label{tab:trivial-32}
\end{table}

\begin{table}[!htpb]
\centering
\begin{tabular}{|c|c|c|c|}
\hline
Piece & Smallest Tiling & Piece & Smallest Tiling\\ \hline
& & & \\ 
\includegraphics[width=0.200000\linewidth]{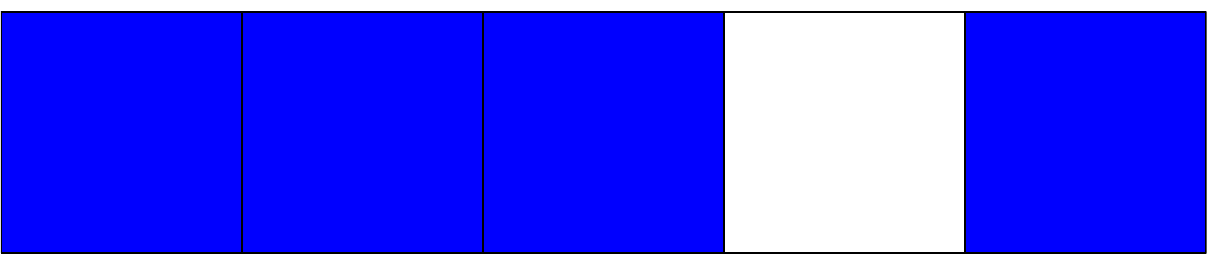} & 1x8 & \includegraphics[width=0.160000\linewidth]{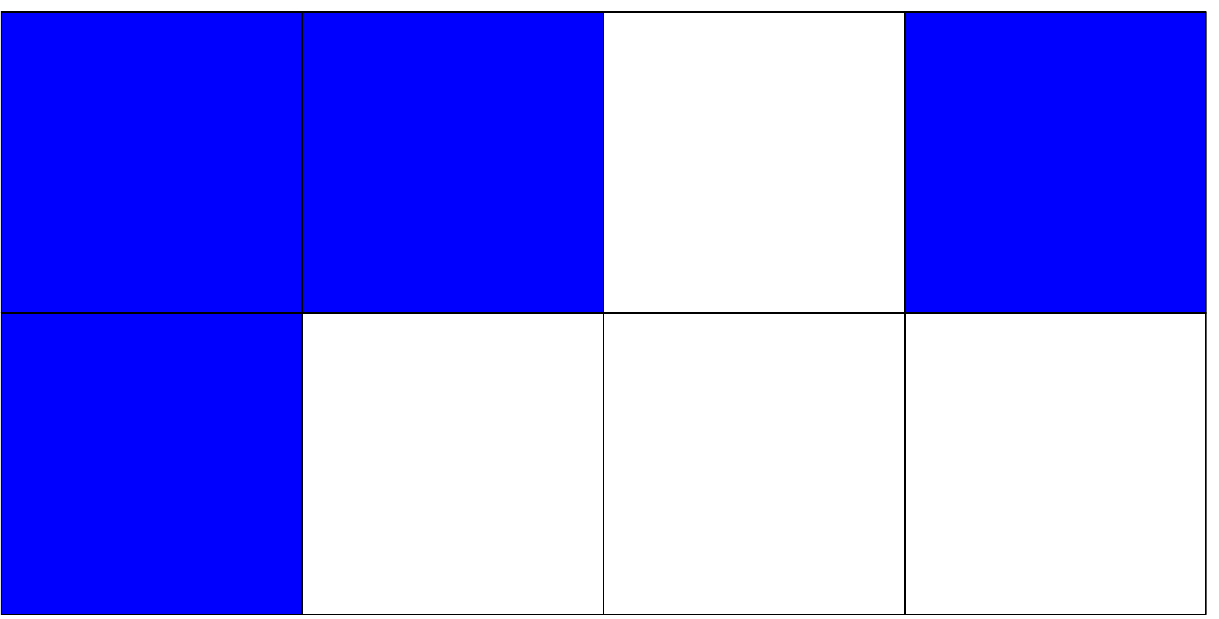} & 4x4 \\ \hline
& & & \\ 
\includegraphics[width=0.160000\linewidth]{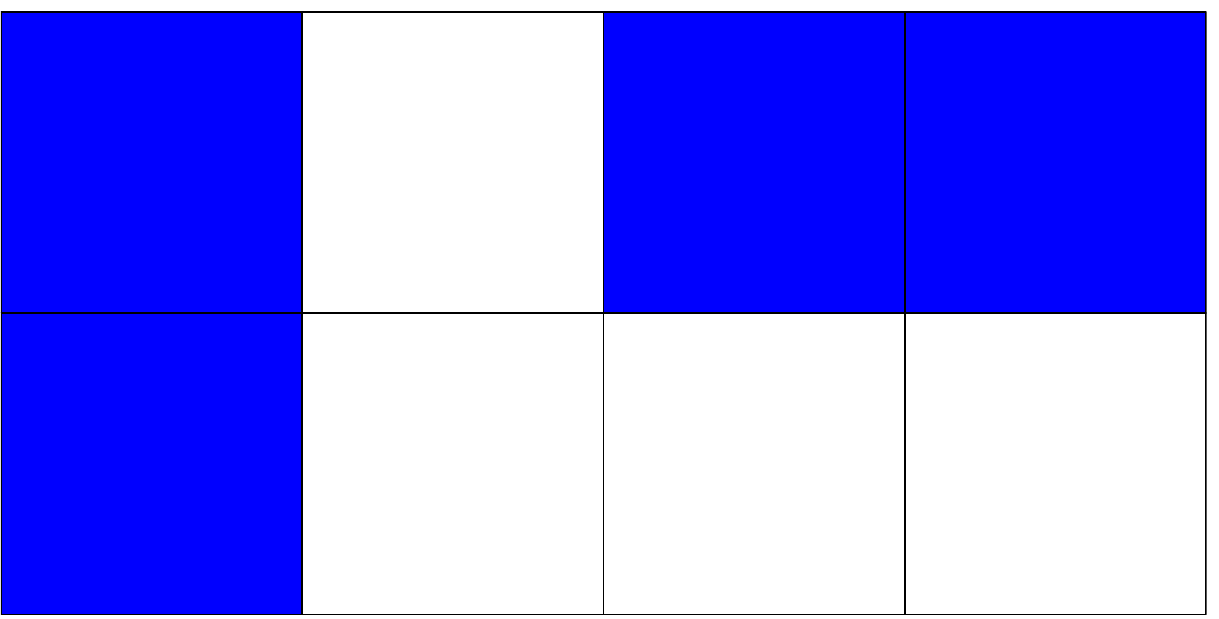} & 4x4 & \includegraphics[width=0.160000\linewidth]{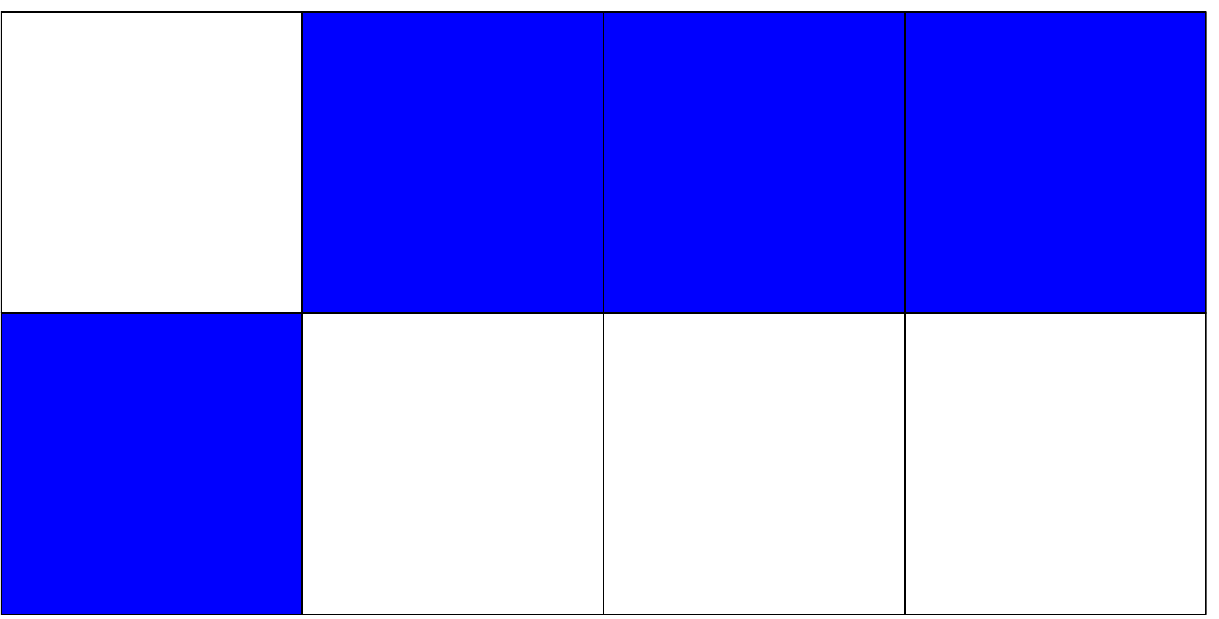} & 2x4 \\ \hline
& & & \\ 
\includegraphics[width=0.160000\linewidth]{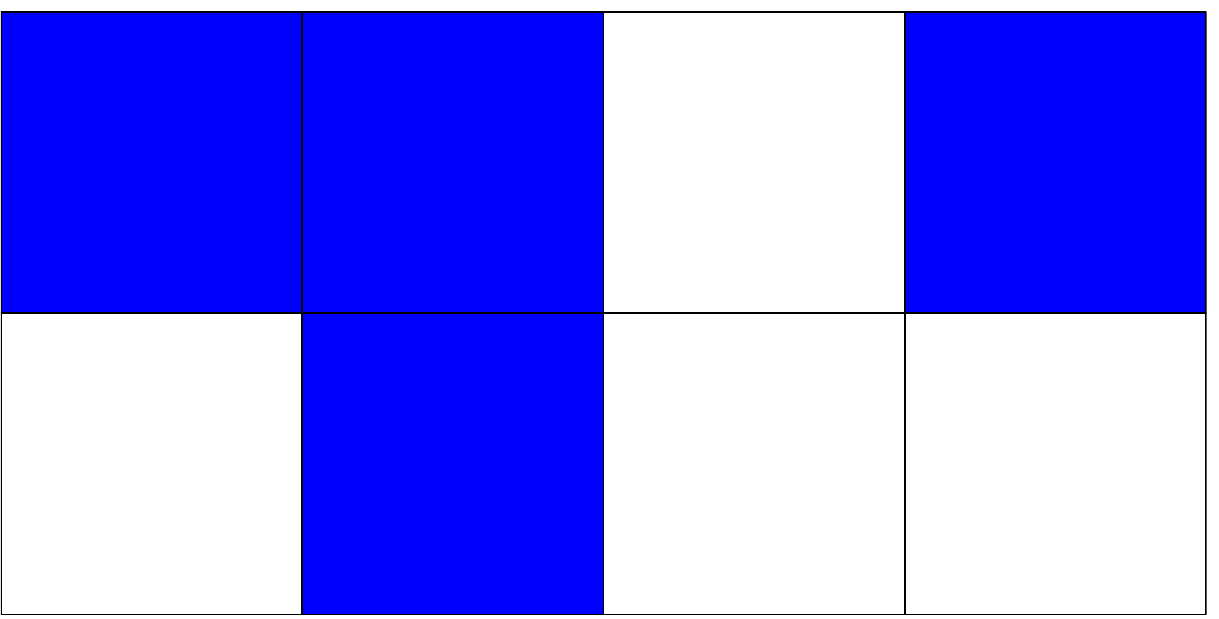} & 2x4 & \includegraphics[width=0.160000\linewidth]{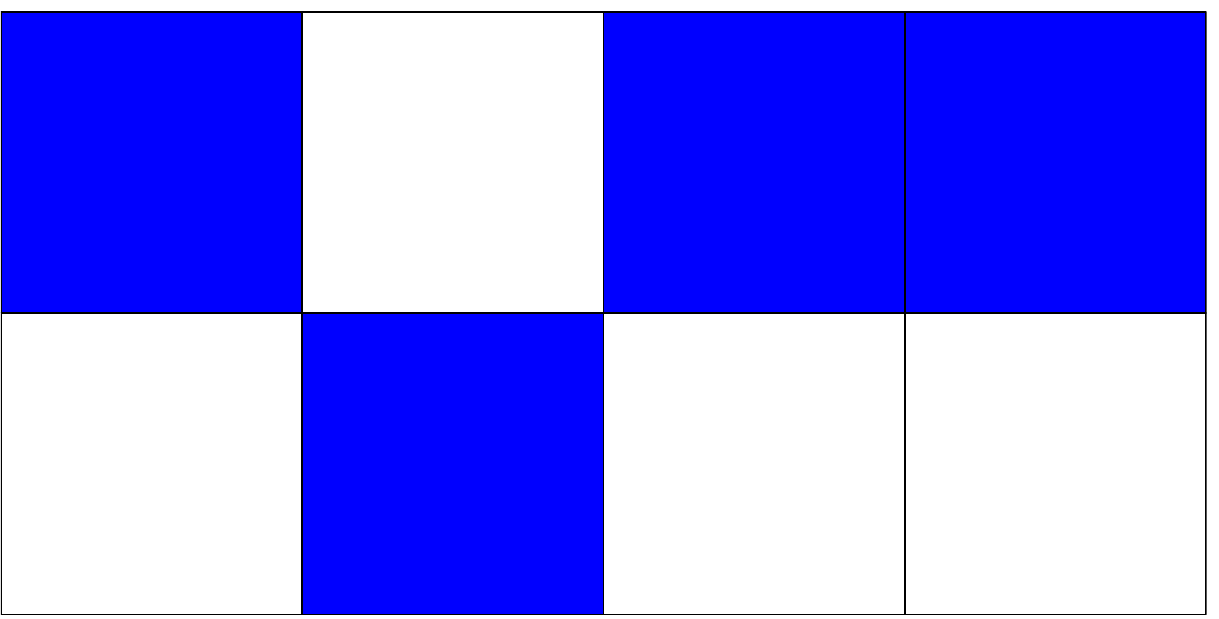} & 2x4 \\ \hline
& & & \\ 
\includegraphics[width=0.160000\linewidth]{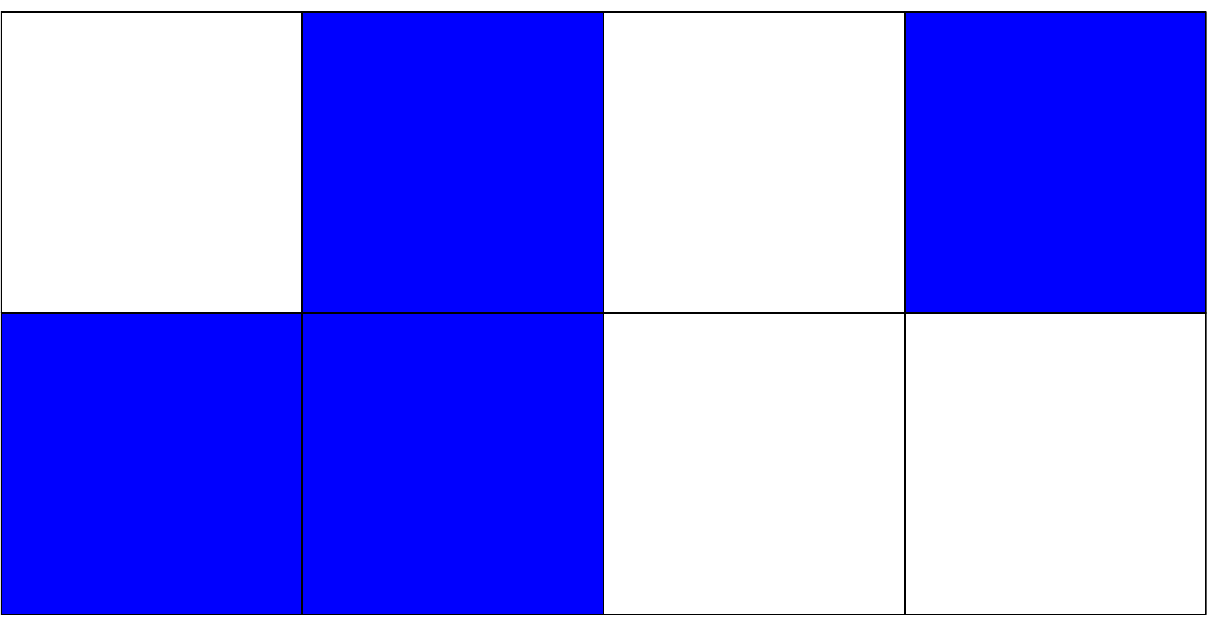} & 2x4 & \includegraphics[width=0.160000\linewidth]{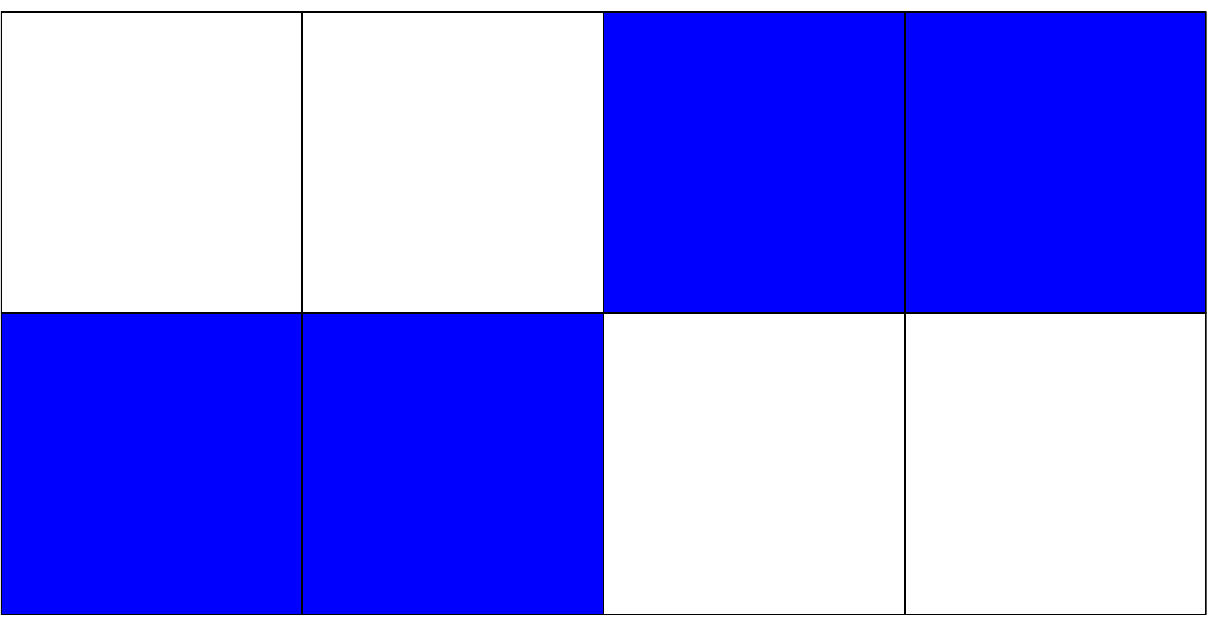} & 2x4 \\ \hline
& & & \\ 
\includegraphics[width=0.120000\linewidth]{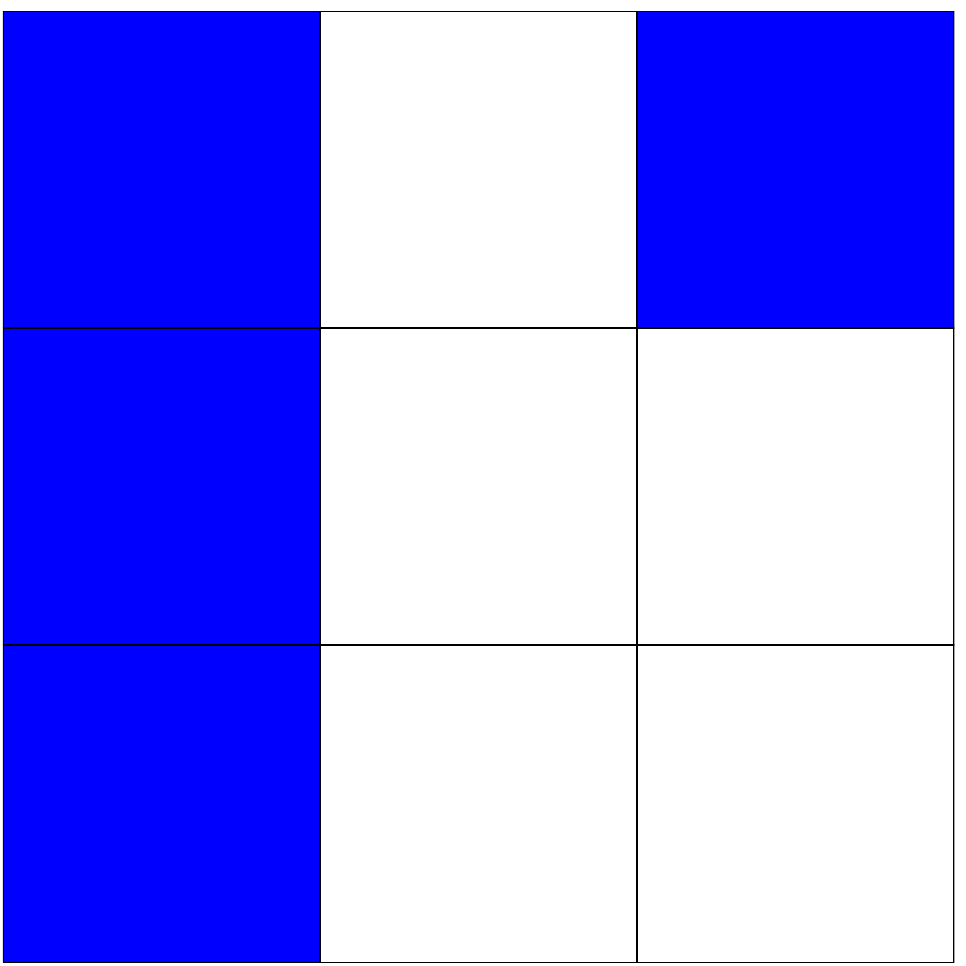} & 4x4 & \includegraphics[width=0.120000\linewidth]{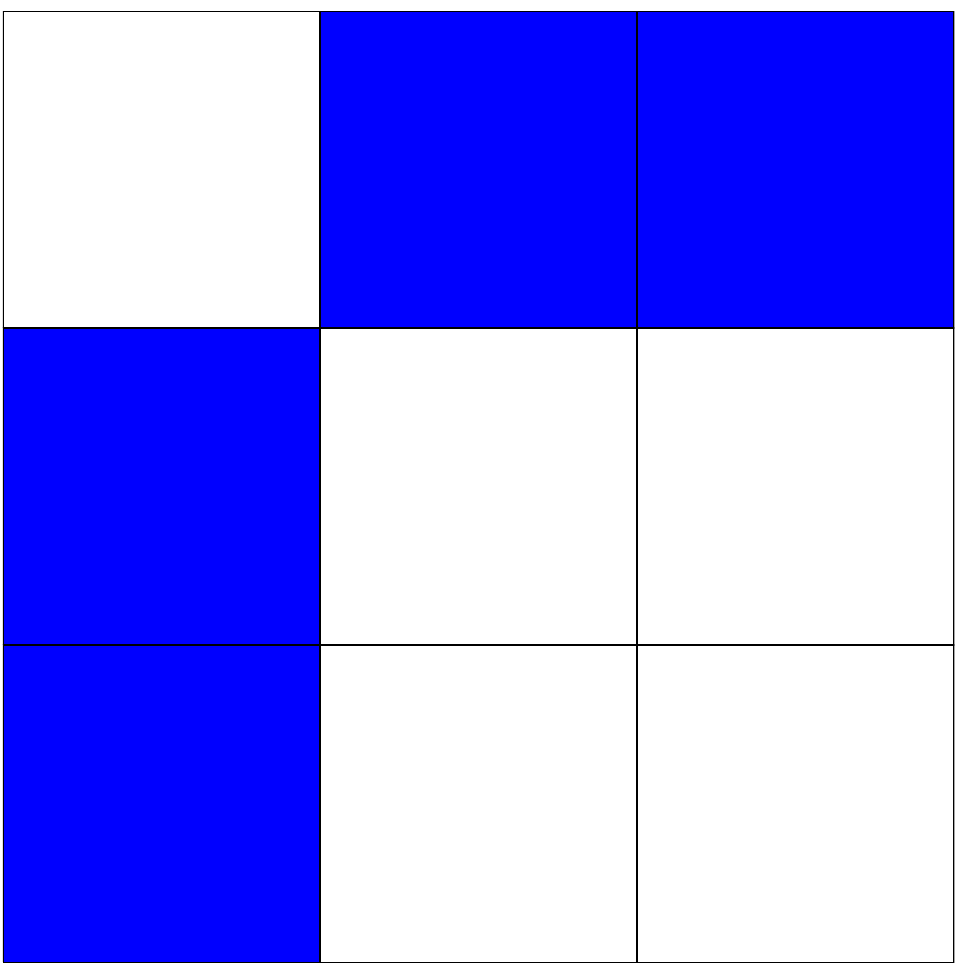} & 14x36 \\ \hline
& & & \\ 
\includegraphics[width=0.120000\linewidth]{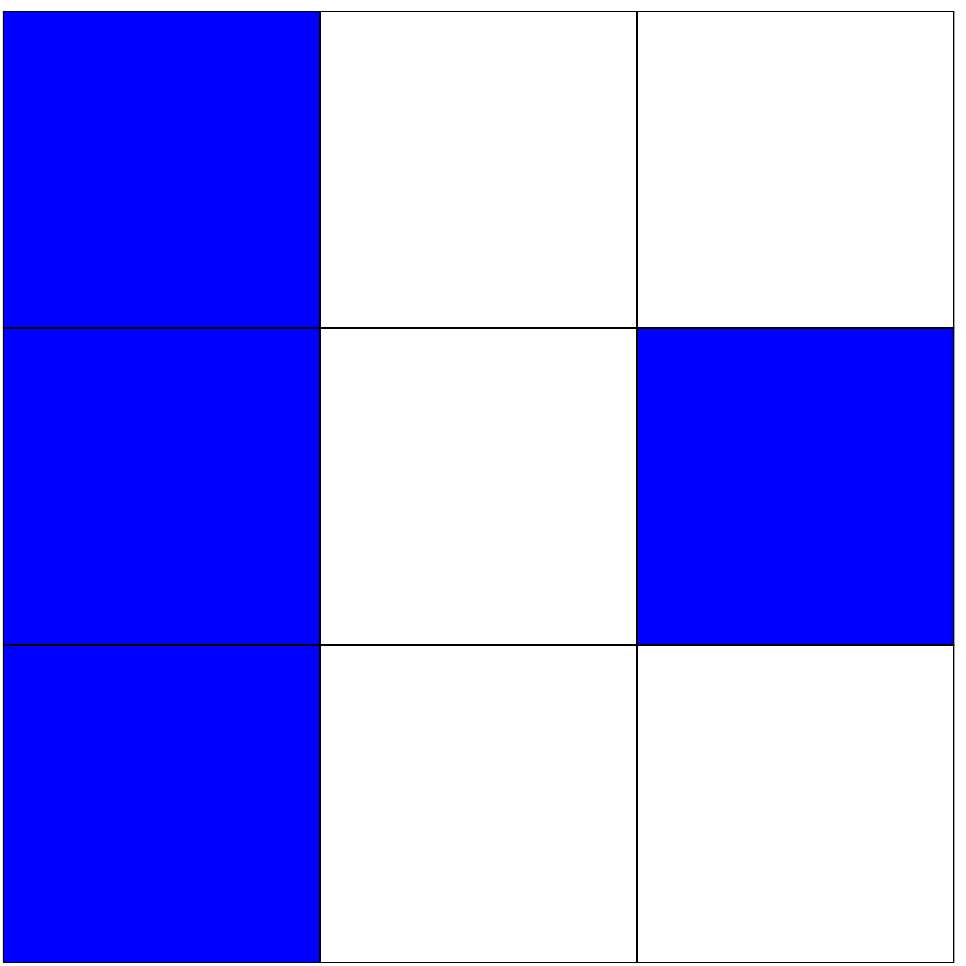} & 4x4 & \includegraphics[width=0.120000\linewidth]{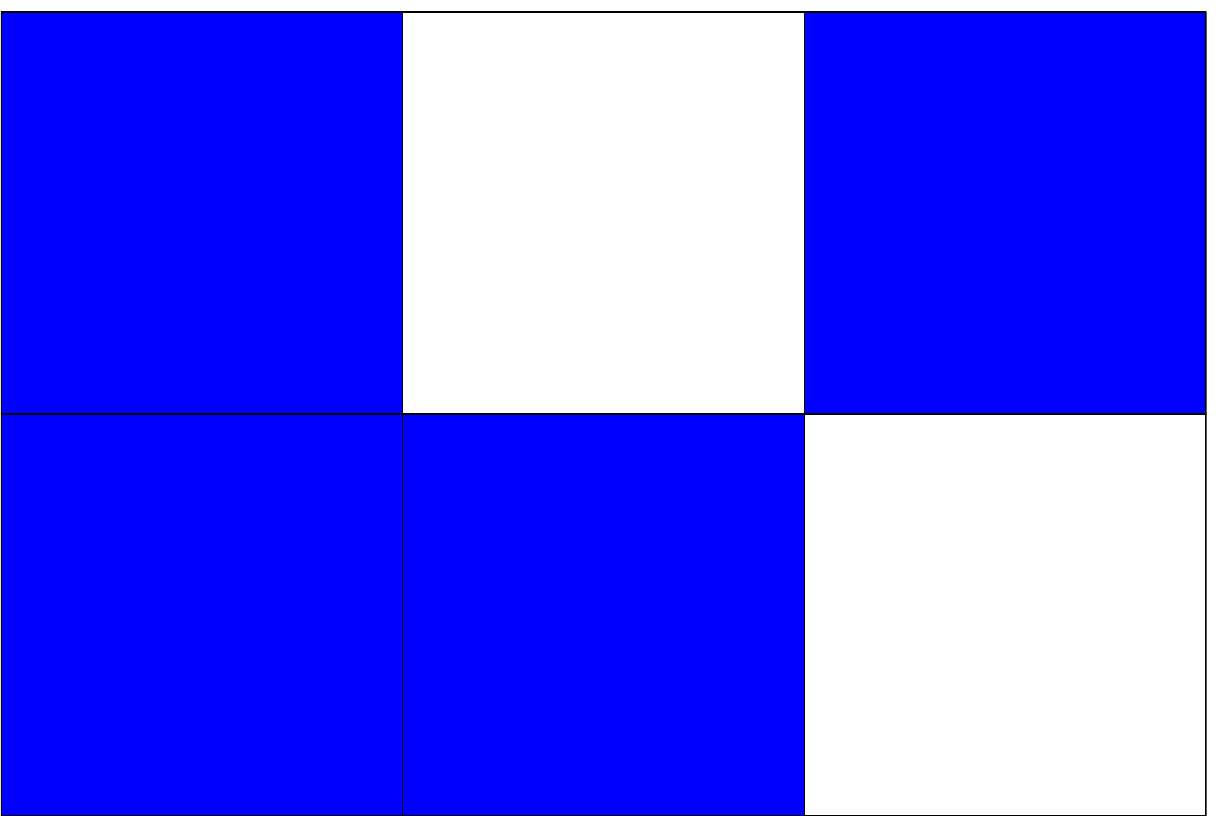} & 2x4 \\ \hline
& & & \\ 
\includegraphics[width=0.120000\linewidth]{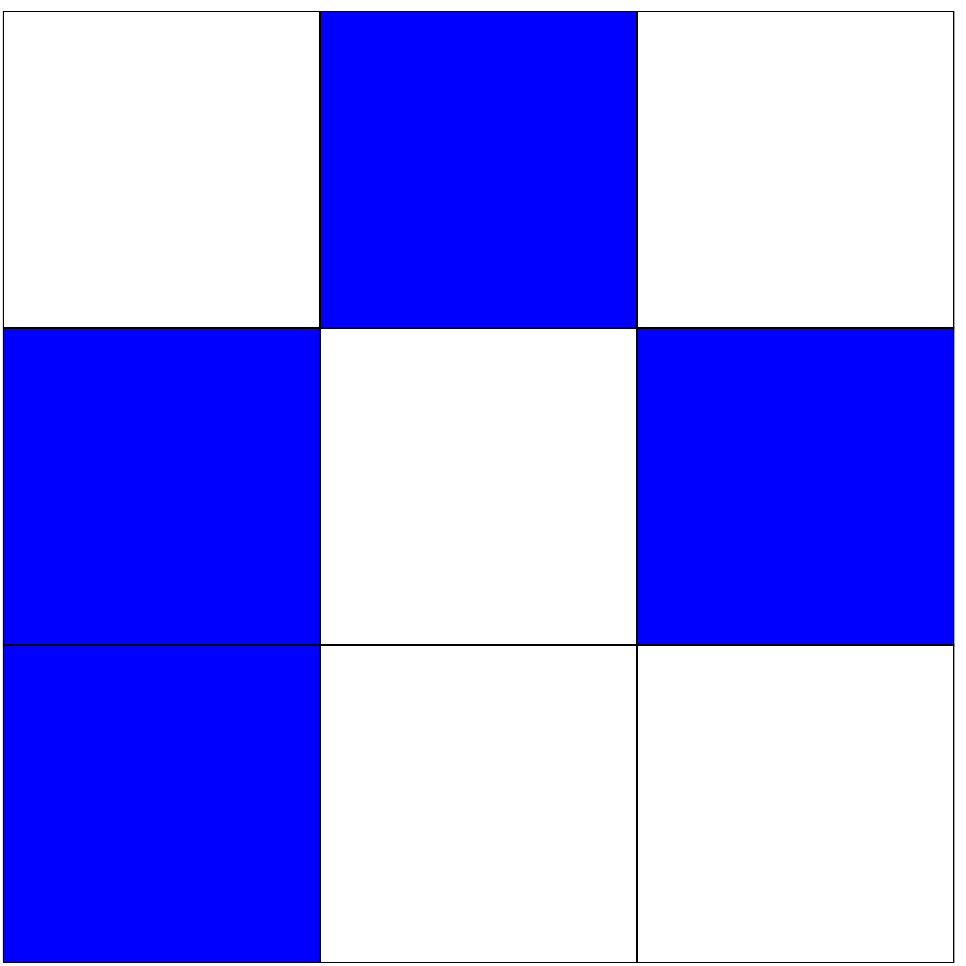} & 8x12 & \includegraphics[width=0.120000\linewidth]{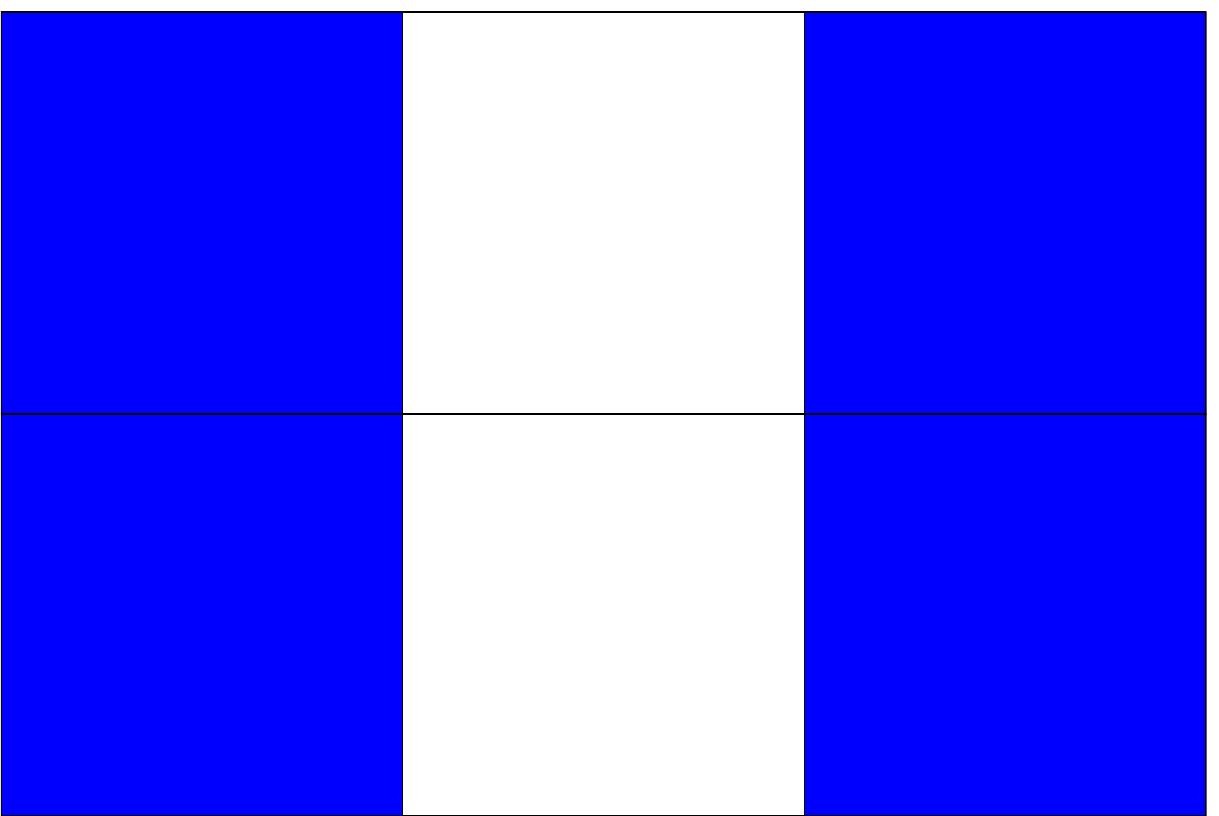} & 2x4 \\ \hline
\end{tabular}
\caption{Smallest known solutions for n=4 and k=1.}
\label{tab:trivial-41}
\end{table}

\begin{table}[!htpb]
\centering
\begin{tabular}{|c|c|c|c|}
\hline
Piece & Smallest Tiling & Piece & Smallest Tiling\\ \hline
& & & \\ 
\includegraphics[width=0.240000\linewidth]{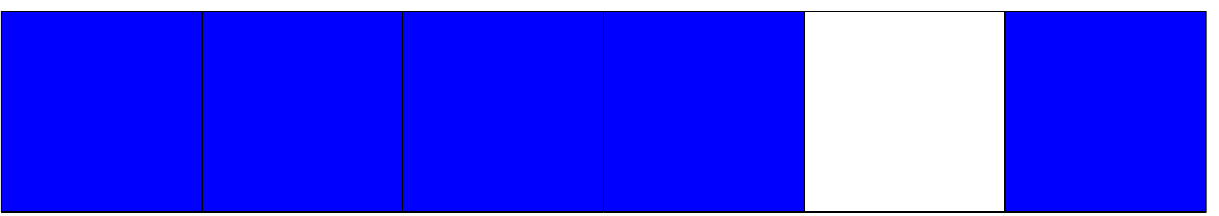} & 1x10 & \includegraphics[width=0.240000\linewidth]{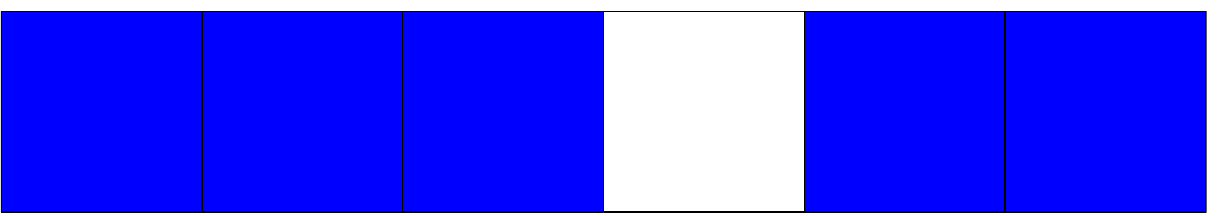} & 6x20 \\ \hline
& & & \\ 
\includegraphics[width=0.200000\linewidth]{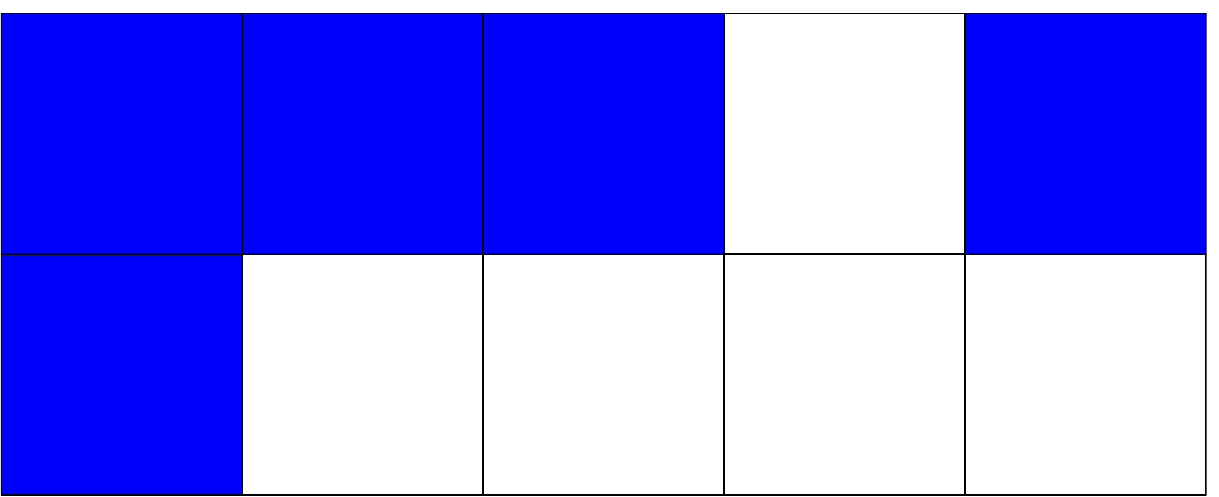} & 10x10 & \includegraphics[width=0.200000\linewidth]{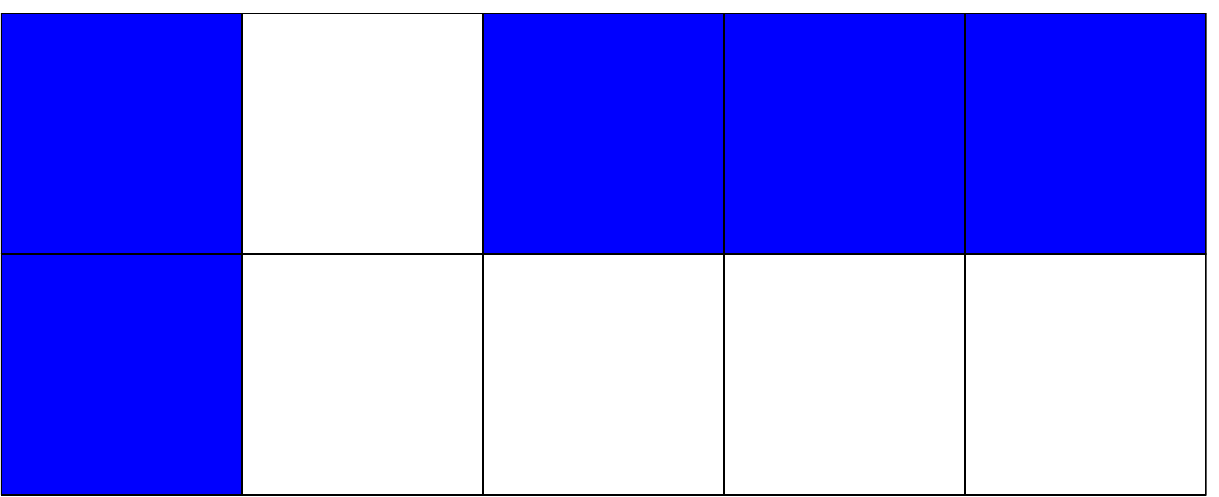} & 10x10 \\ \hline
& & & \\ 
\includegraphics[width=0.200000\linewidth]{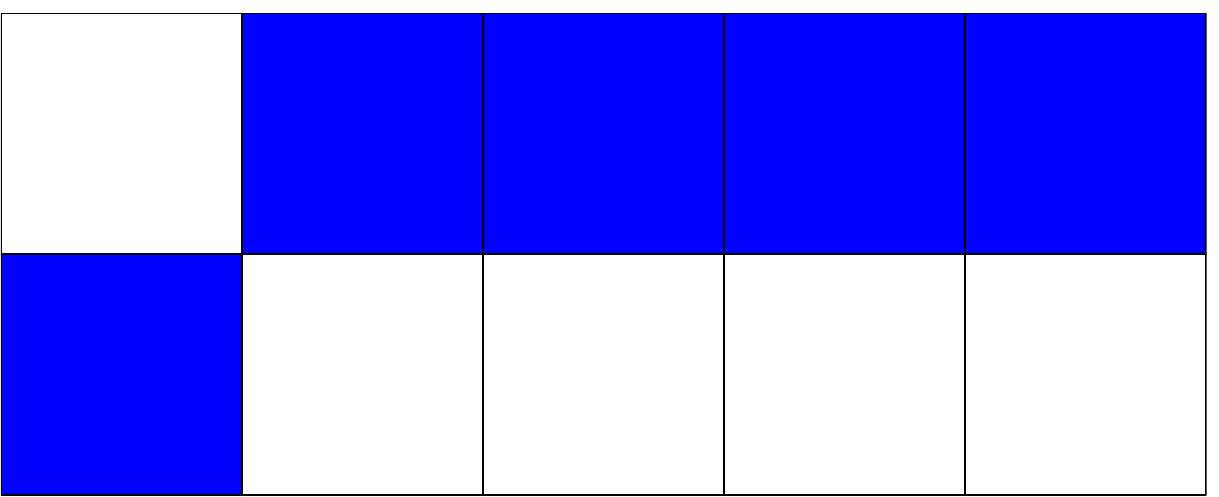} & 2x5 & \includegraphics[width=0.200000\linewidth]{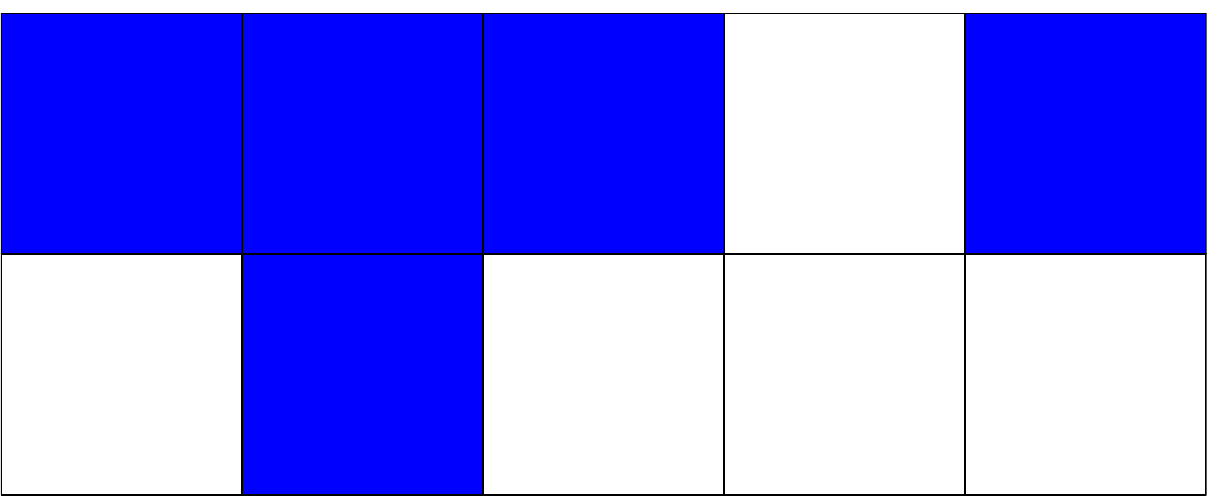} & 2x5 \\ \hline
& & & \\ 
\includegraphics[width=0.200000\linewidth]{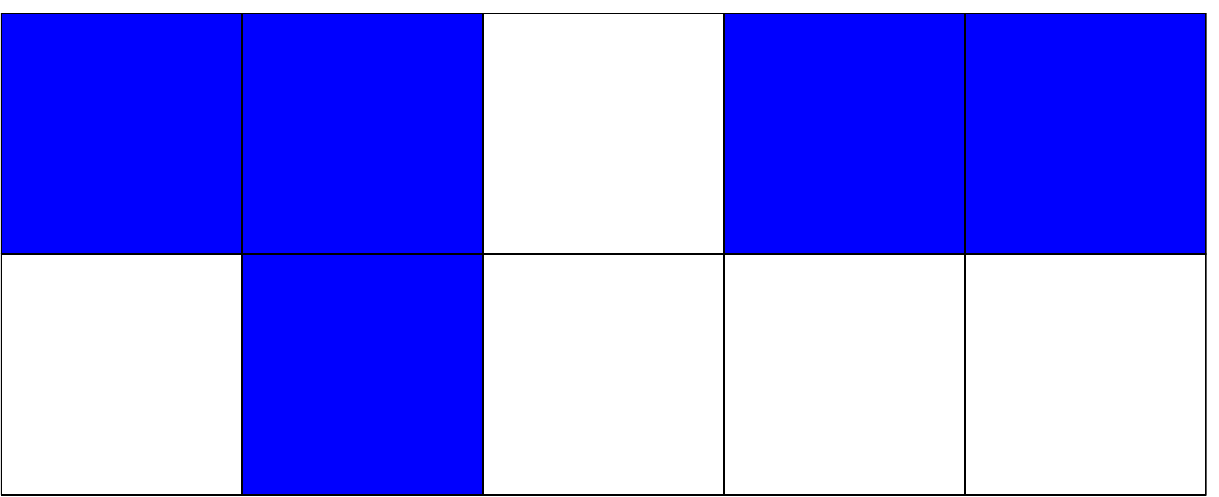} & 12x15 & \includegraphics[width=0.200000\linewidth]{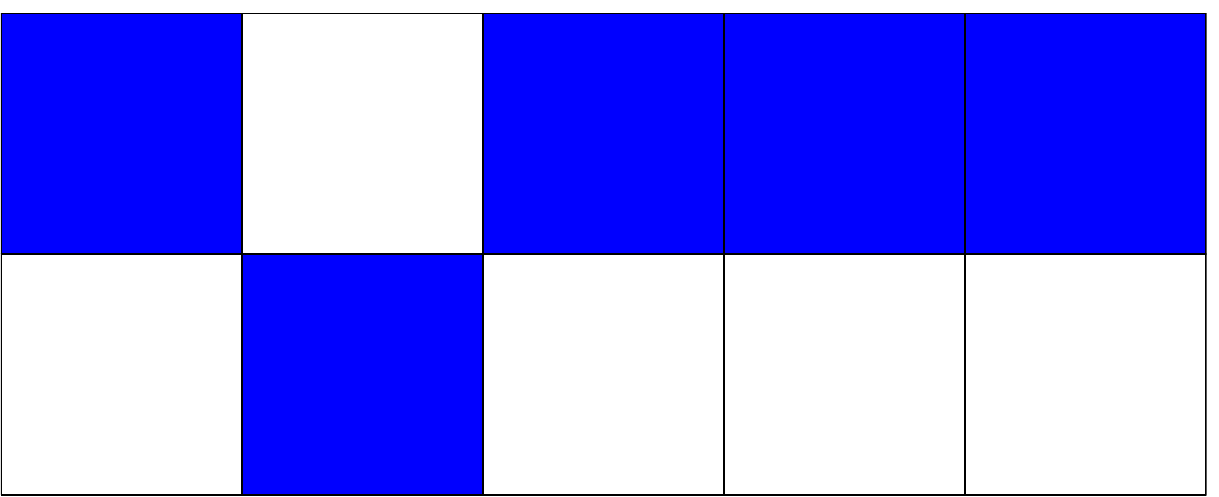} & 2x5 \\ \hline
& & & \\ 
\includegraphics[width=0.200000\linewidth]{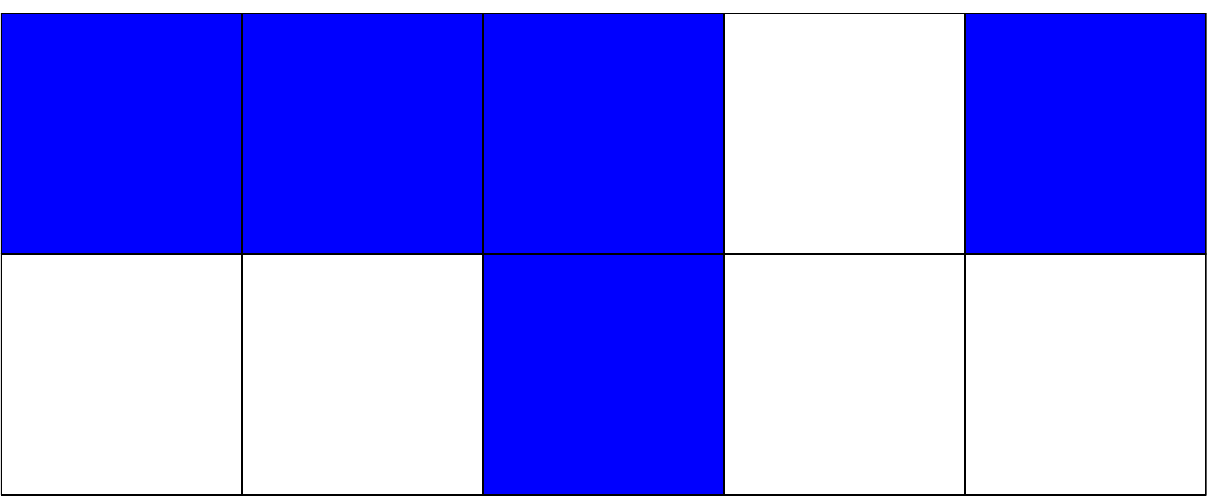} & 15x18 & \includegraphics[width=0.200000\linewidth]{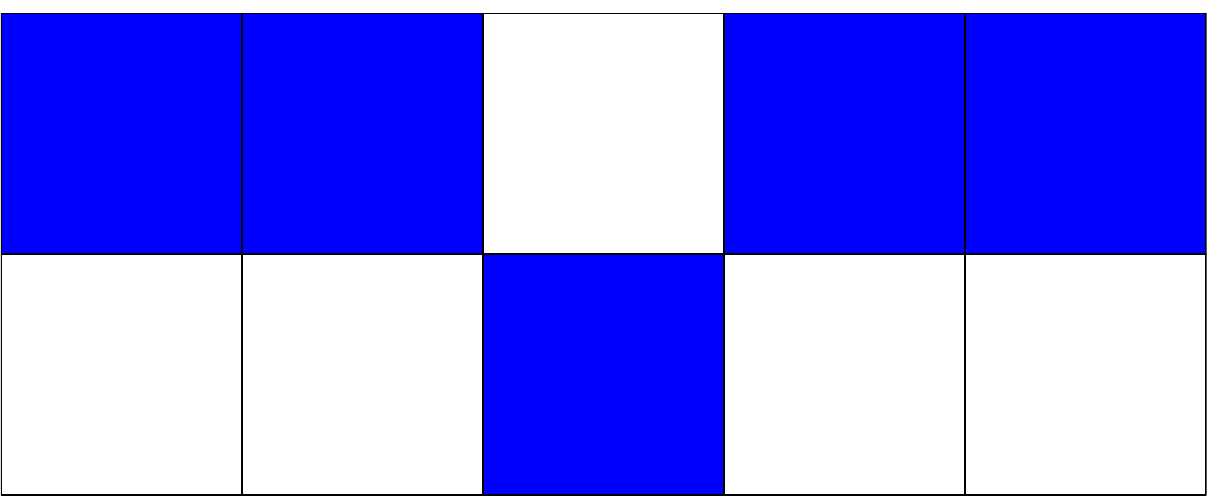} & 2x5 \\ \hline
& & & \\ 
\includegraphics[width=0.200000\linewidth]{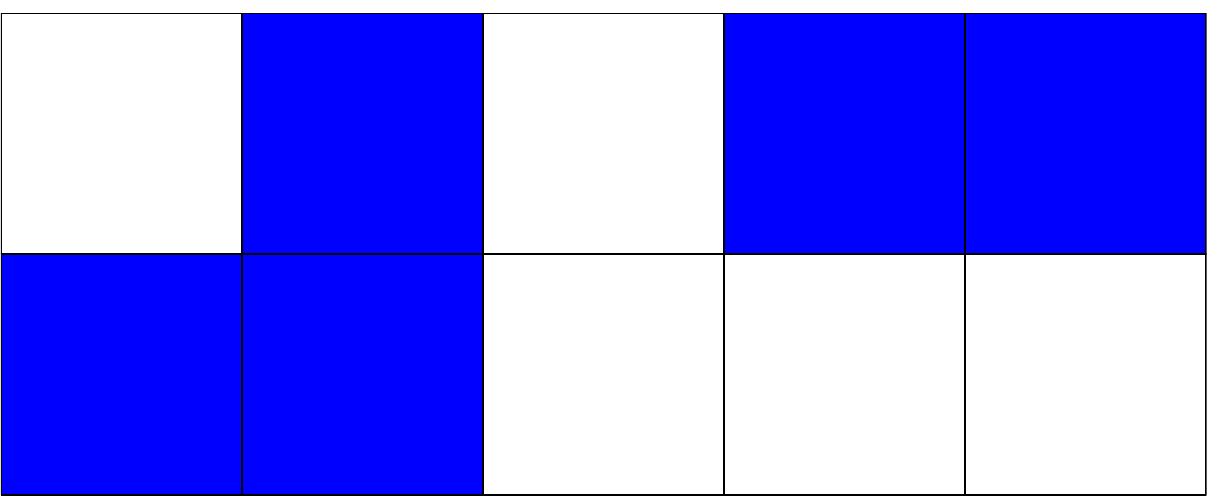} & 20x22 & \includegraphics[width=0.200000\linewidth]{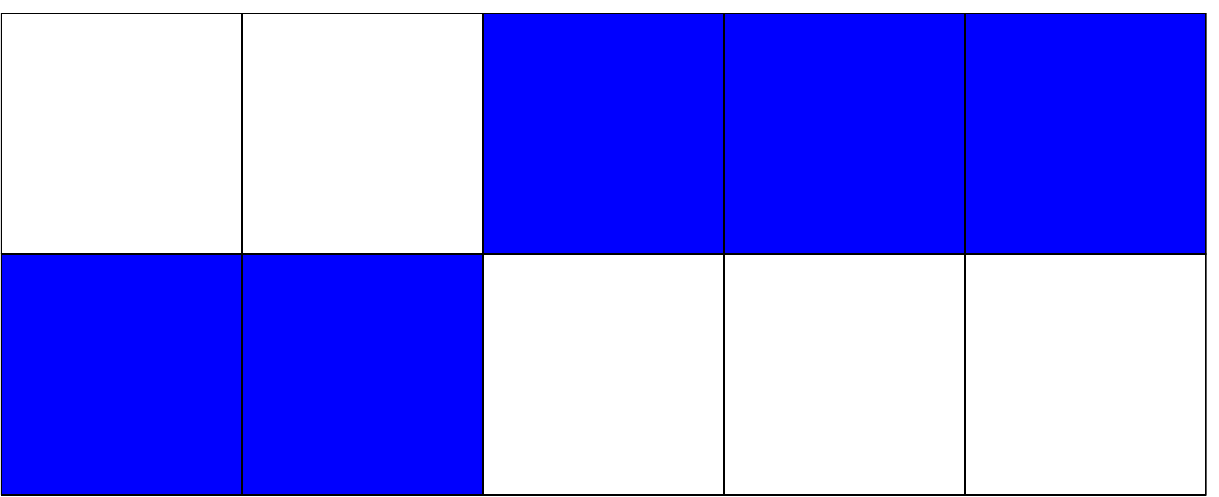} & 2x5 \\ \hline
& & & \\ 
\includegraphics[width=0.160000\linewidth]{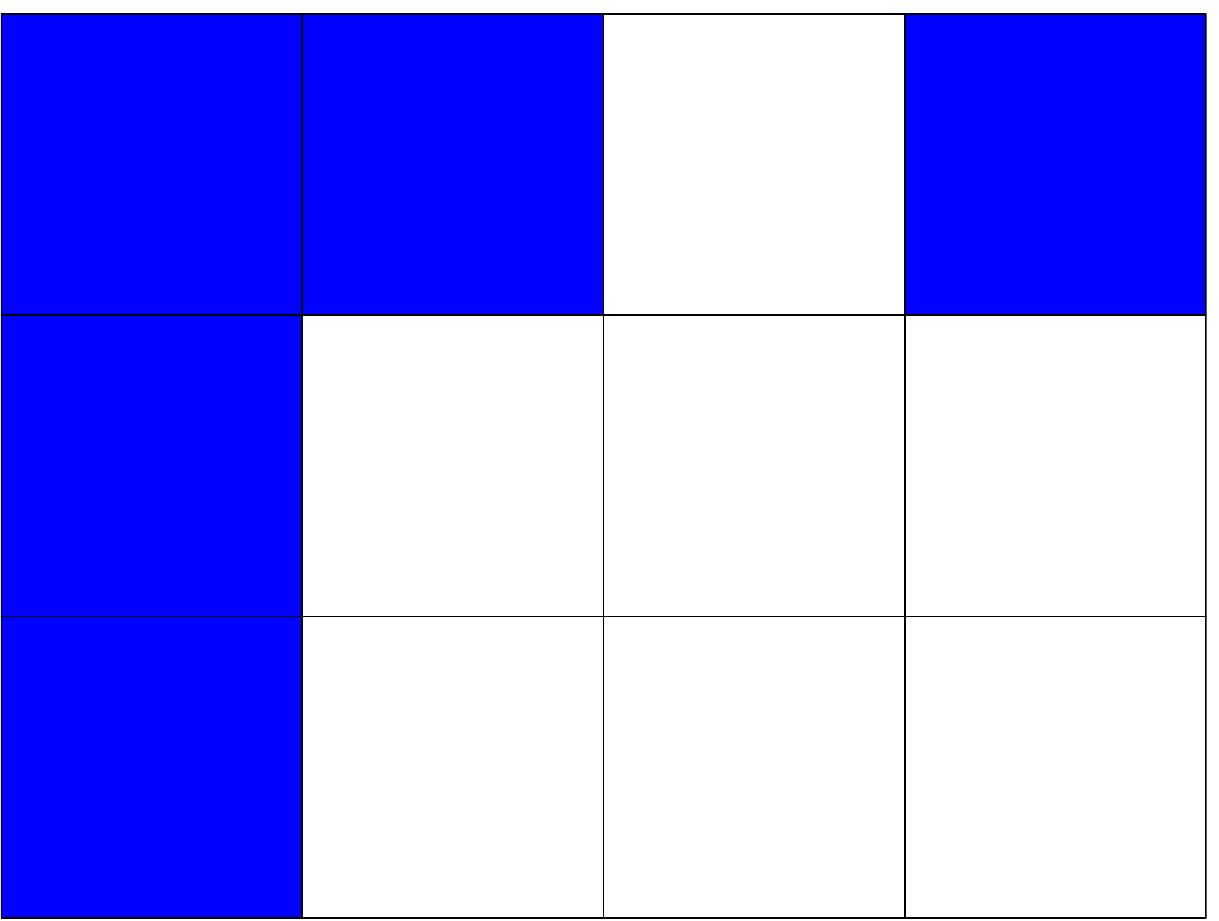} & 4x5 & \includegraphics[width=0.160000\linewidth]{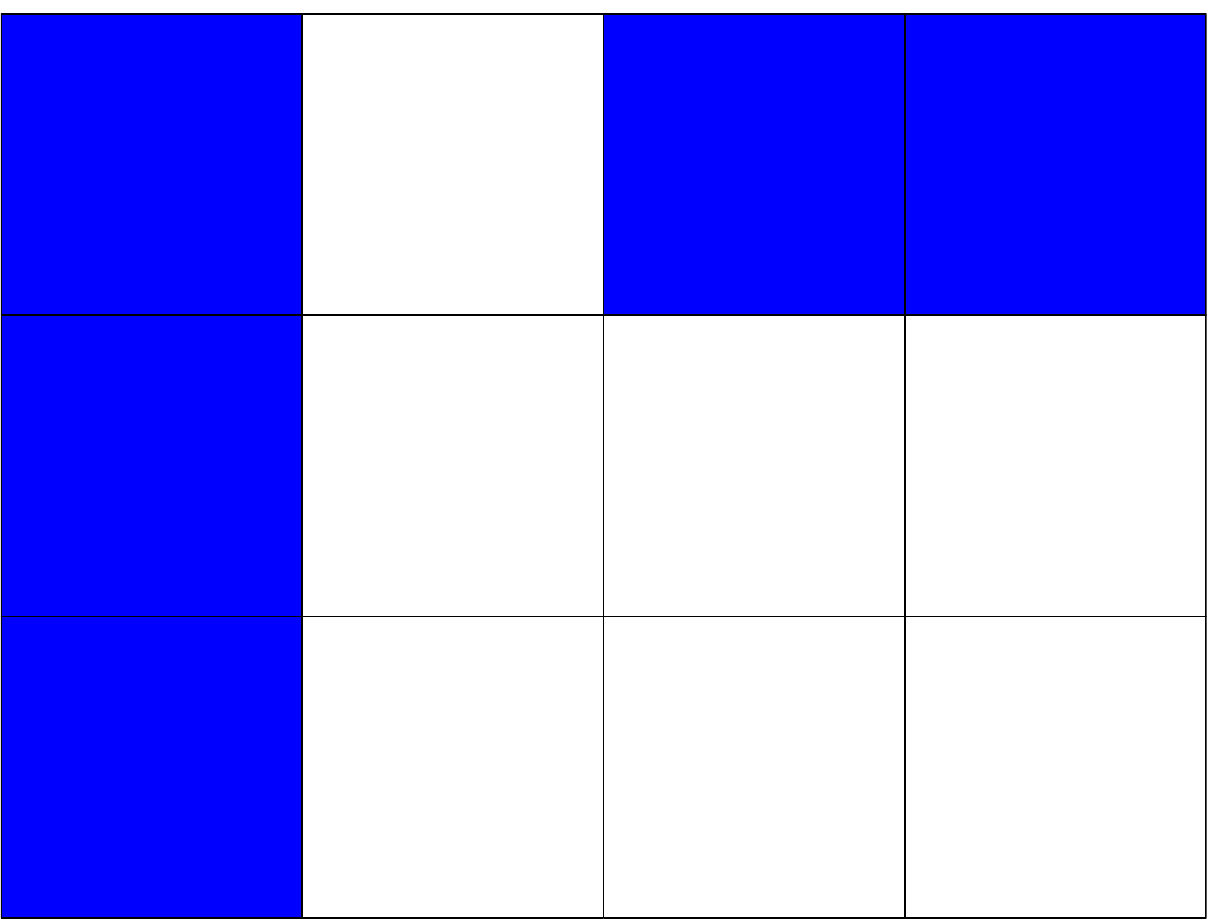} & 18x20 \\ \hline
\end{tabular}
\caption{Smallest known solutions for n=5 and k=1.}
\label{tab:trivial-51a}
\end{table}

\begin{table}[!htpb]
\centering
\begin{tabular}{|c|c|c|c|}
\hline
Piece & Smallest Tiling & Piece & Smallest Tiling\\ \hline
& & & \\ 
\includegraphics[width=0.160000\linewidth]{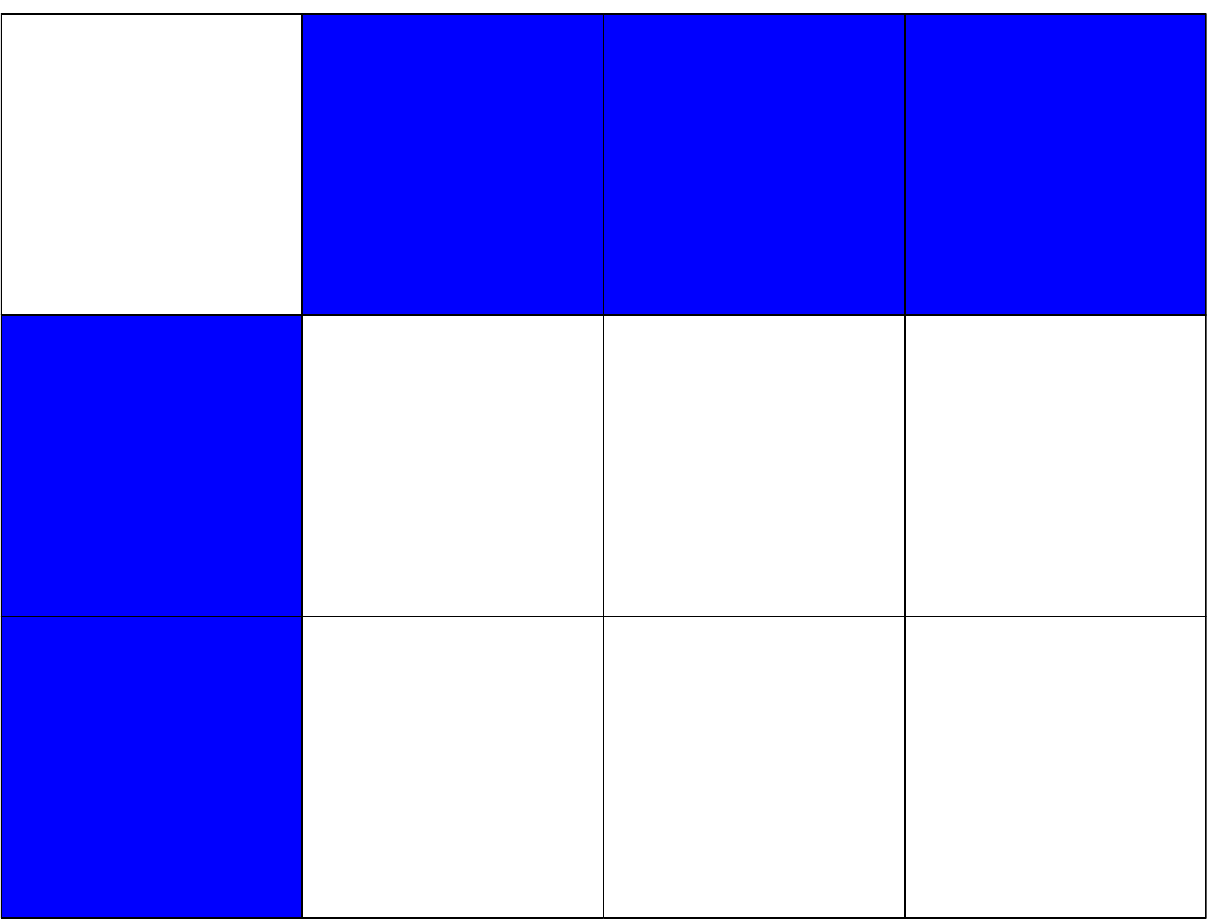} & 20x20 & \includegraphics[width=0.160000\linewidth]{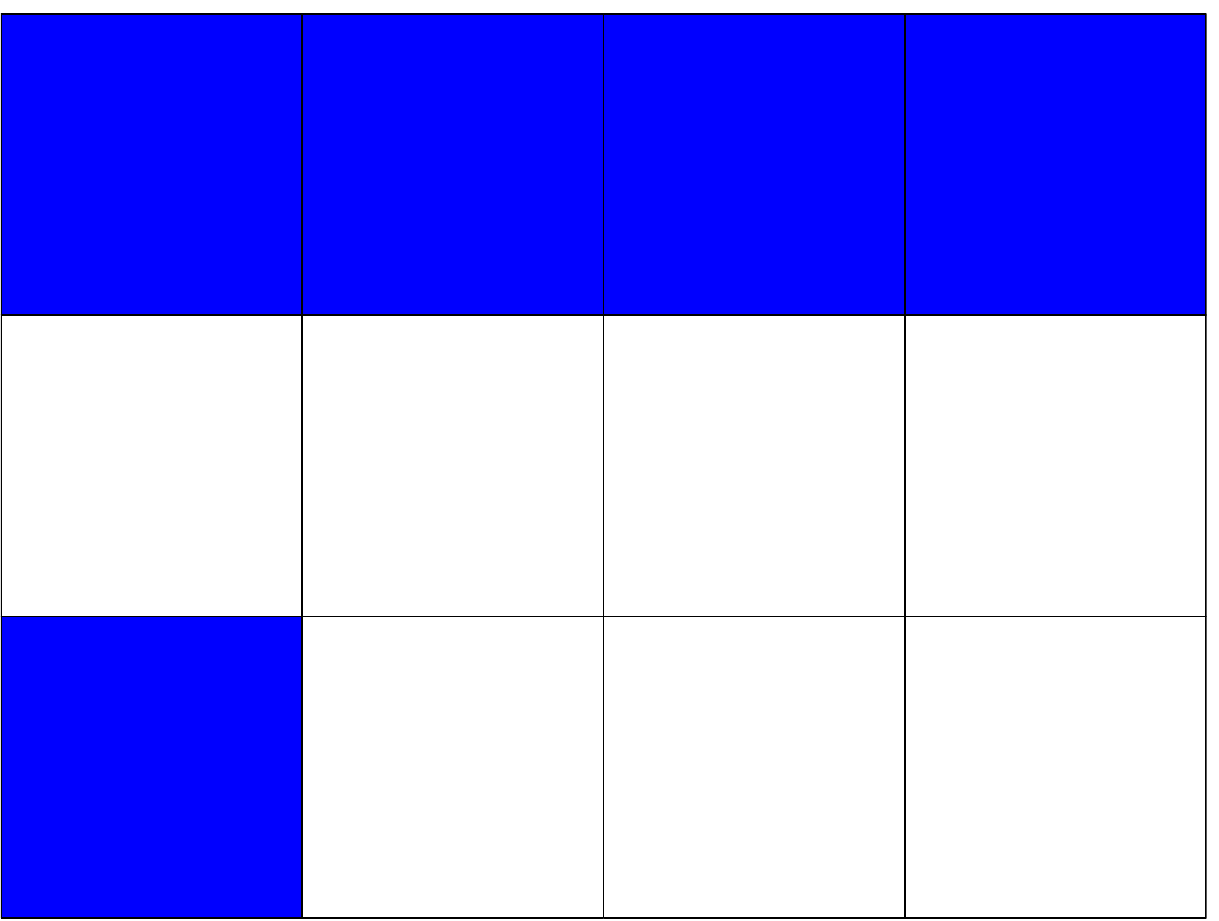} & 4x5 \\ \hline
& & & \\ 
\includegraphics[width=0.160000\linewidth]{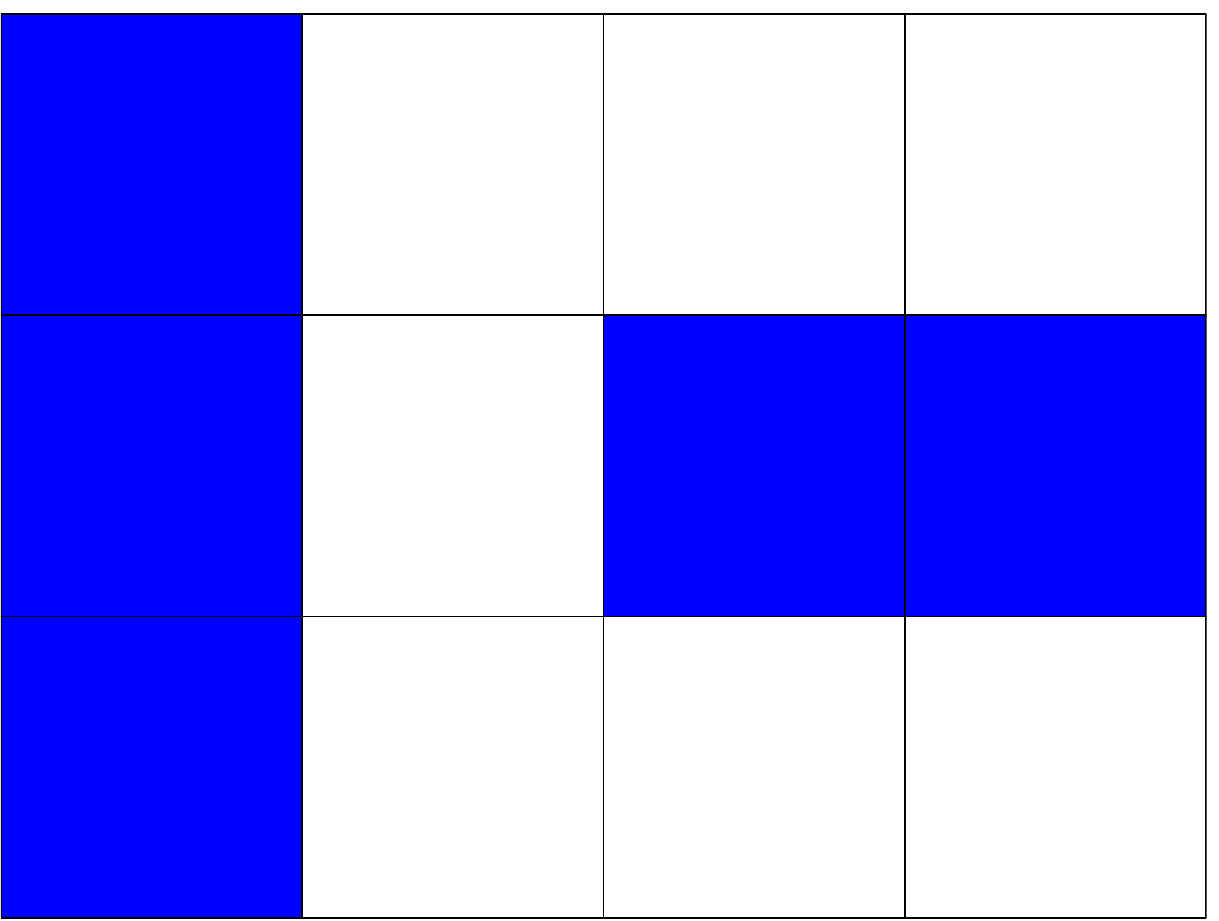} & 4x5 & \includegraphics[width=0.160000\linewidth]{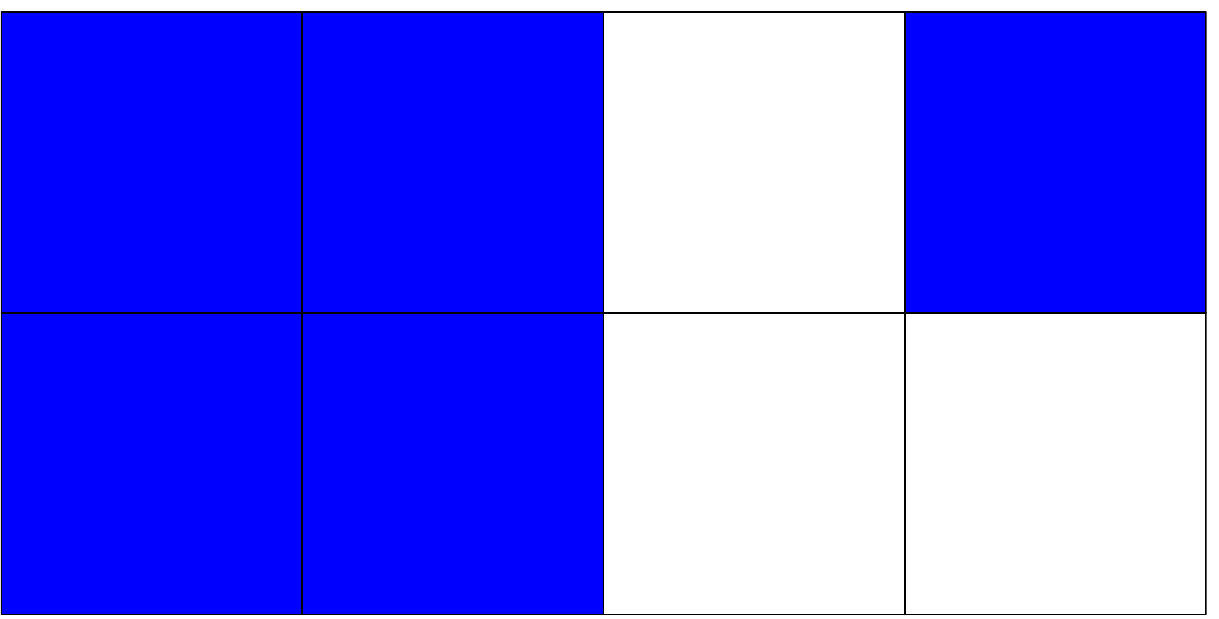} & 6x10 \\ \hline
& & & \\ 
\includegraphics[width=0.160000\linewidth]{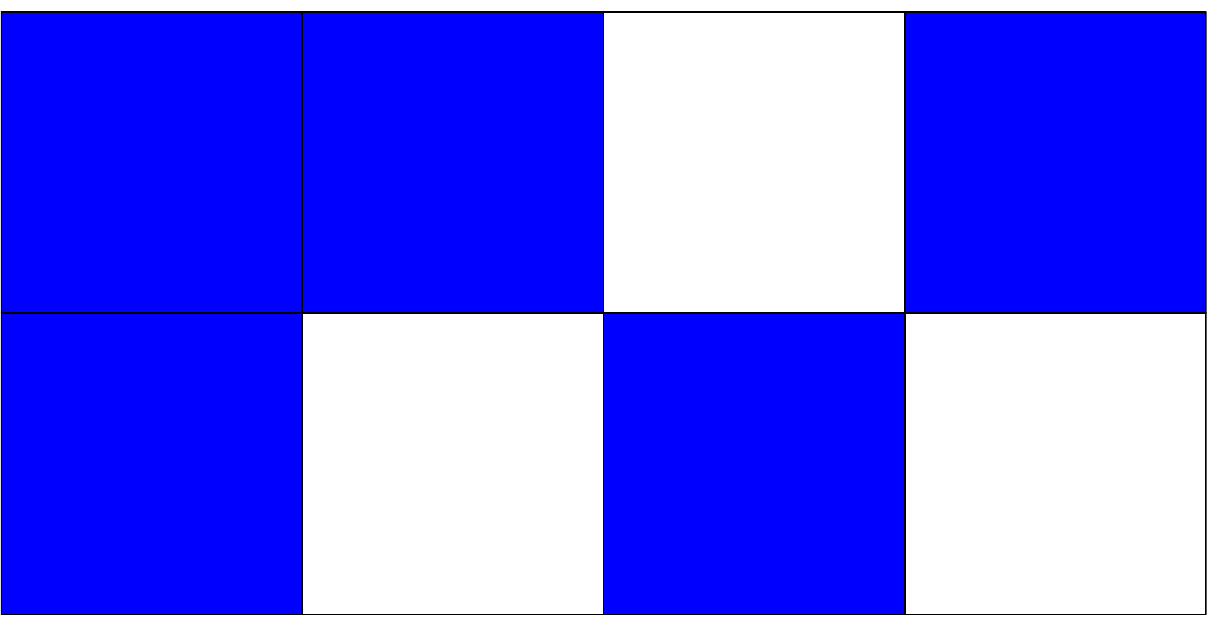} & 2x5 & \includegraphics[width=0.160000\linewidth]{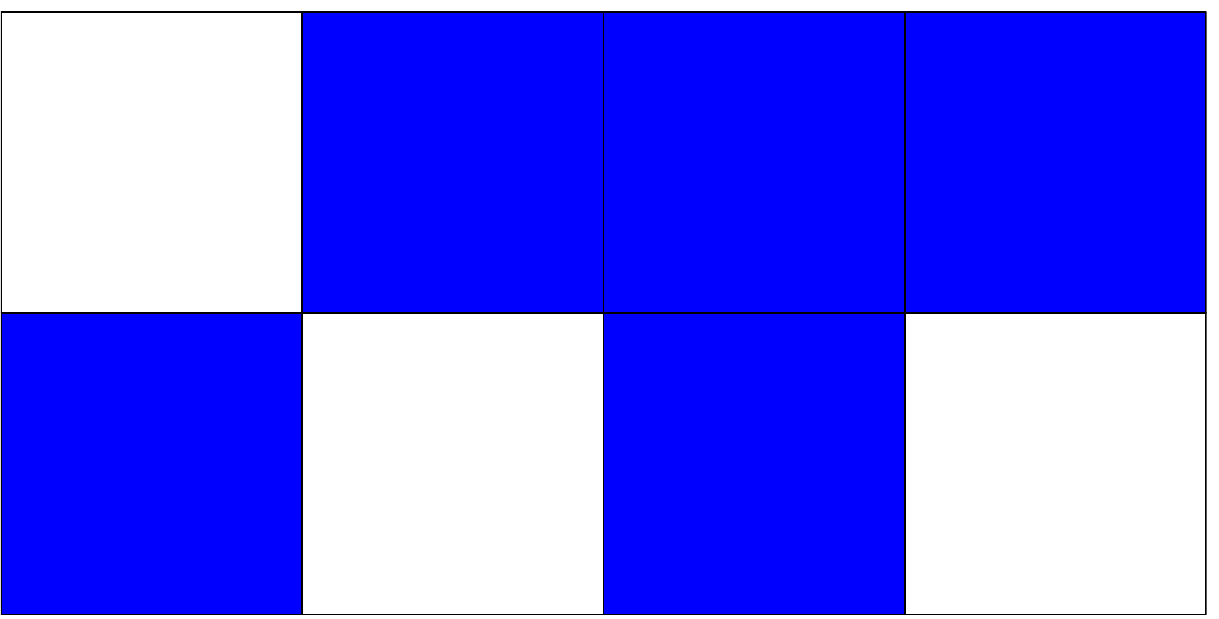} & 4x10 \\ \hline
& & & \\ 
\includegraphics[width=0.160000\linewidth]{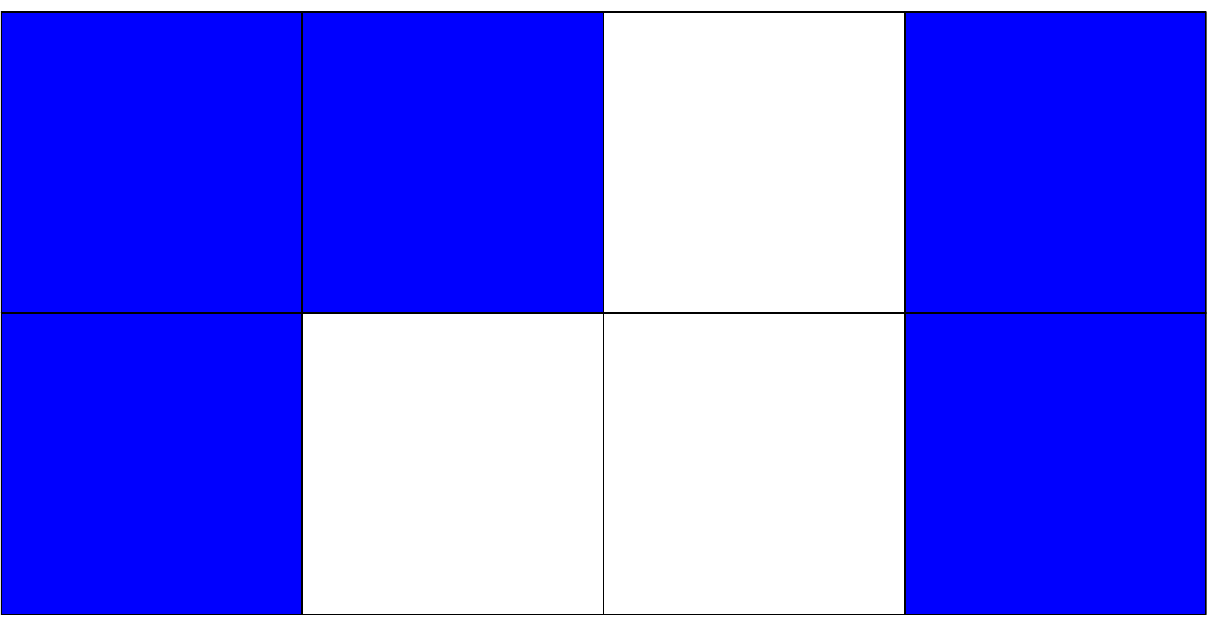} & 2x5 & \includegraphics[width=0.160000\linewidth]{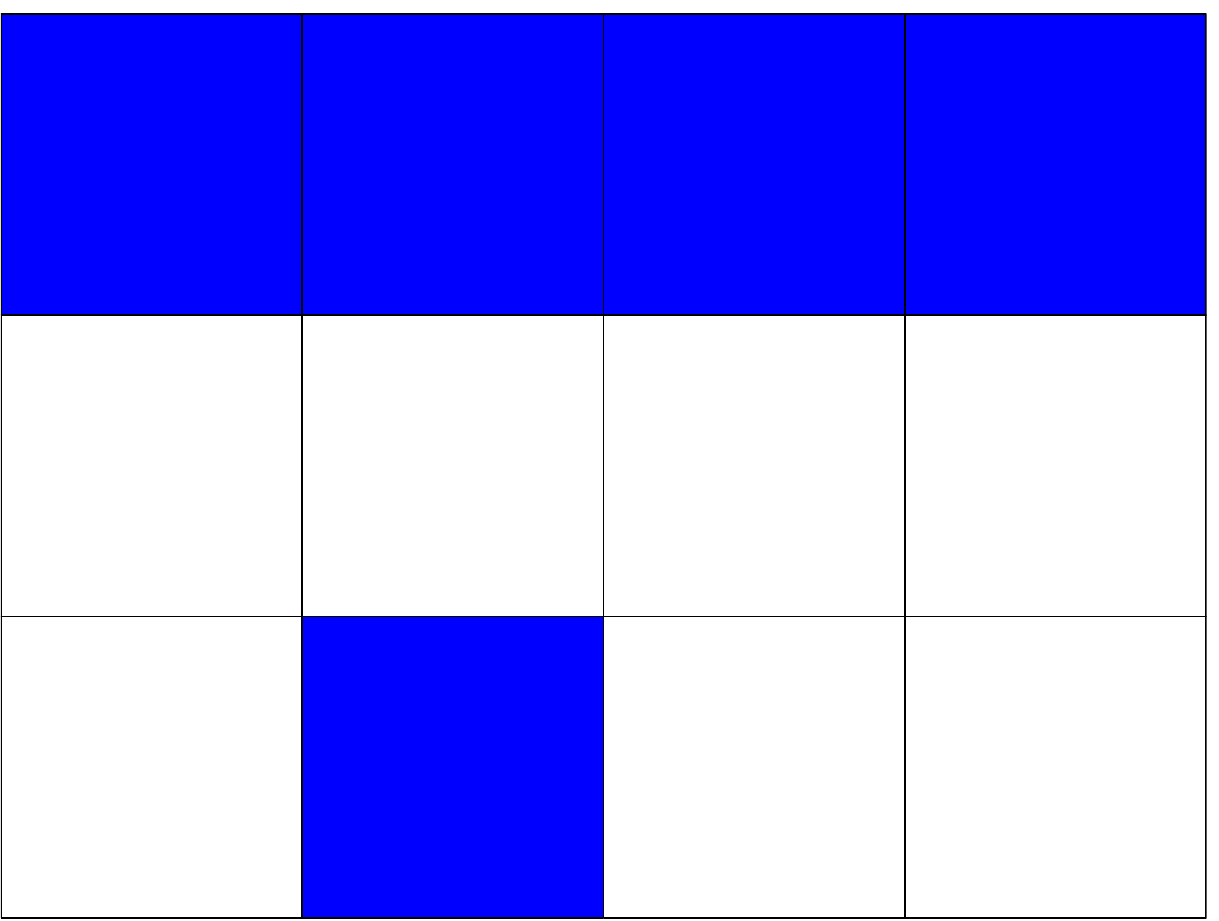} & 5x6 \\ \hline
& & & \\ 
\includegraphics[width=0.160000\linewidth]{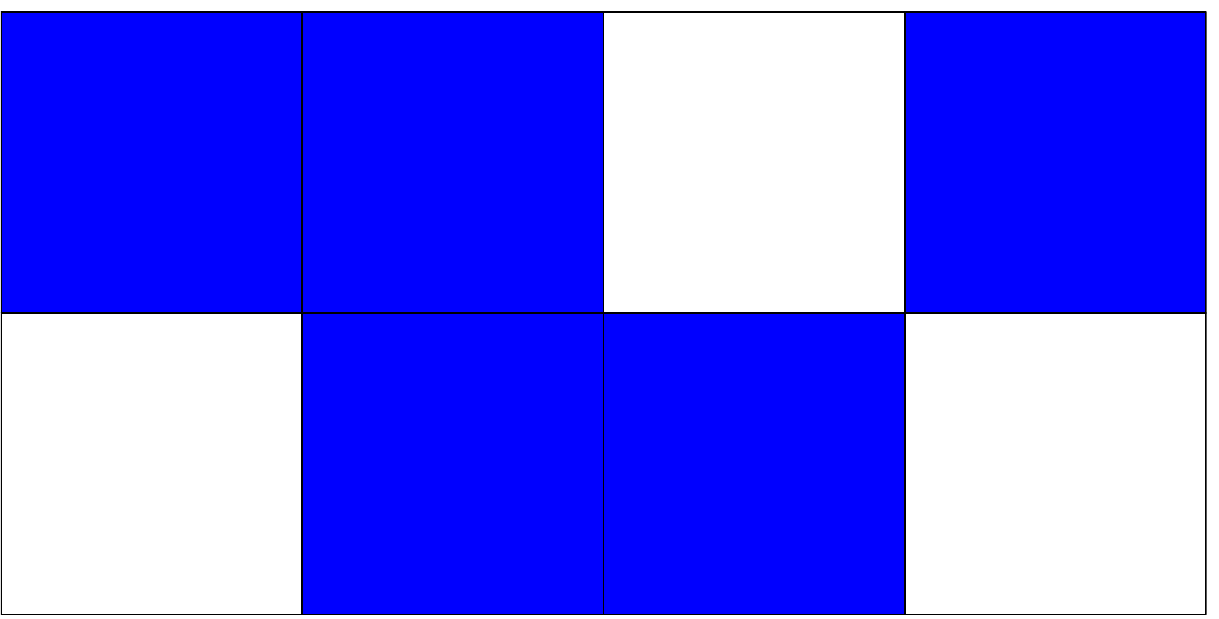} & 4x10 & \includegraphics[width=0.200000\linewidth]{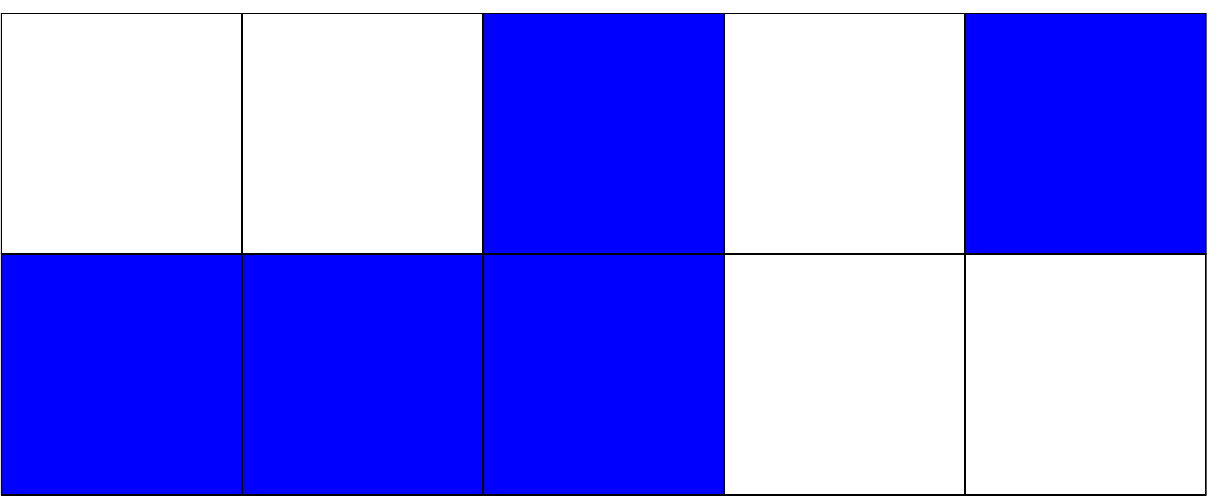} & 10x12 \\ \hline
& & & \\ 
\includegraphics[width=0.120000\linewidth]{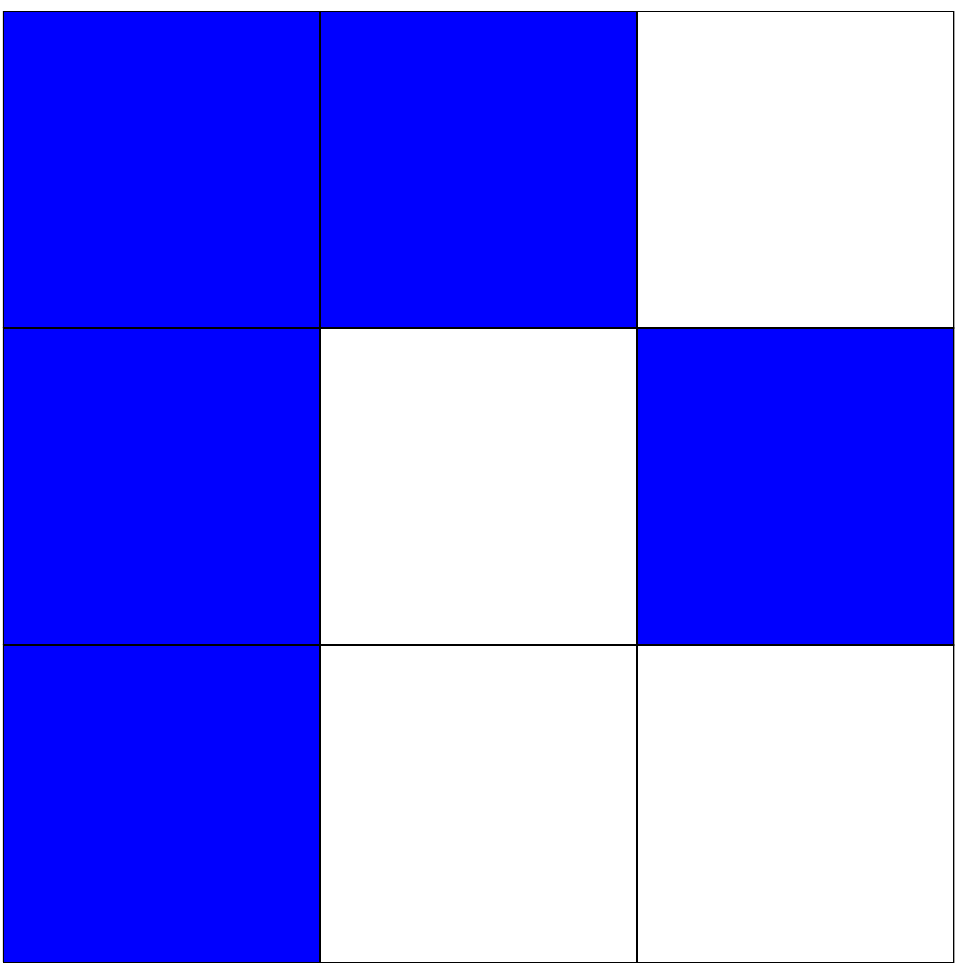} & 6x10 & \includegraphics[width=0.120000\linewidth]{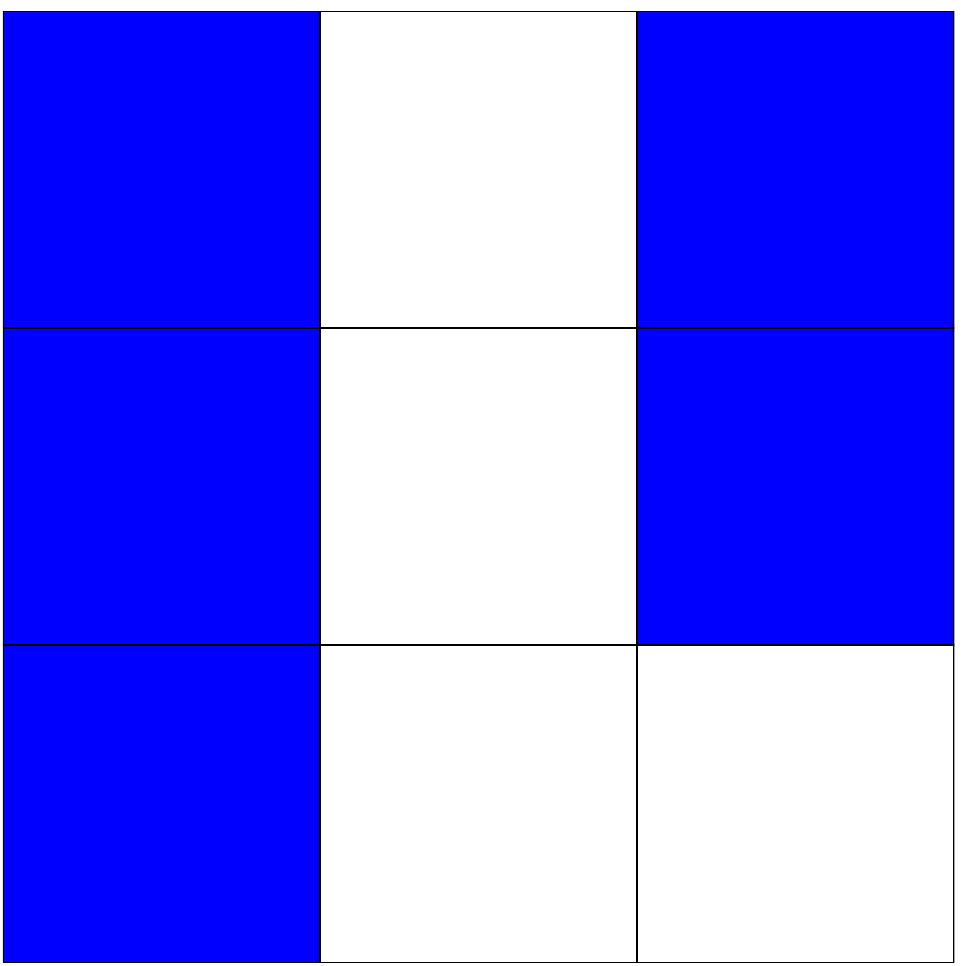} & 4x5 \\ \hline
& & & \\ 
\includegraphics[width=0.120000\linewidth]{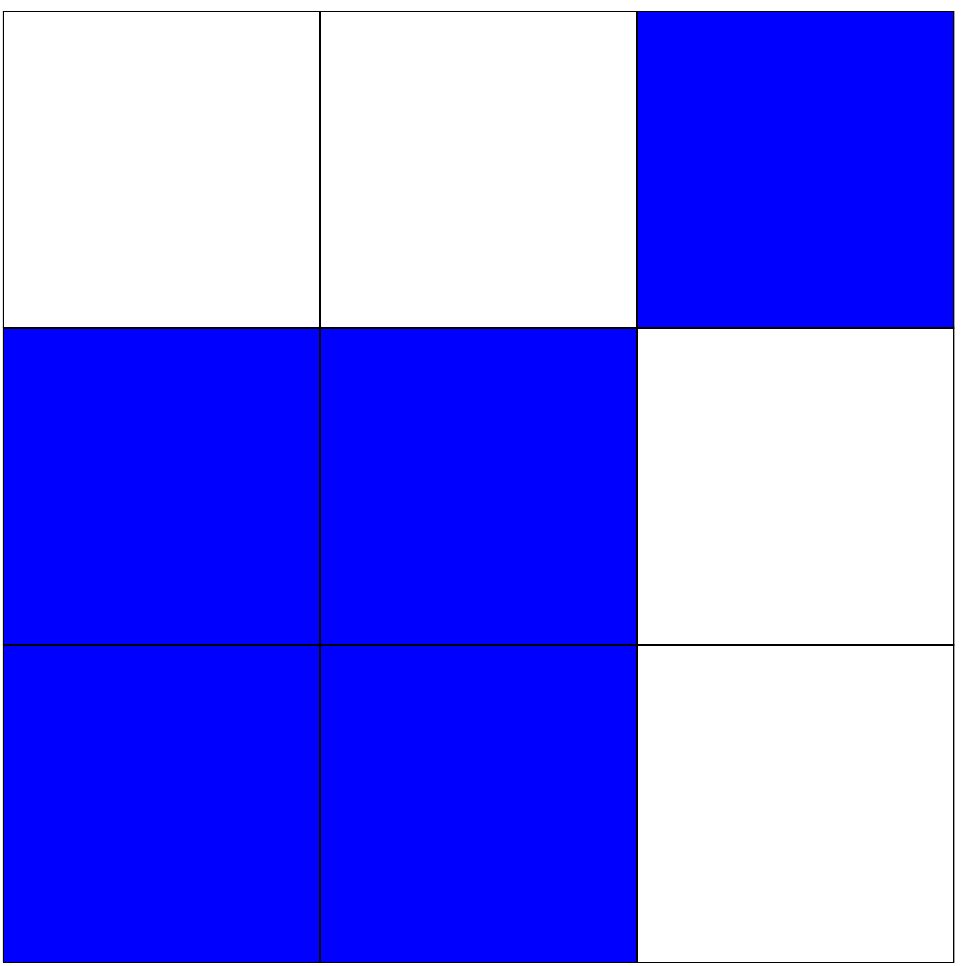} & 20x24 & \includegraphics[width=0.200000\linewidth]{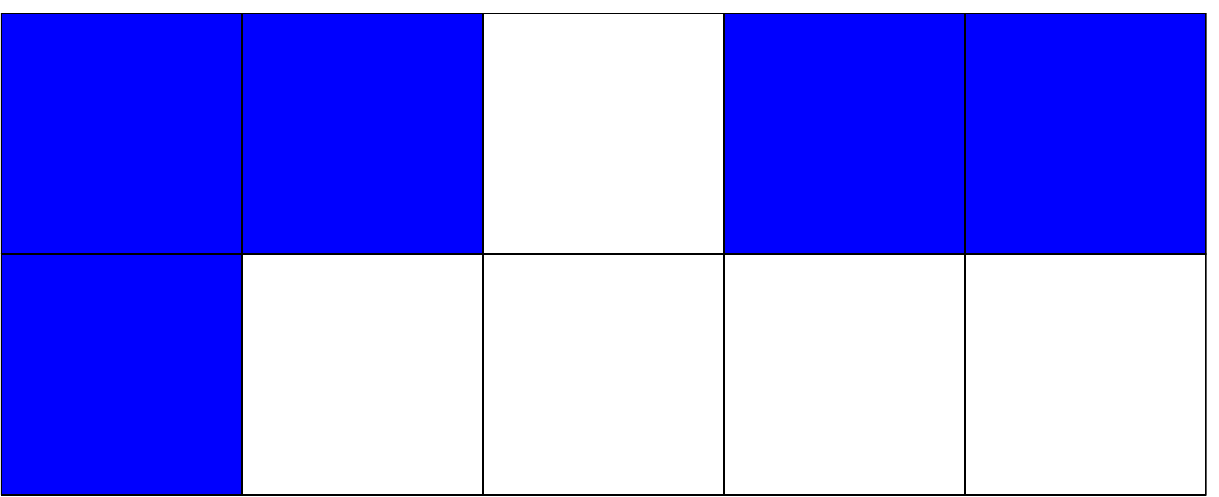} & 11x40 \\ \hline
\end{tabular}
\caption{Smallest known solutions for n=5 and k=1.}
\label{tab:trivial-51b}
\end{table}

\begin{table}[!htpb]
\centering
\begin{tabular}{|c|c|c|}
\hline
Piece & Solution & Size\\ \hline
 & & \\ 
\includegraphics[width=0.076923\linewidth]{tile54_11.pdf} &\includegraphics[width=0.923077\linewidth]{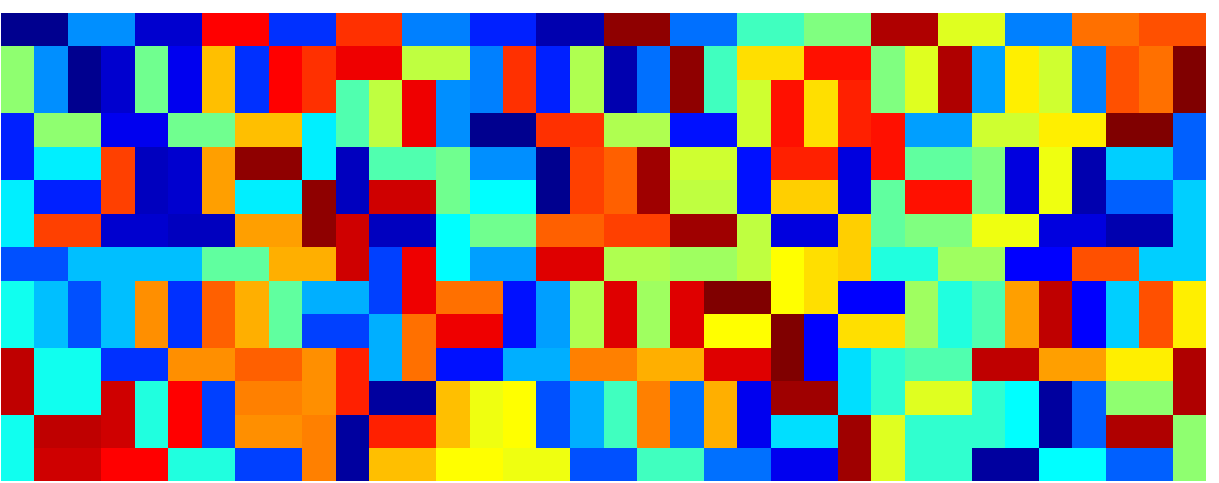} &14x36 \\ \hline
 & & \\ 
\includegraphics[width=0.120000\linewidth]{tile54_15.pdf} &\includegraphics[width=0.480000\linewidth]{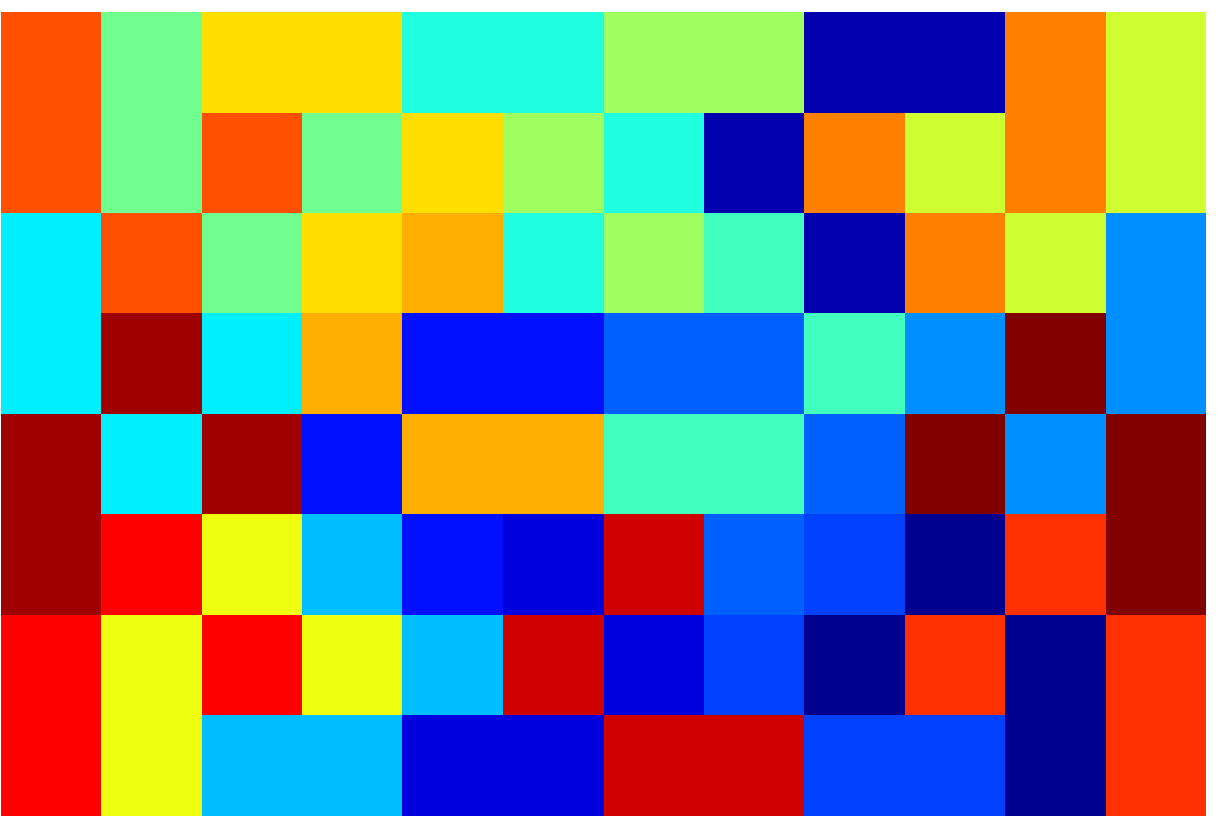} &8x12 \\ \hline
 & & \\ 
\includegraphics[width=0.230769\linewidth]{tile65_2.pdf} &\includegraphics[width=0.769231\linewidth]{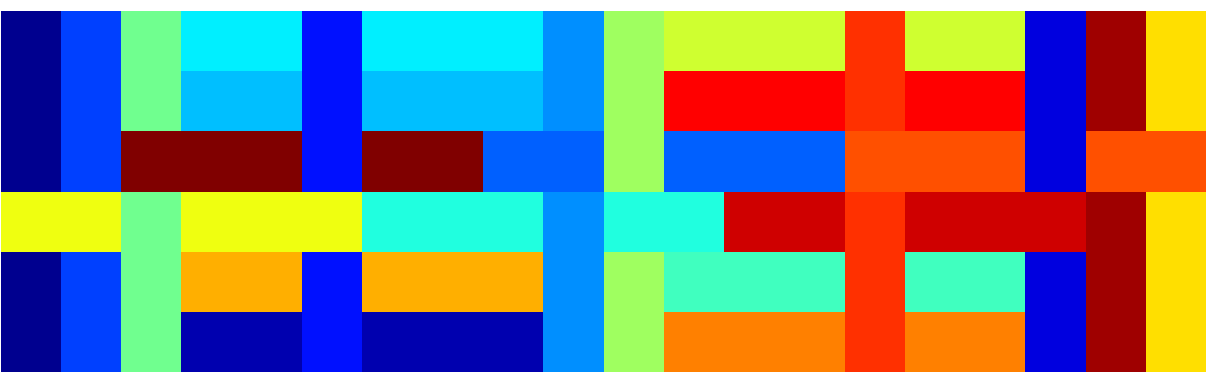} &6x20 \\ \hline
 & & \\ 
\includegraphics[width=0.200000\linewidth]{tile65_3.pdf} &\includegraphics[width=0.400000\linewidth]{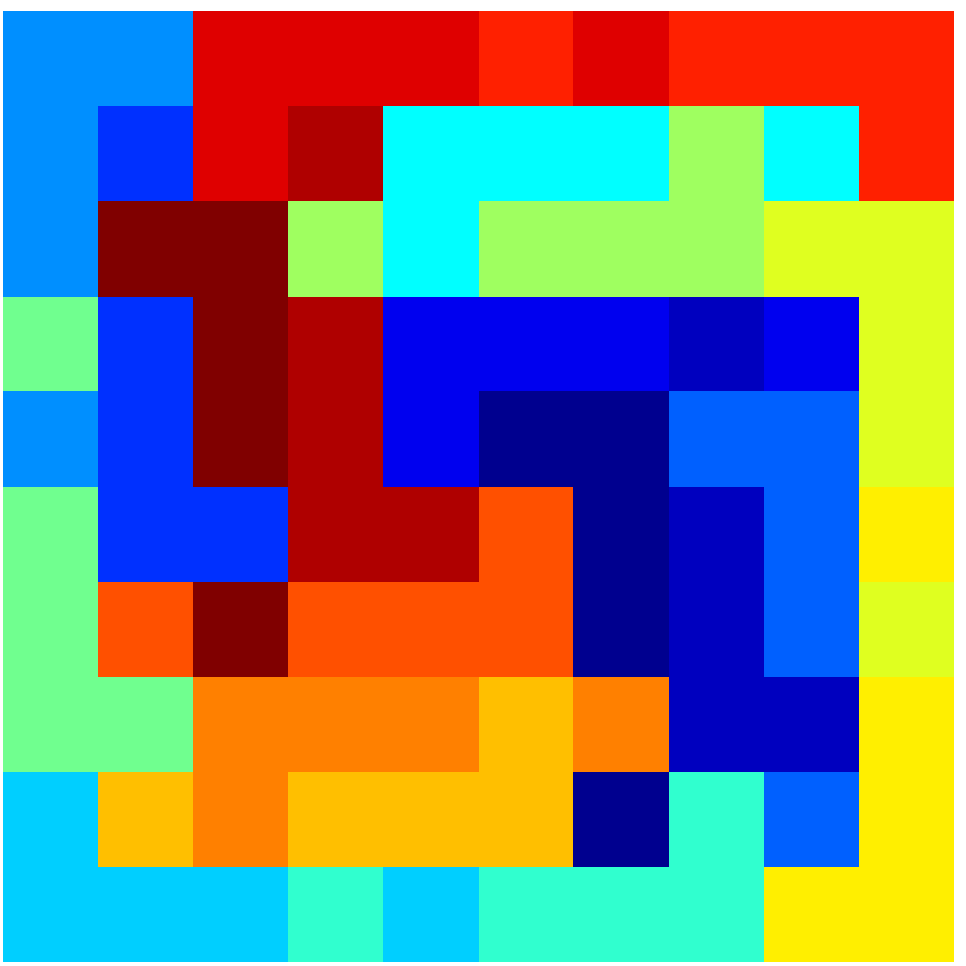} &10x10 \\ \hline
\end{tabular}
\caption{Non-trivial solutions 1.}
\label{tab:cool1}
\end{table}

\begin{table}[!htpb]
\centering
\begin{tabular}{|c|c|c|}
\hline
Piece & Solution & Size\\ \hline
 & & \\ 
\includegraphics[width=0.200000\linewidth]{tile65_5.pdf} &\includegraphics[width=0.400000\linewidth]{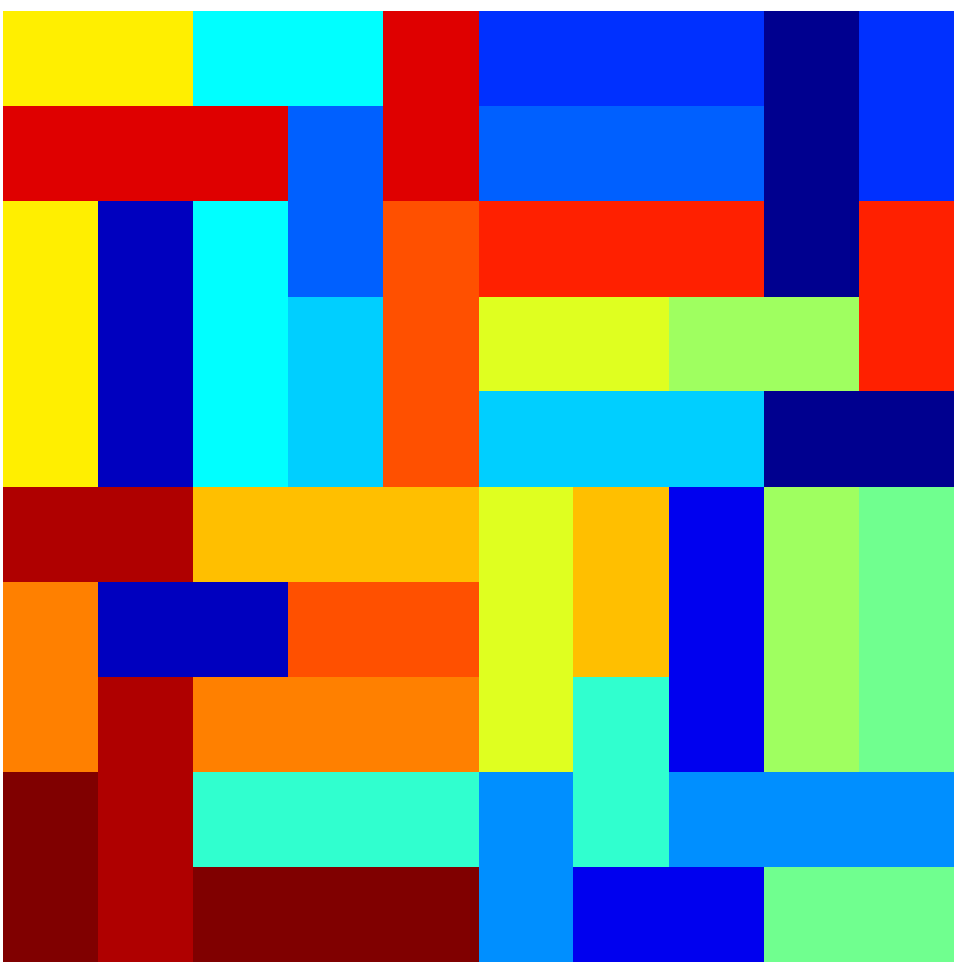} &10x10 \\ \hline
 & & \\ 
\includegraphics[width=0.200000\linewidth]{tile65_8.pdf} &\includegraphics[width=0.600000\linewidth]{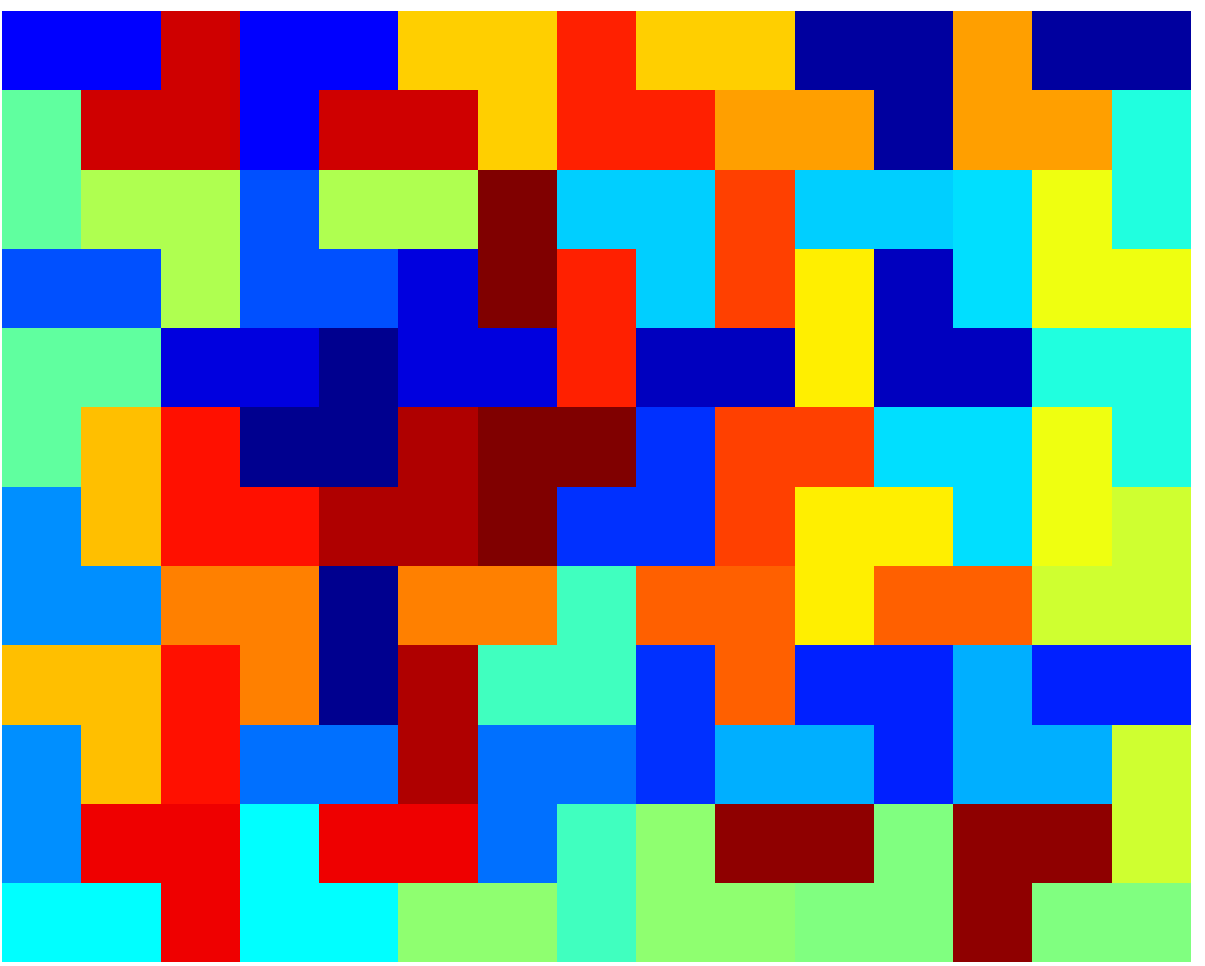} &12x15 \\ \hline
 & & \\ 
\includegraphics[width=0.111111\linewidth]{tile65_4.pdf} &\includegraphics[width=0.888888\linewidth]{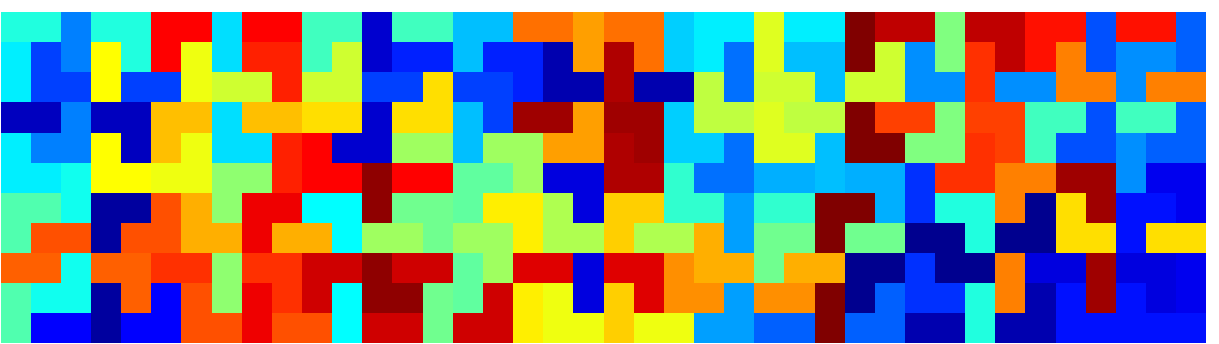} &11x40 \\ \hline
\end{tabular}
\caption{Non-trivial solutions 2.}
\label{tab:cool2}
\end{table}

\begin{table}[!htpb]
\centering
\begin{tabular}{|c|c|c|}
\hline
Piece & Solution & Size\\ \hline
 & & \\ 
\includegraphics[width=0.200000\linewidth]{tile65_10.pdf} &\includegraphics[width=0.720000\linewidth]{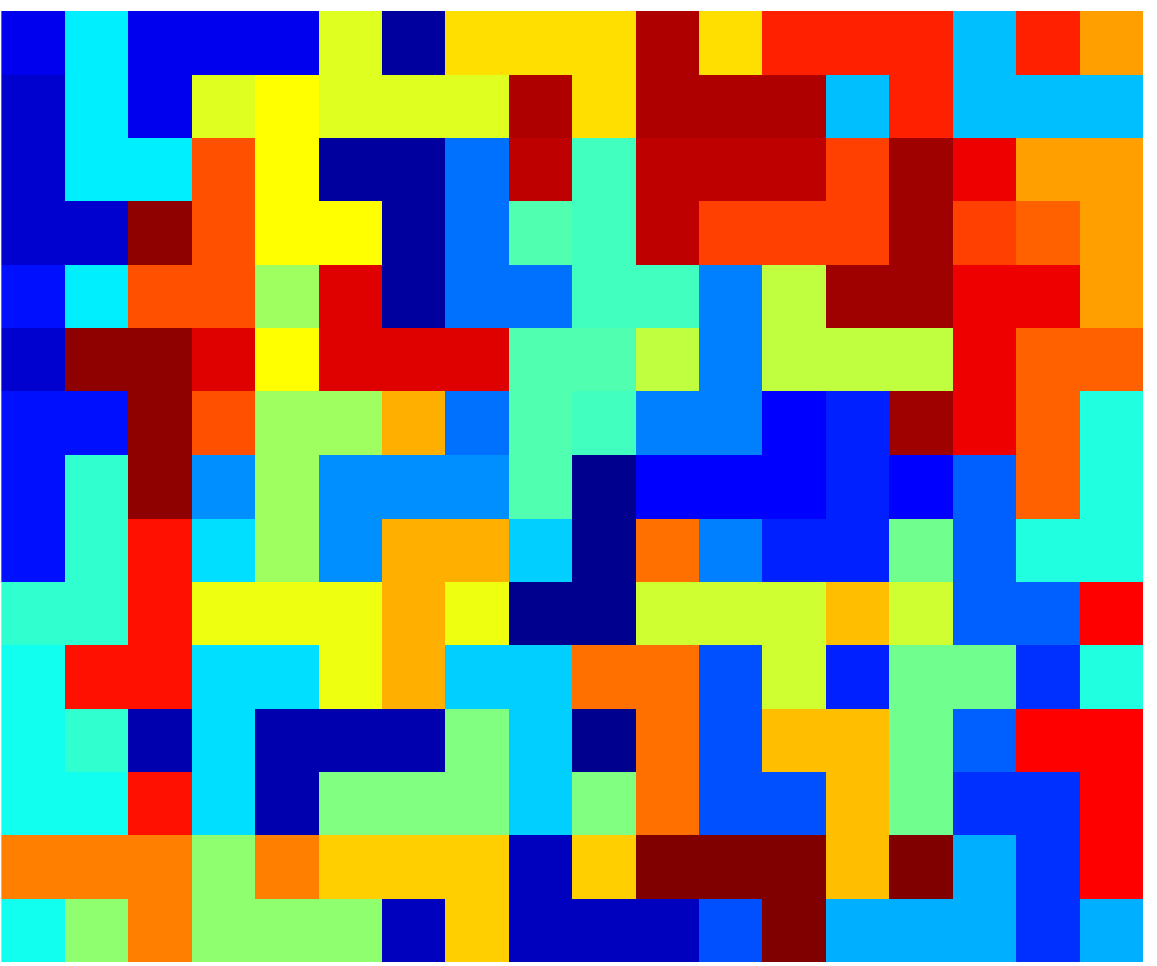} &15x18 \\ \hline
 & & \\ 
\includegraphics[width=0.185185\linewidth]{tile65_13.pdf} &\includegraphics[width=0.814815\linewidth]{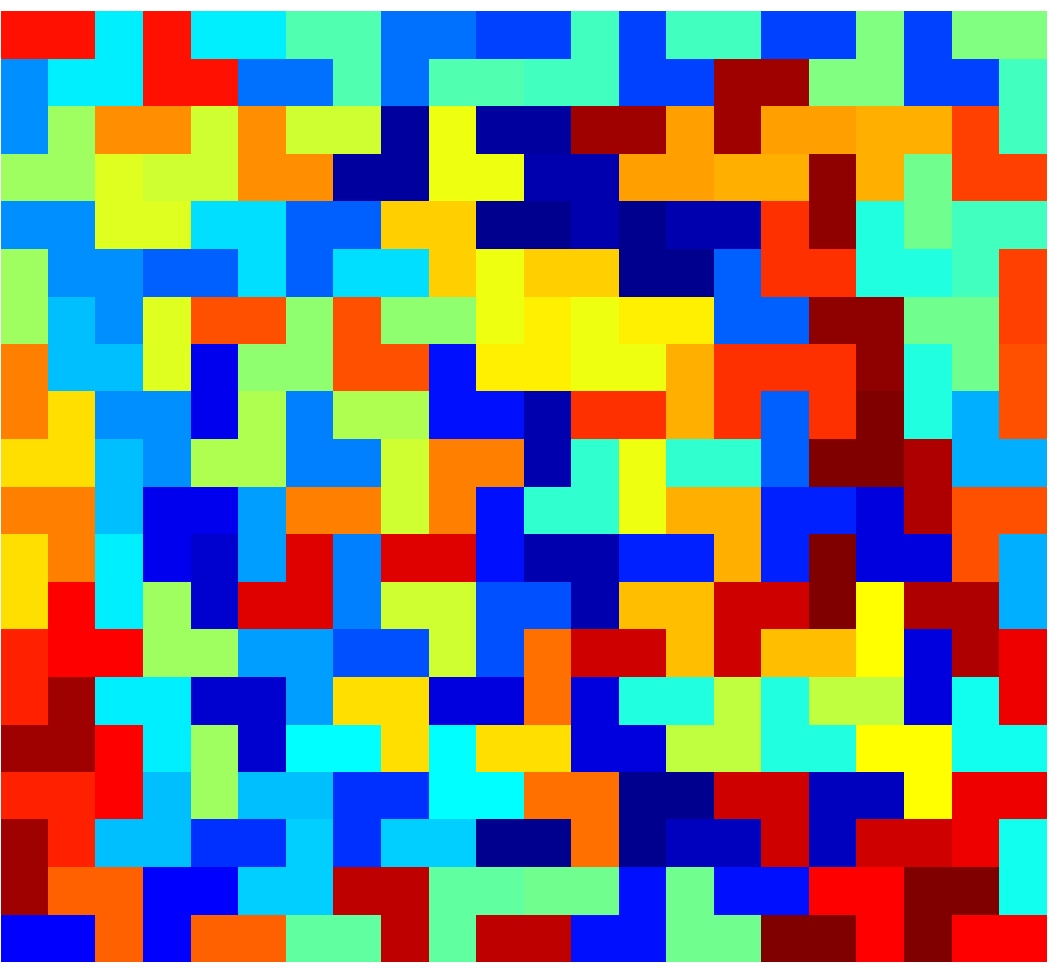} &20x22 \\ \hline
\end{tabular}
\caption{Non-trivial solutions 3.}
\label{tab:cool3}
\end{table}

\pagebreak
\begin{table}[!htpb]
\centering
\begin{tabular}{|c|c|c|}
\hline
Piece & Solution & Size\\ \hline
 & & \\ 
\includegraphics[width=0.160000\linewidth]{tile65_16.pdf} &\includegraphics[width=0.800000\linewidth]{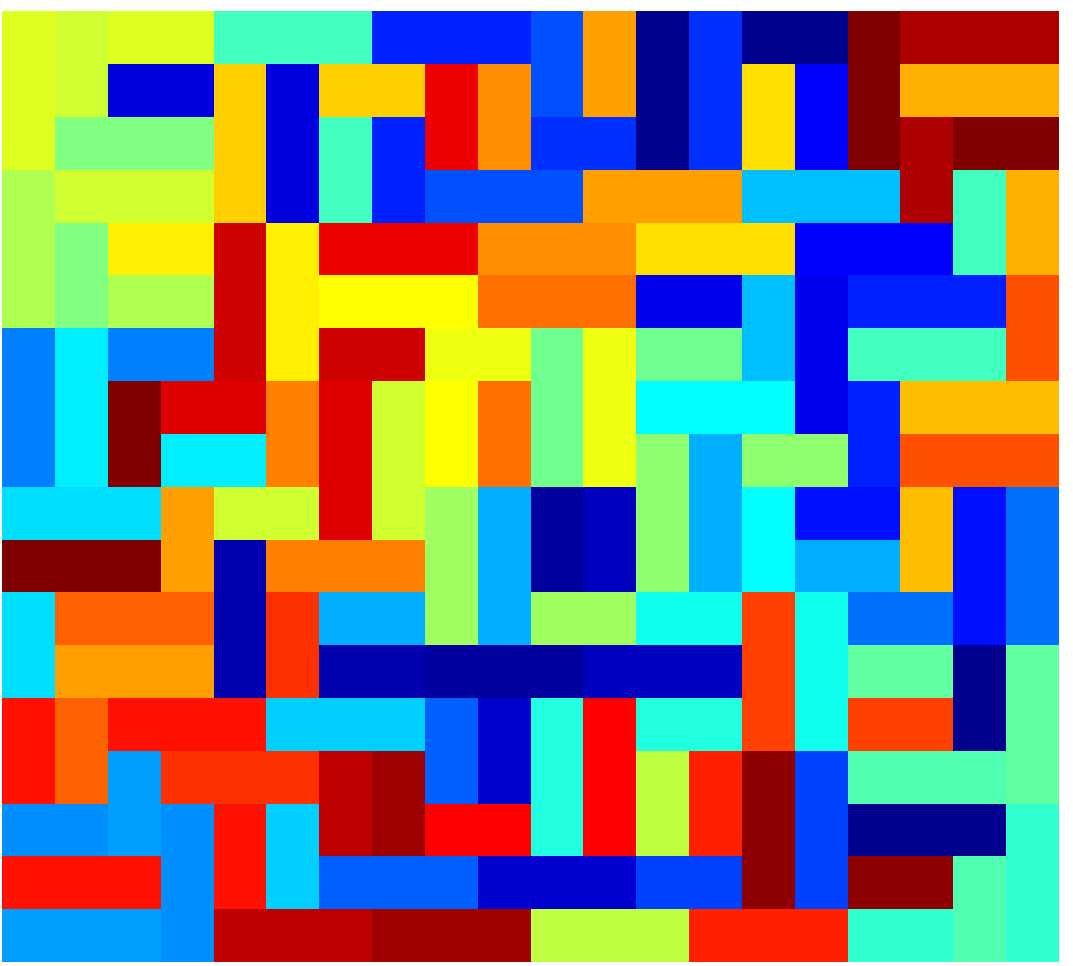} &18x20 \\ \hline
\end{tabular}
\caption{Non-trivial solutions 4.}
\label{tab:cool4}
\end{table}

\pagebreak
\begin{table}[!htpb]
\centering
\begin{tabular}{|c|c|c|}
\hline
Piece & Solution & Size\\ \hline
 & & \\ 
\includegraphics[width=0.160000\linewidth]{tile65_17.pdf} &\includegraphics[width=0.800000\linewidth]{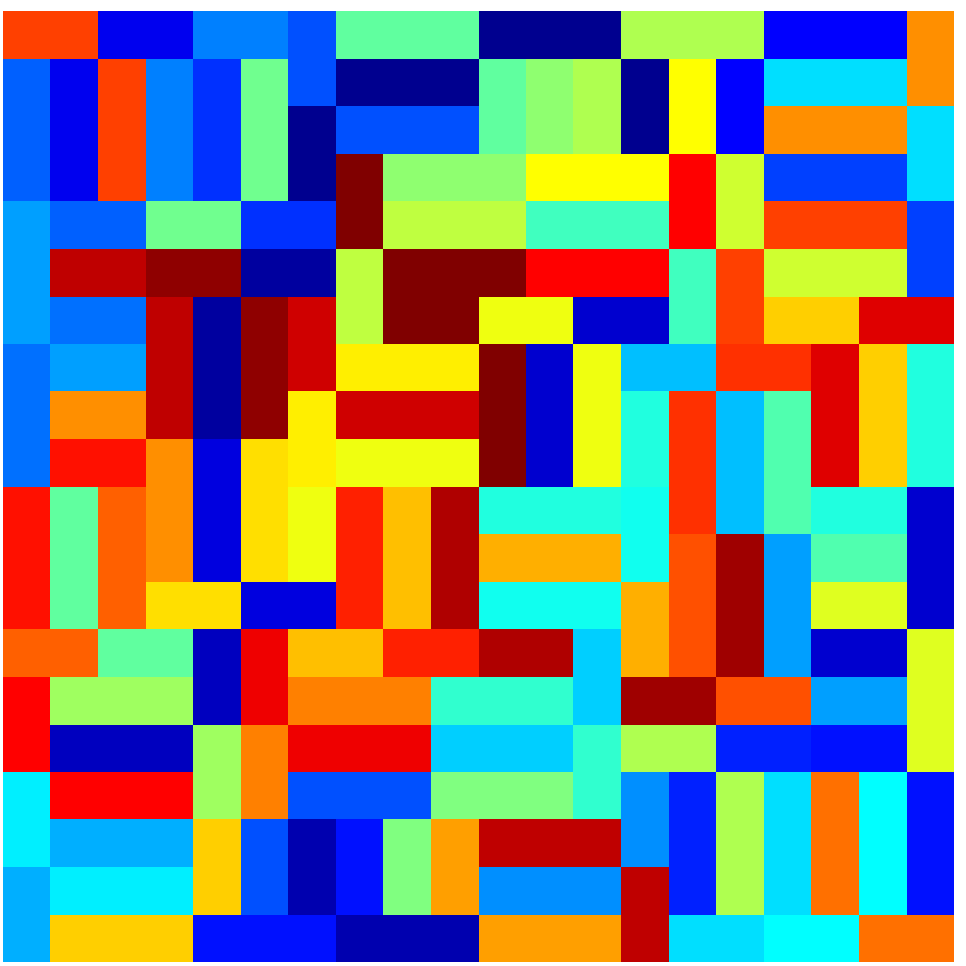} &20x20 \\ \hline
 & & \\ 
\includegraphics[width=0.160000\linewidth]{tile65_22.pdf} &\includegraphics[width=0.400000\linewidth]{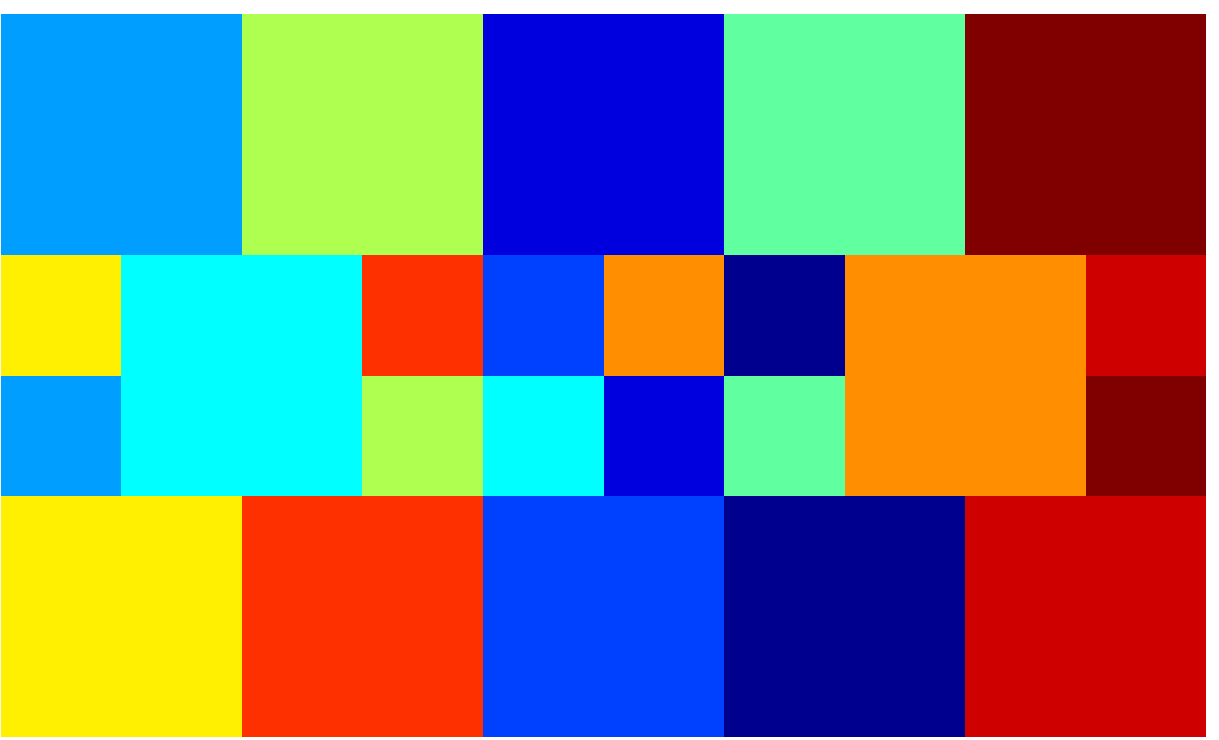} &6x10 \\ \hline
\end{tabular}
\caption{Non-trivial solutions 5.}
\label{tab:cool5}
\end{table}

\pagebreak
\begin{table}[!htpb]
\centering
\begin{tabular}{|c|c|c|}
\hline
Piece & Solution & Size\\ \hline
 & & \\ 
\includegraphics[width=0.200000\linewidth]{tile65_44.pdf} &\includegraphics[width=0.480000\linewidth]{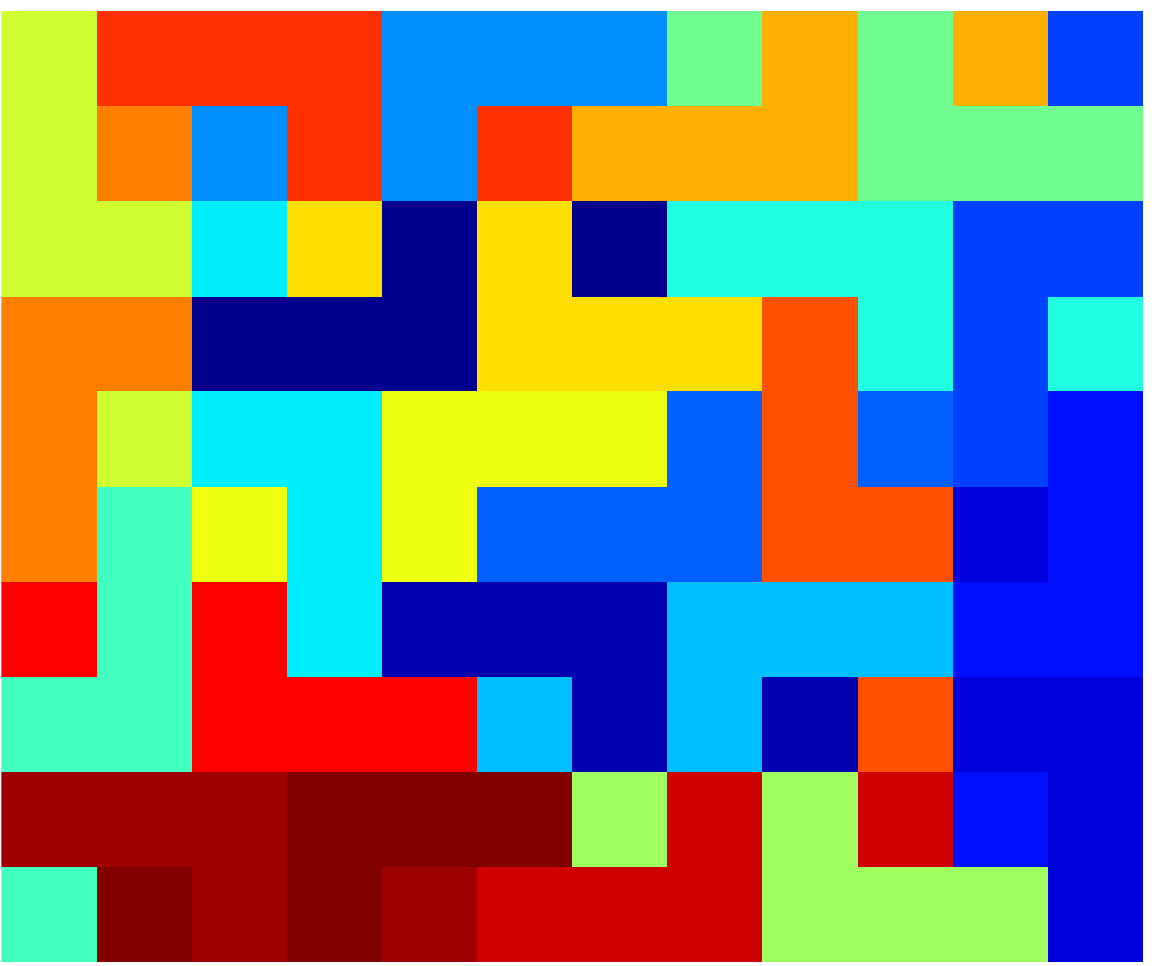} &10x12 \\ \hline
 & & \\ 
\includegraphics[width=0.120000\linewidth]{tile65_60.pdf} &\includegraphics[width=0.400000\linewidth]{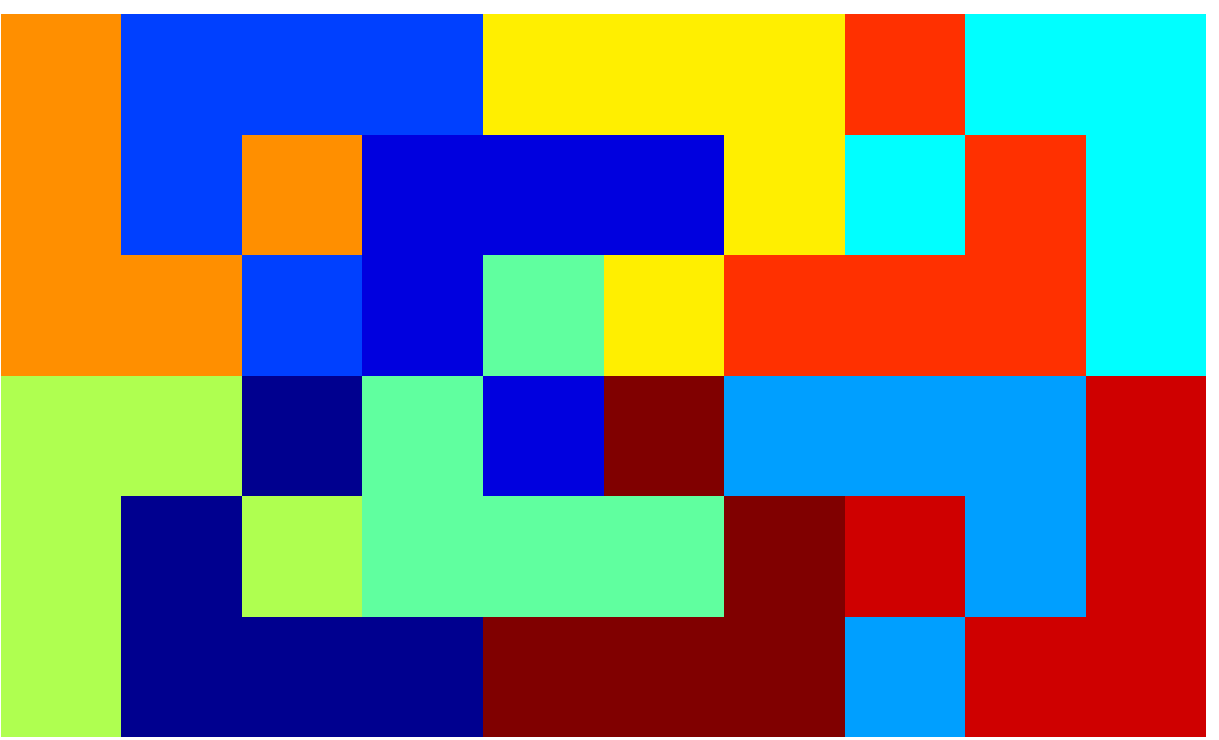} &6x10 \\ \hline
 & & \\ 
\includegraphics[width=0.111111\linewidth]{tile65_67.pdf} &\includegraphics[width=0.888889\linewidth]{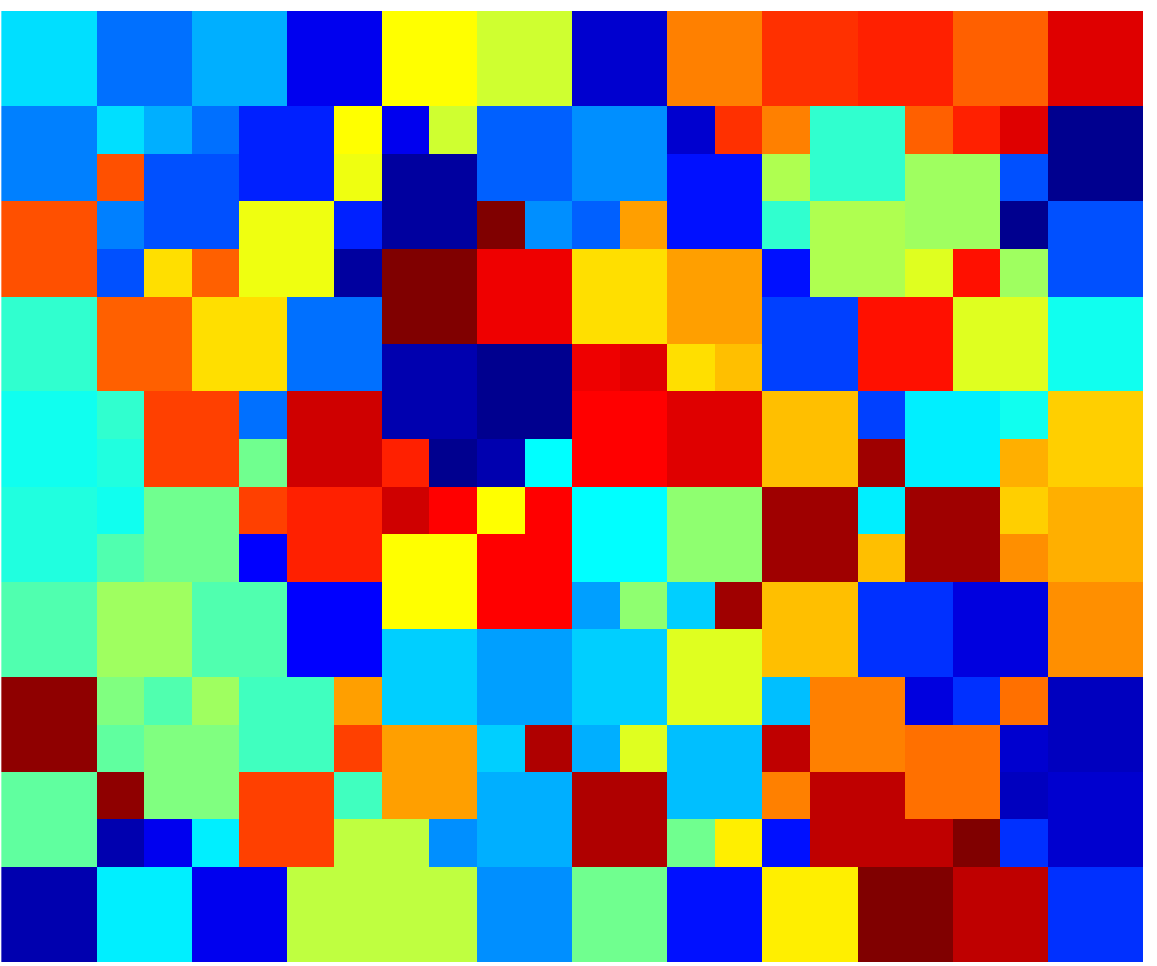} &20x24 \\ \hline
\end{tabular}
\caption{Non-trivial solutions 6.}
\label{tab:cool6}
\end{table}


\pagebreak
\subsection{Impossible polyominoes}

In some trivial cases it is possible to determine if a polyomino is unrectifiable. First we define a tile's \emph{bounding box}:

\begin{defn}
A tile's \emph{bounding box} is the smallest rectangle that surrounds the tile's visible cells.
\end{defn}

\begin{theorem}
A tile is not rectifiable if every corner cell of its bounding box is transparent.
\label{th:impossible2}
\end{theorem}
\begin{proof}
If every corner cell of a polyomino's bounding box is transparent then 
it cannot tile a corner. Hence such a polyomino cannot tile a rectangle.
\end{proof}
For example, this theorem can applied to the right-most tile in Table~\ref{tab:impossible-41}.
Note that there are rectifiable tiles with 3 (out of 4) transparent corners in their bounding box,
for example see second row in Table~\ref{tab:cool1}.

\begin{defn}
An empty (transparent) square cannot be \emph{filled} with a tile if there is no way of legally placing that tile such that
the square becomes visible.
\end{defn}
\begin{theorem}
A tile is not rectifiable if it has a transparent square inside its bounding box that cannot be filled with that tile.
\label{th:impossibleHole}
\end{theorem}
\begin{proof}
For a tile to be rectifiable, it must be possible to fill every transparent square inside its bounding box.
Hence if there is a transparent square that cannot be filled then the tile is not rectifiable.
\end{proof}
The above theorem can be applied to the tile shown in Figure~\ref{fig:impossibleHole}, where it is impossible to fill in the centre square.

\begin{figure}[!htpb]
\centering
\includegraphics[width=0.20\linewidth]{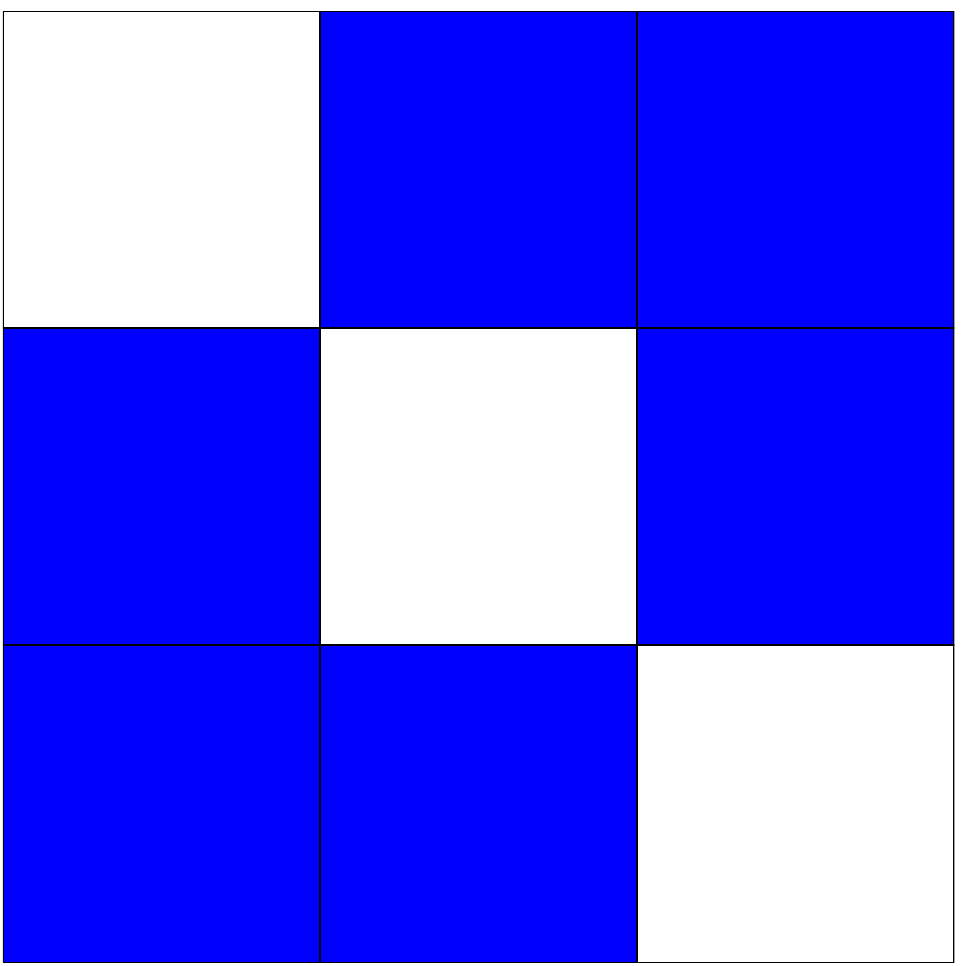}
\caption{This polyomino is unrectifiable, because it is impossible to fill its centre square.}
\label{fig:impossibleHole}
\end{figure}

Apart from the simple cases covered by Theorems~\ref{th:impossible2}-\ref{th:impossibleHole},
we do not have an efficient method for verifying that a piece is impossible.
In fact
there is no single pattern of visible and transparent squares whose existence in the polyomino would
signal that the polyomino is impossible. This is shown in Theorem~\ref{th:impossible}.

\begin{defn}
A polyomino $P$ \emph{contains} polyomino $T$, if $T$ can be placed on top of $P$ (superpositioned)
such that its visible squares lie on top of $P$'s visible squares and its transparent squares lie
on top of $P$'s transparent squares. Note that a polyomino always contains itself.
\end{defn}
\begin{theorem}
For every polyomino $T$ there exists a rectifiable polyomino $P$ that contains $T$.
\label{th:impossible}
\end{theorem}
\begin{proof}
If $T$ is already rectifiable then set $P$ to $T$. Otherwise we need to show that 
any unrectifiable polyomino $T$ can be made rectifiable by adding extra
squares to it. Consider an arbitrary unrectifiable polyomino $T$ (such as the one in Figure~\ref{fig:impossible}(a)).
Create the inverse polyomino $T'$ by making all visible squares in $T$ transparent and vice versa.
Create a polyomino $P$ by rotating $T'$ 180 degrees and placing it next to $T$ (see Figure~\ref{fig:impossible}(b)).
Clearly $P$ contains $T$. $P$ is rectifiable, because it forms a rectangle when combined with another copy
(see Figure~\ref{fig:impossible}(c)).
\end{proof}

\begin{figure}[!htpb]
\centering
\begin{tabular}{ccc}
\includegraphics[width=0.18\linewidth]{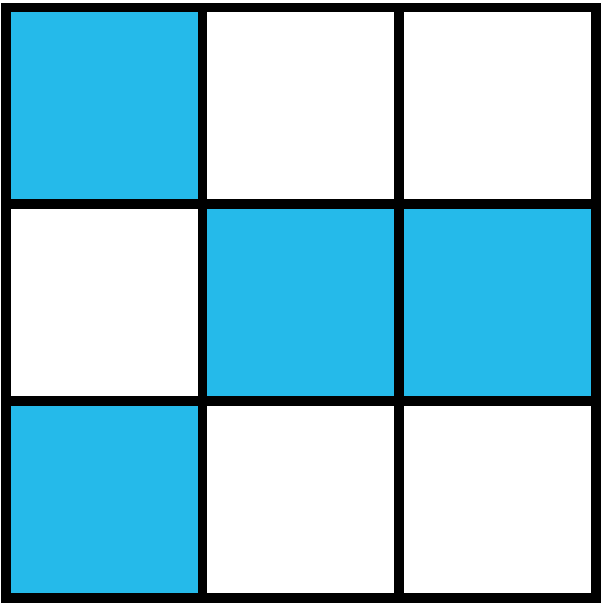} & \includegraphics[width=0.36\linewidth]{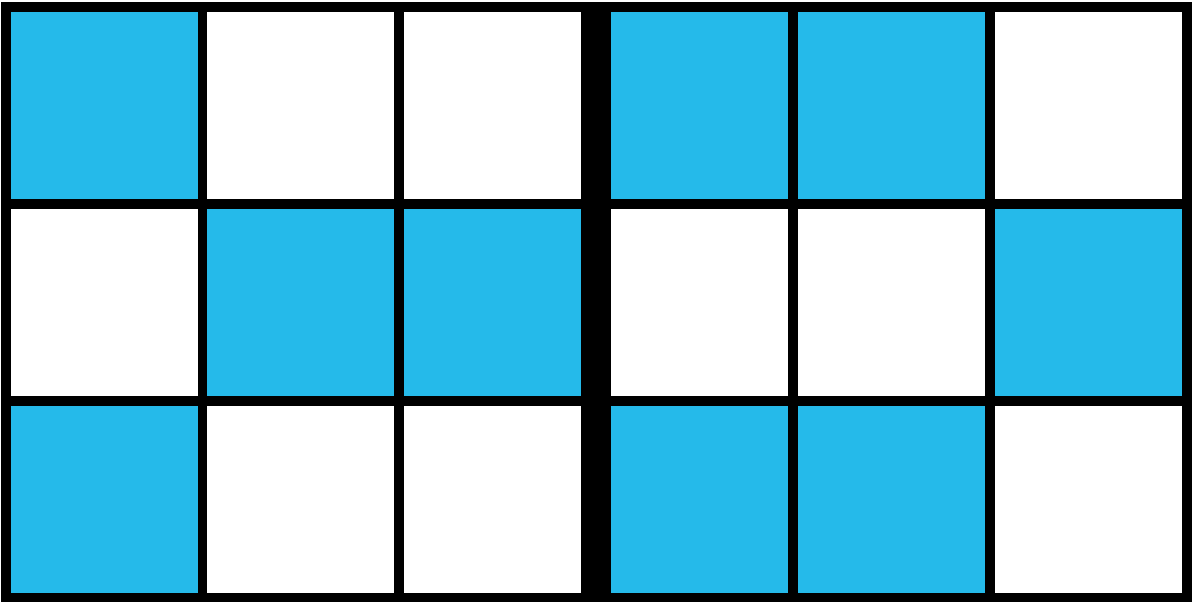} & \includegraphics[width=0.36\linewidth]{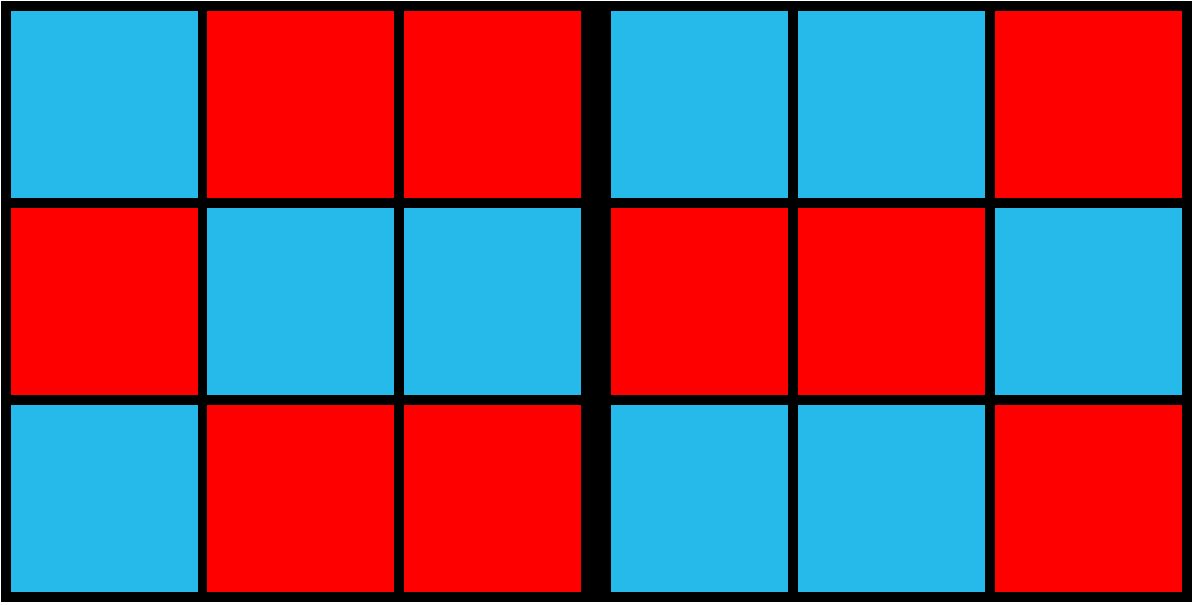} \\
(a) & (b) & (c)
\end{tabular}
\caption{Proof that every unrectifiable polyomino can be made rectifiable.
(a) Unrectifiable polyomino $T$. (b) Rectifiable polyomino $P$ composed of $T$ and its inverse $T'$.
(c) Rectangle tiles by $P$.}
\label{fig:impossible}
\end{figure}

Theorem~\ref{th:impossible} shows how any unrectifiable polyomino can be converted to a rectifiable one by doubling the number of its squares.
However, for some polyominoes it is possible to make the conversion without such a drastic change. In the following example we
make the conversion by making a single transparent square visible. Consider the polyomino in Figure~\ref{fig:impossible2}(left).
If we make square a visible, the resulting polyomino tiles a rectangle (see Figure~\ref{fig:impossible2}(right)).

\begin{figure}[!htpb]
\centering
\begin{tabular}{ccc}
\includegraphics[width=0.30\linewidth]{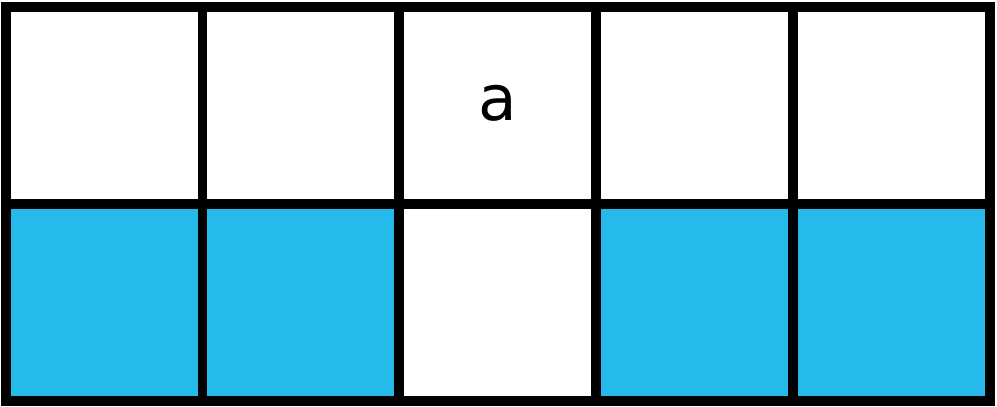} & ~~~~ & \includegraphics[width=0.30\linewidth]{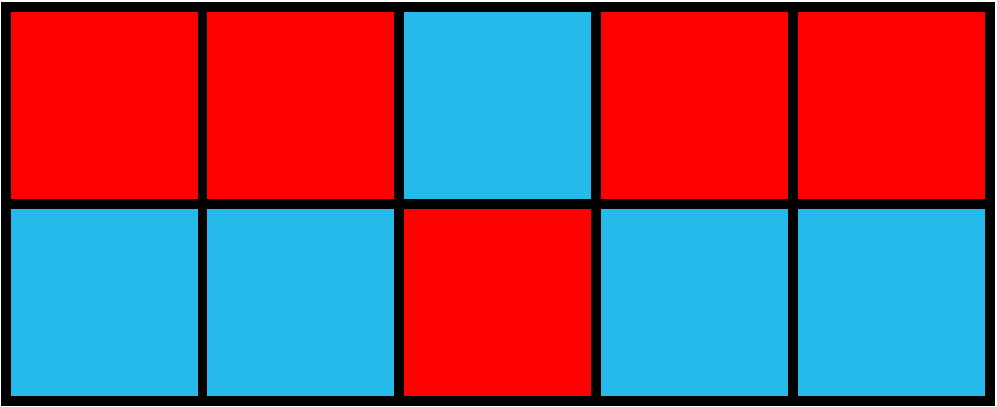}
\end{tabular}
\caption{An impossible polyomino (left) is made rectifiable by making square a visible. The resulting polyomino tiles a rectangle (right).}
\label{fig:impossible2}
\end{figure}

We designed an algorithm which can prove that a polyomino is unrectifiable (see Algorithm~\ref{alg:prover}).
The idea is to start with a sufficiently large empty grid\footnote{In our experiments we found that a $100 \times 100$ grid was sufficient.}
and place new tiles in an empty cell with the smallest value as shown in Figure~\ref{fig:zigzag}. Let $Z(\cdot)$
be a function that takes a grid and returns such a location (see line 1).
This ordering aims to fill the top-left corner, which is a pre-requisite to filling a rectangle.
If we cannot fill this corner after trying all possible tile
placements then the polyomino is not rectifiable.

\begin{figure}[!htpb]
\centering
\includegraphics[width=0.40\linewidth]{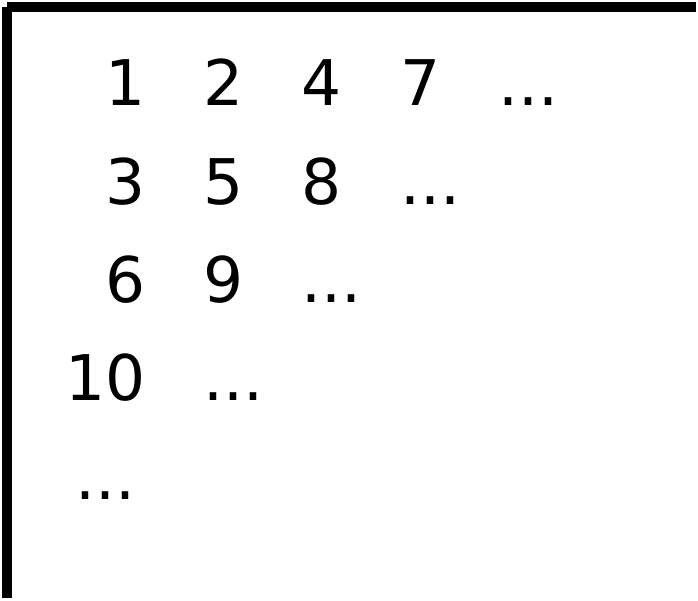}
\caption{}
\label{fig:zigzag}
\end{figure}

\begin{algorithm}[!htpb]
\renewcommand{\arraystretch}{1.15}
\caption{: Algorithm for proving whether a polyomino is unrectifiable.}
\vspace{0.5ex}
\bf function Prove(grid, tile) \\
\begin{tabular}{rl}
\\
1 & $(r,c):=Z(grid)~~~~//first~empty~using~order~from~Figure~\ref{fig:zigzag}$ \\
2 & \\
3 & $//we~found~a~solution$ \\
4 & if $(r,c)=\emptyset$ \\
5 & ~~$grid.print()$ \\
6 & ~~terminate \\
7 & end \\
8 & \\
9 & $//grid~is~too~small,~need~to~enlarge~it$ \\
10 & if $r=grid.height$ \\
11 & ~~terminate \\
12 & end \\
13 & \\
14 & $//for~each~rotated~version~of~the~tile$ \\
15 & $\forall t \in Rot(tile)$ \\
16 & ~~$//move~that~places~t's~visible~cell~with~the~smallest~index$\footnotemark $~at~(r,c)$ \\
17 & ~~$move=~$new $Move(t,r,c)$ \\
18 & ~~if $grid.isValidMove(move)$ \\
19 & ~~~~$grid.makeMove(move)$ \\
20 & ~~~~$Prove(grid,tile)$ ~~~~$//recurse$\\
21 & ~~~~$grid.undoMove(move)$ \\
22 & ~~end \\
23 & end \\
\end{tabular}
\label{alg:prover}
\end{algorithm}
\footnotetext{Here the smallest index is given by the order defined by $Z(\cdot)$.}

If Algorithm~\ref{alg:prover} terminates at line 6 then the piece is rectifiable since we have used it to tile a rectangle.
If it terminates at line 11 then we do not have a proof that the piece is unrectifiable. We can try to enlarge the original
grid, but this may give us the same result. In this case we may never be able to prove that the piece is unrectifiable using this method.
Otherwise, if the algorithm terminates normally then it has proven that the tile is unrectifiable. 
We used Algorithm~\ref{alg:prover} to find most of the impossible pieces in Tables~\ref{tab:impossible-32}-\ref{tab:impossible-51}.

Algorithm~\ref{alg:prover} was not sufficient for some pieces and we had to rely on other methods. Tom Sirgedas~\cite{sirgedas}
presented a nice proof of the impossibility of the last piece in Table~\ref{tab:impossible-41}. His proof is based on the fact that
this piece cannot tile any $3 \times n$ rectangle.



\begin{table}[!htpb]
\centering
\begin{tabular}{|c|c|}
\hline
& \\ 
\includegraphics[width=0.120000\linewidth]{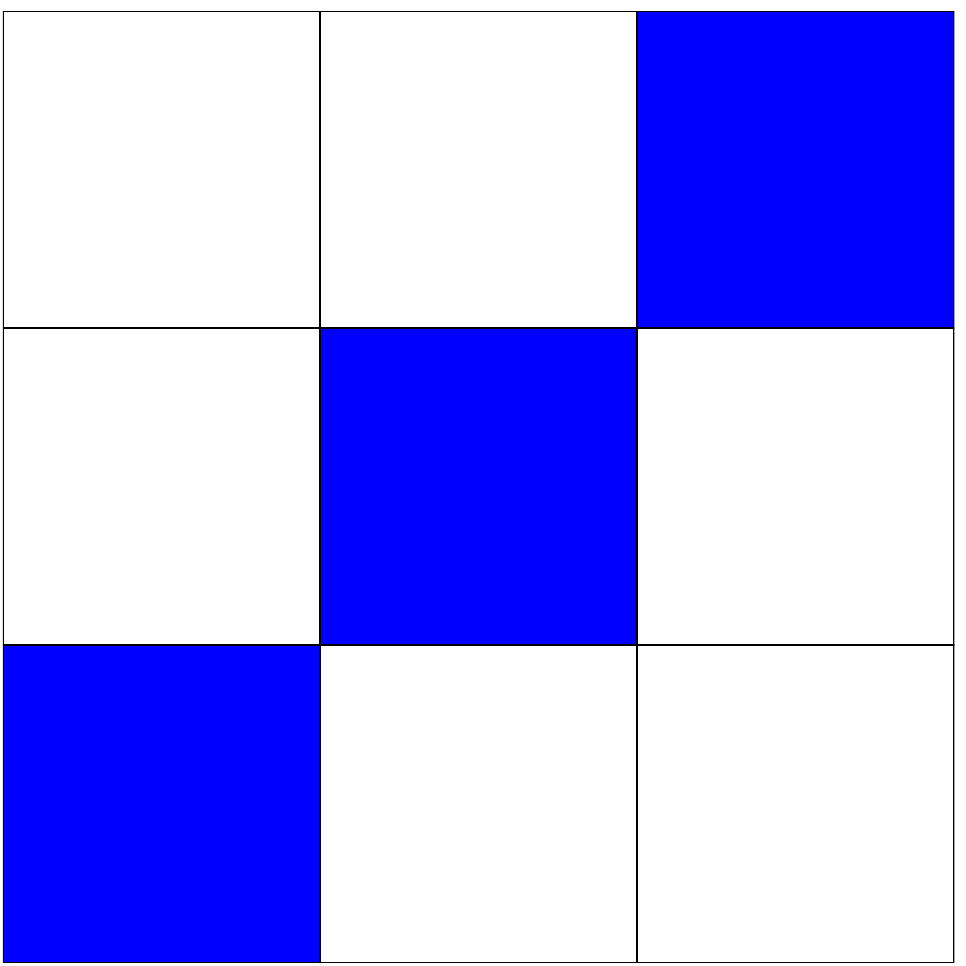} & \includegraphics[width=0.120000\linewidth]{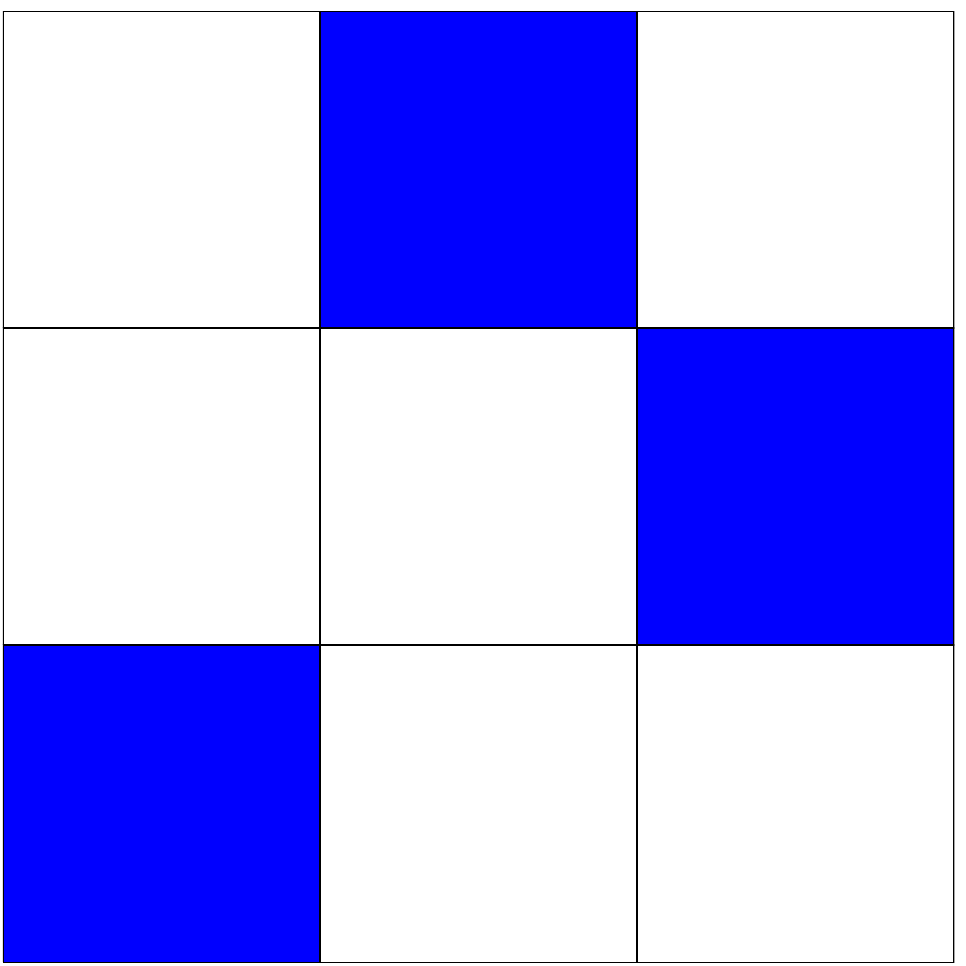} \\ \hline
\end{tabular}
\caption{Impossible pieces for n=3 and k=2.}
\label{tab:impossible-32}
\end{table}

\begin{table}[!htpb]
\centering
\begin{tabular}{|c|c|c|c|c|c|}
\hline
& & & & \\ 
\includegraphics[width=0.200000\linewidth]{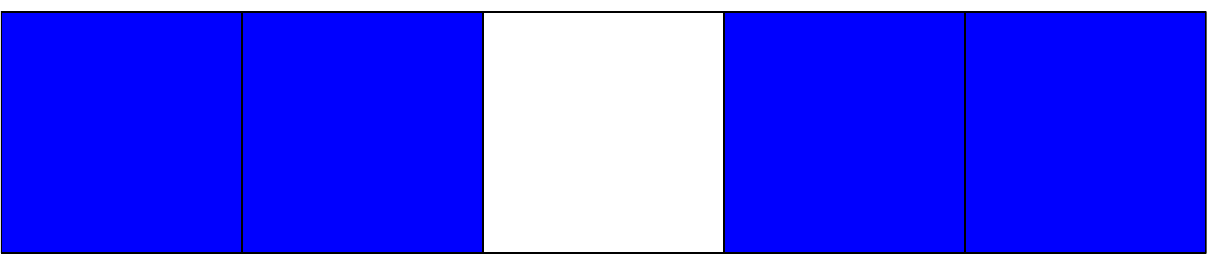} & \includegraphics[width=0.120000\linewidth]{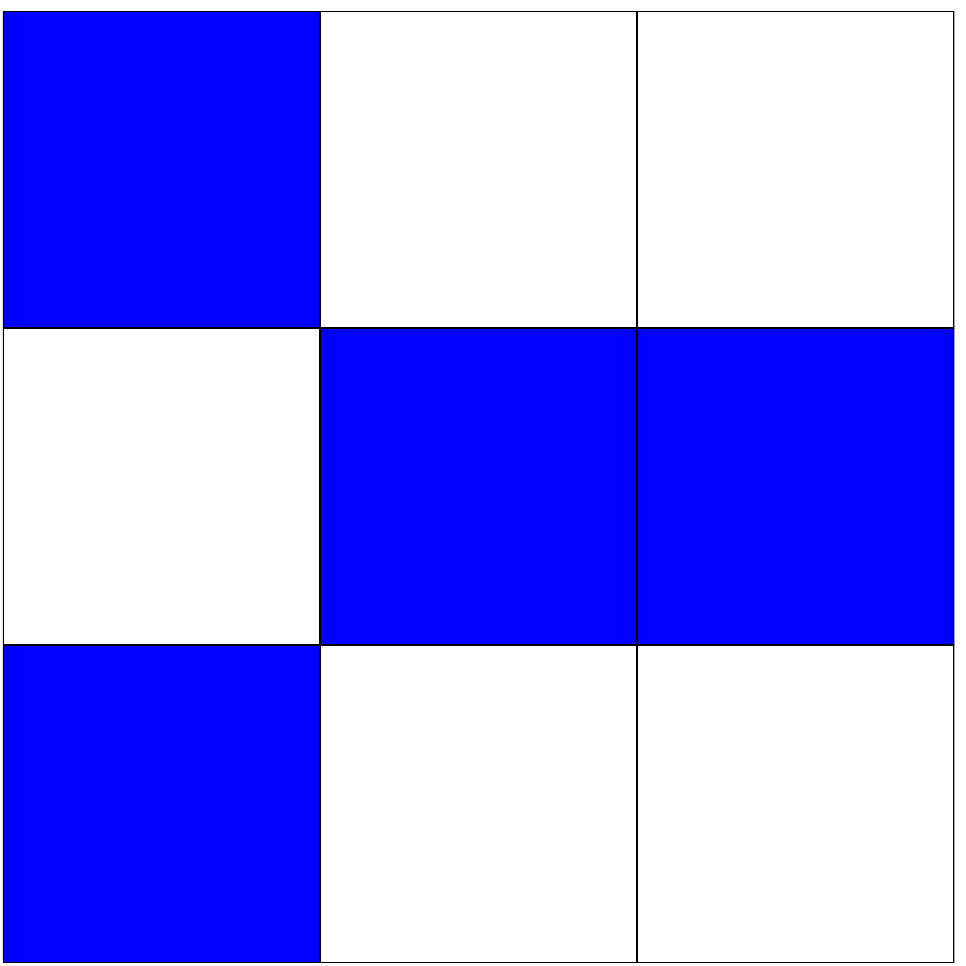} & \includegraphics[width=0.120000\linewidth]{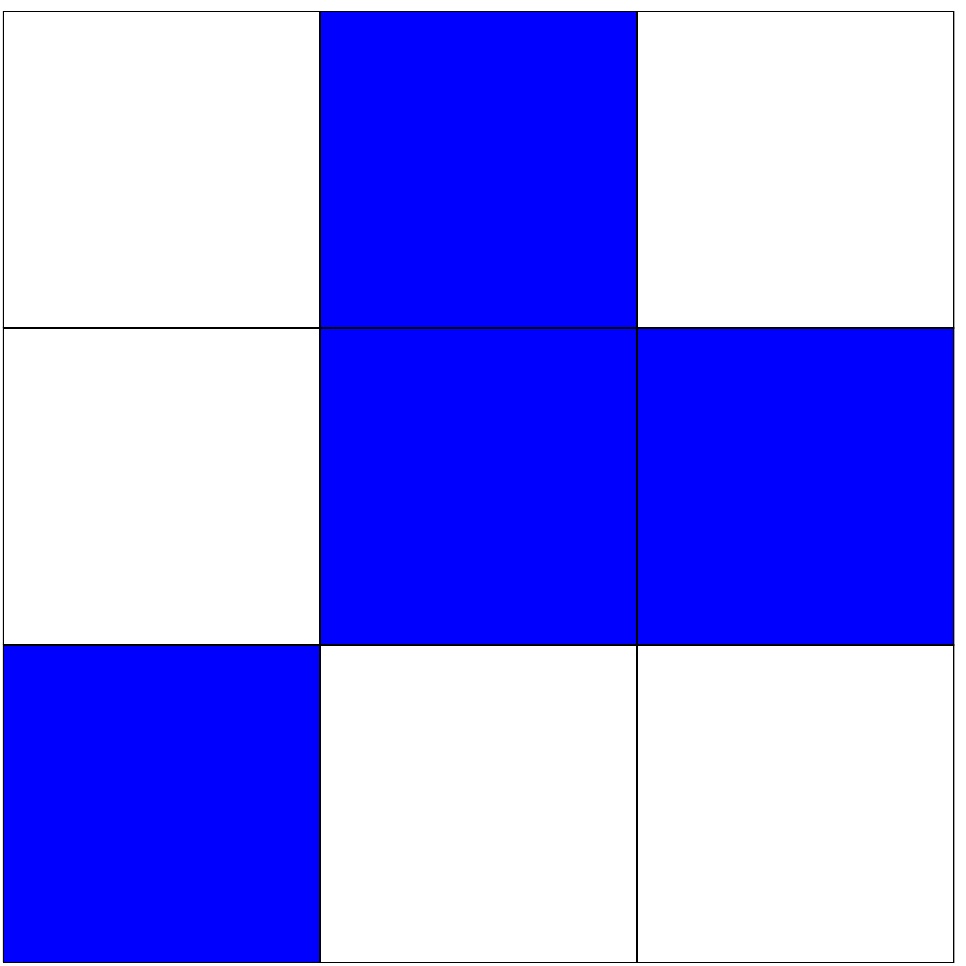} & \includegraphics[width=0.120000\linewidth]{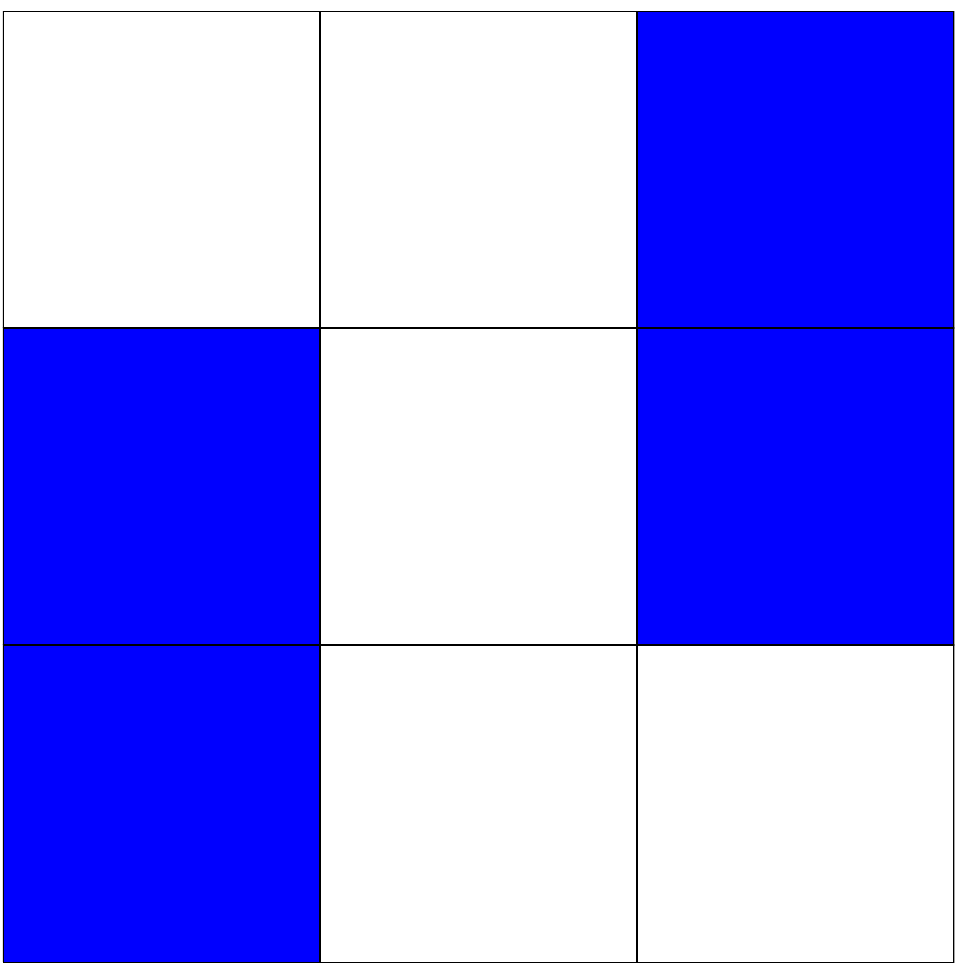} & \includegraphics[width=0.120000\linewidth]{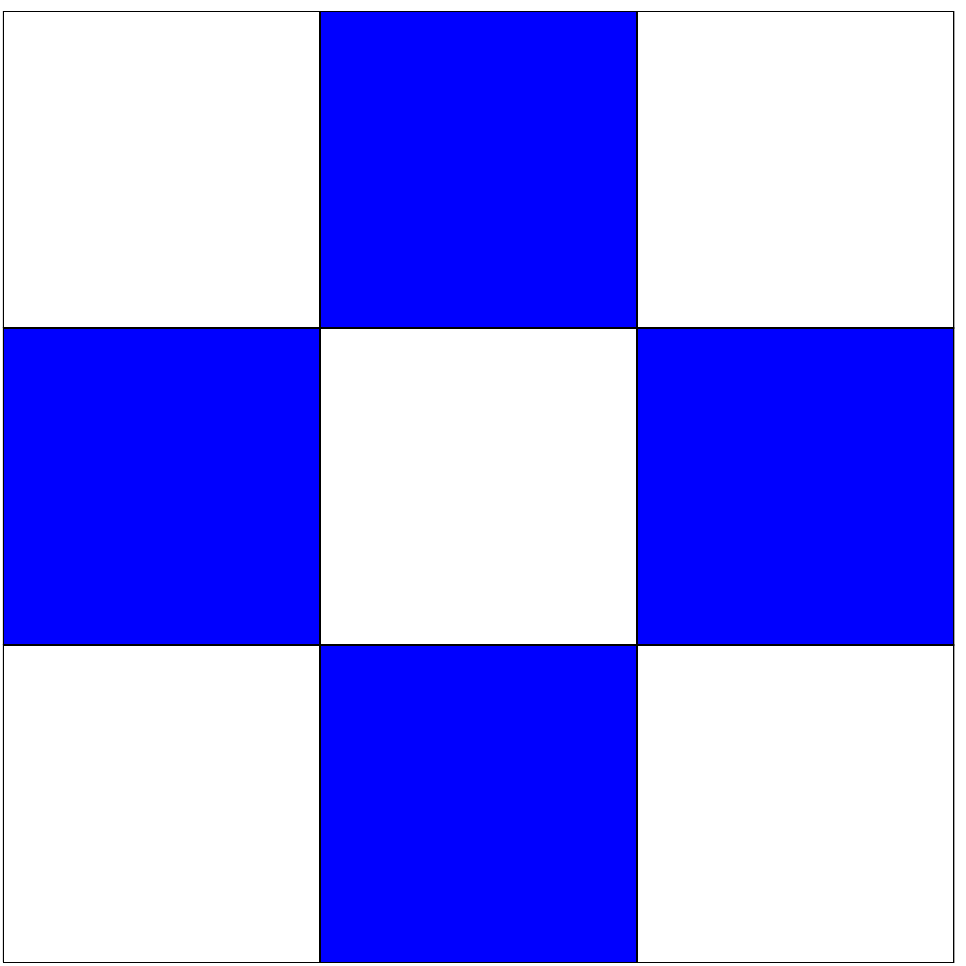} & \includegraphics[width=0.120000\linewidth]{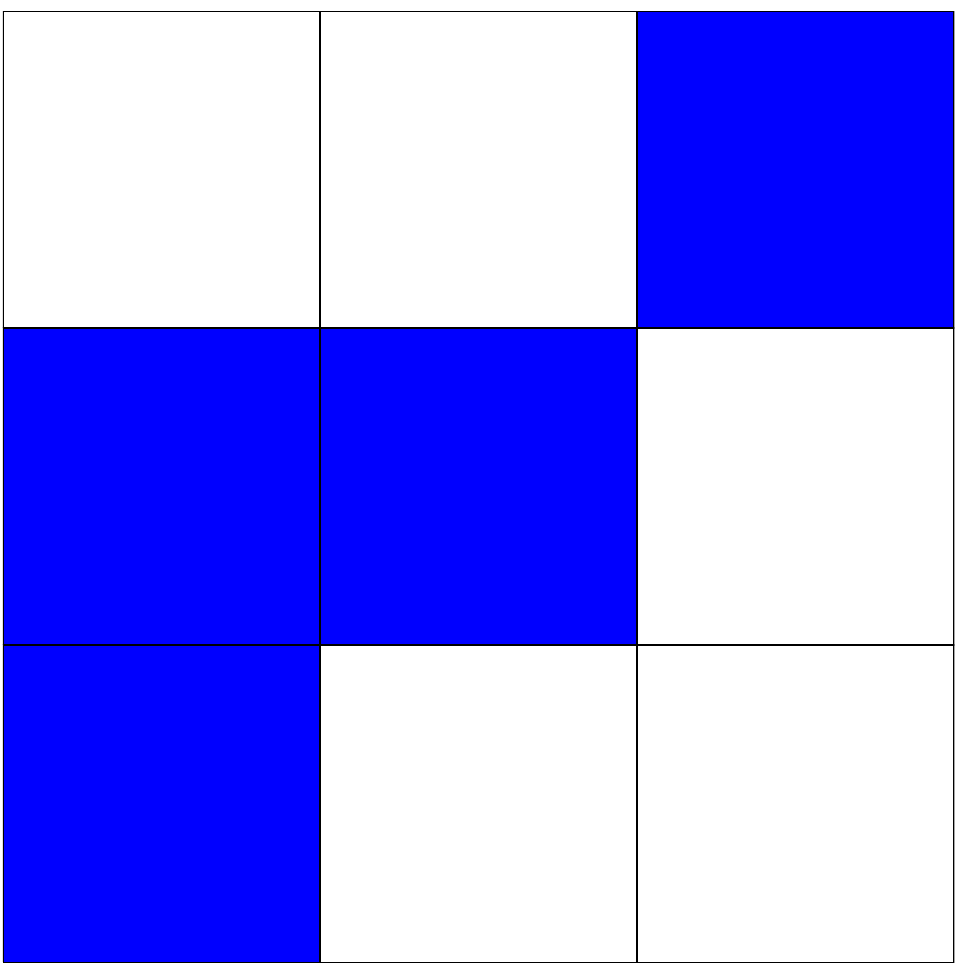} \\ \hline
\end{tabular}
\caption{Impossible pieces for n=4 and k=1.}
\label{tab:impossible-41}
\end{table}

\begin{table}[!htpb]
\centering
\begin{tabular}{|c|c|c|c|c|}
\hline
& & & & \\ 
\includegraphics[width=0.160000\linewidth]{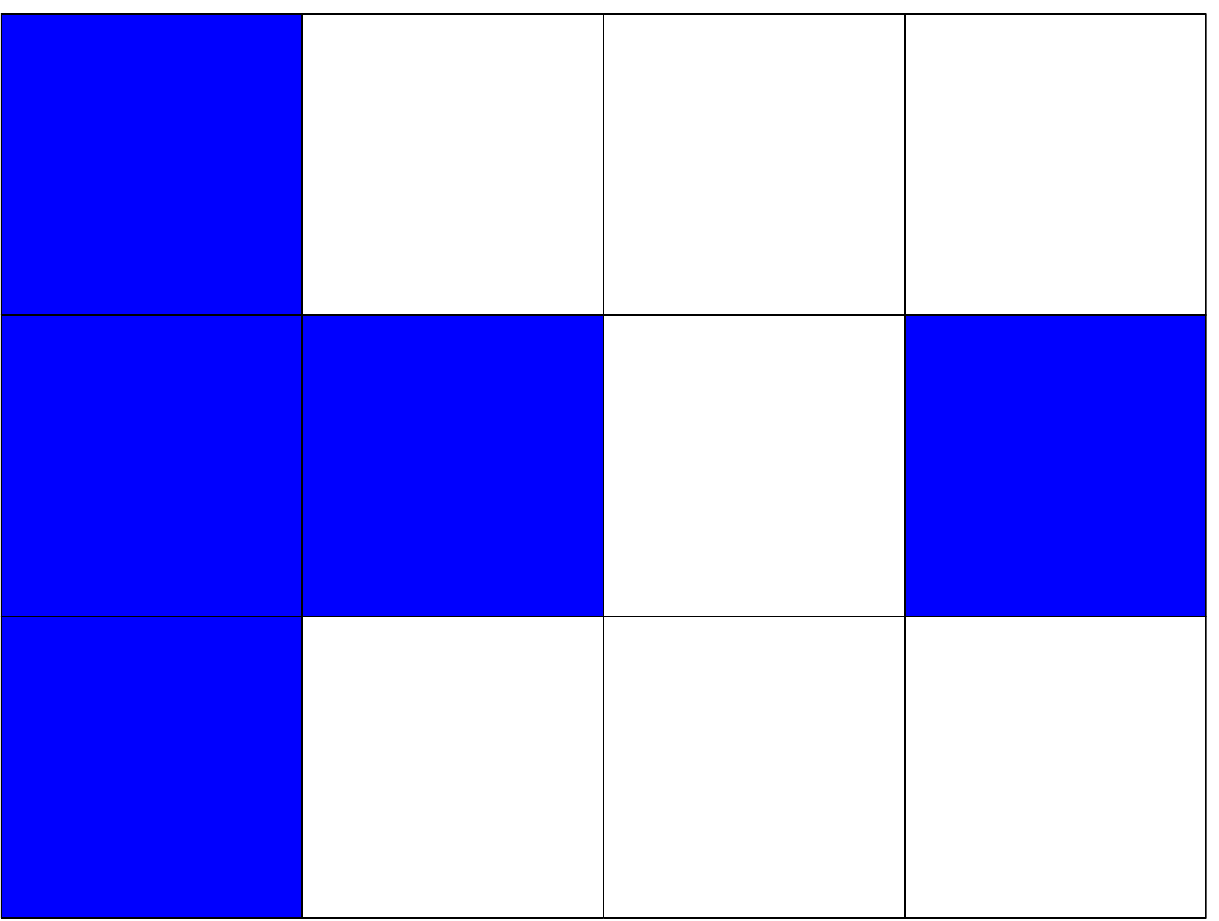} & \includegraphics[width=0.160000\linewidth]{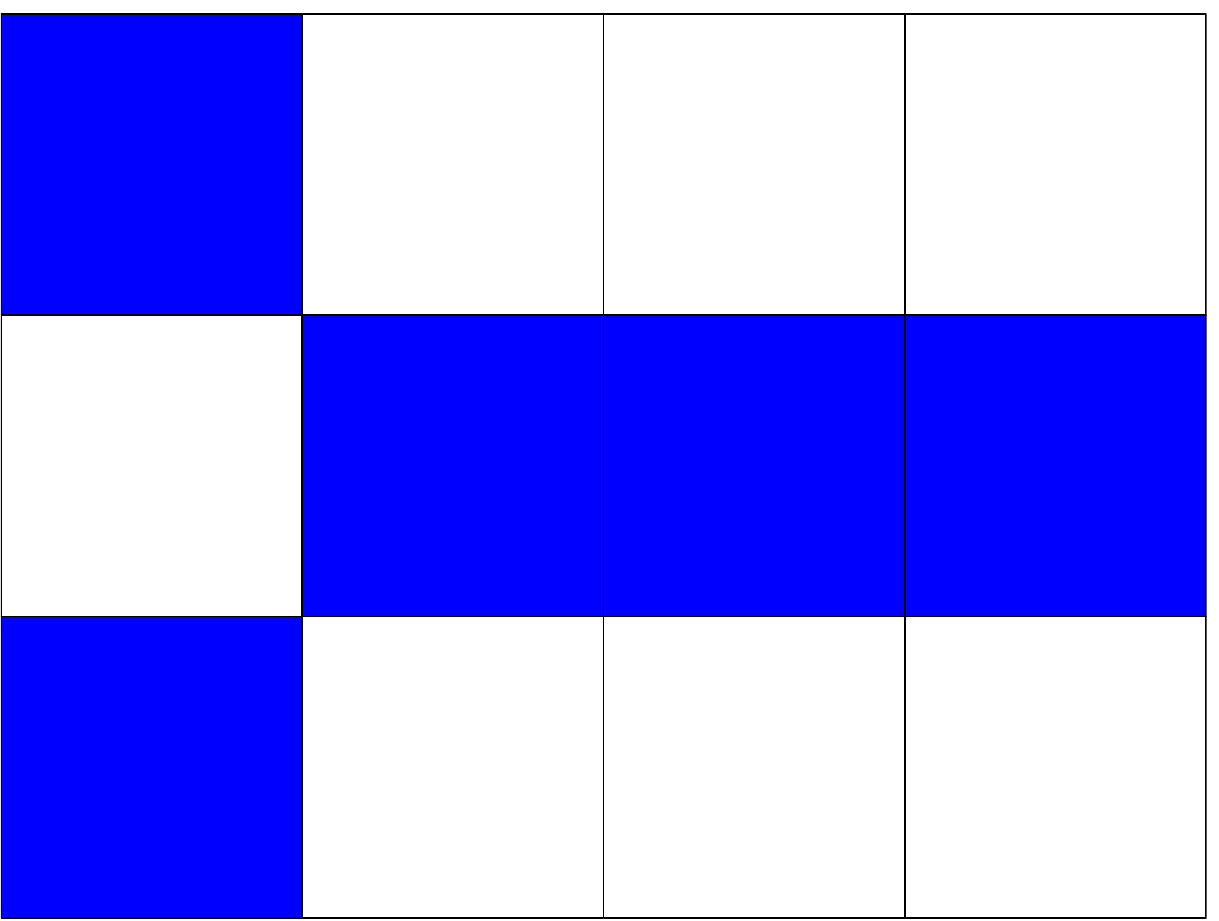} & \includegraphics[width=0.160000\linewidth]{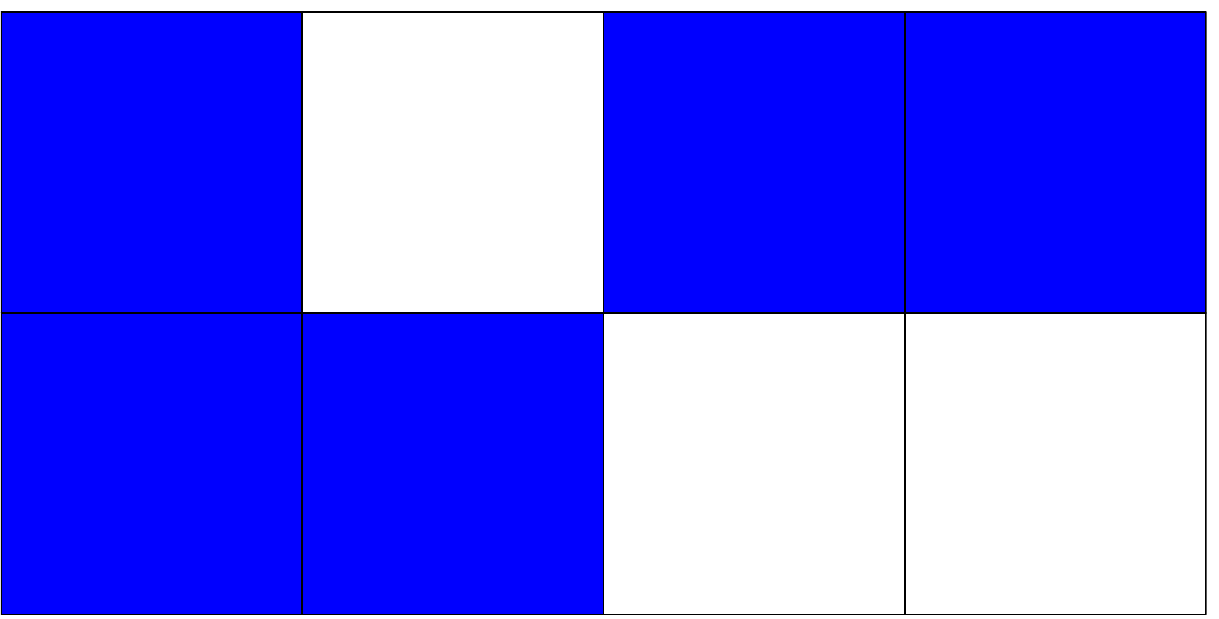} & \includegraphics[width=0.160000\linewidth]{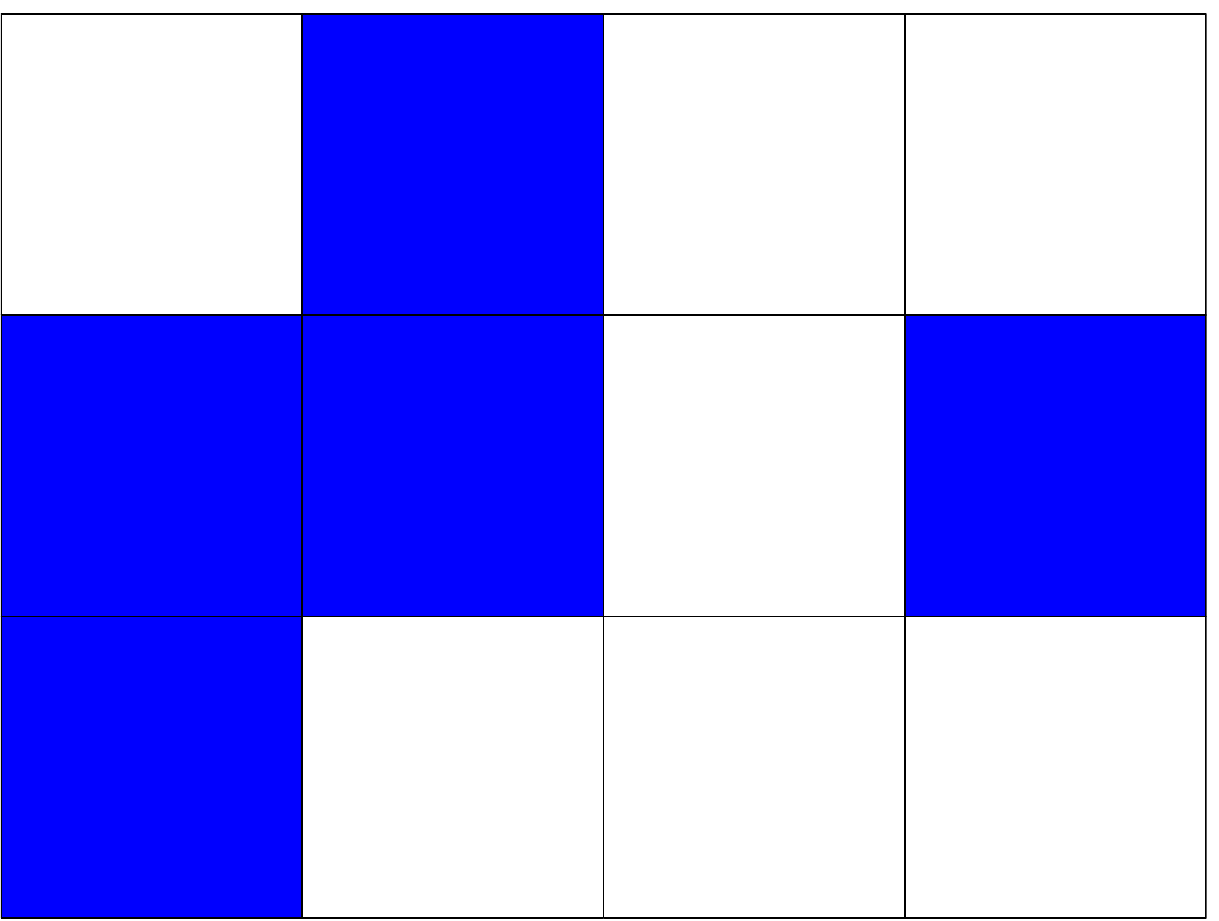} & \includegraphics[width=0.160000\linewidth]{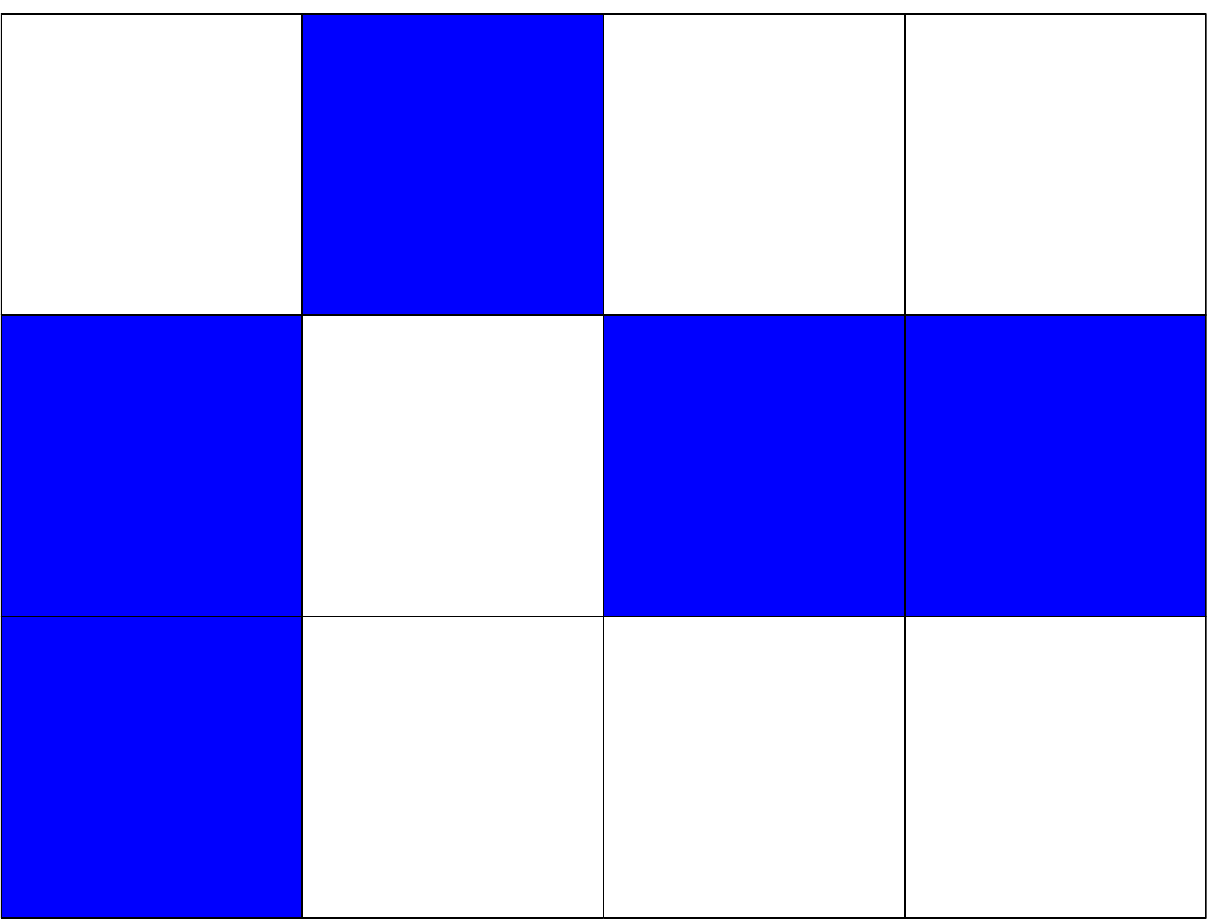} \\ \hline
& & & & \\ 
\includegraphics[width=0.160000\linewidth]{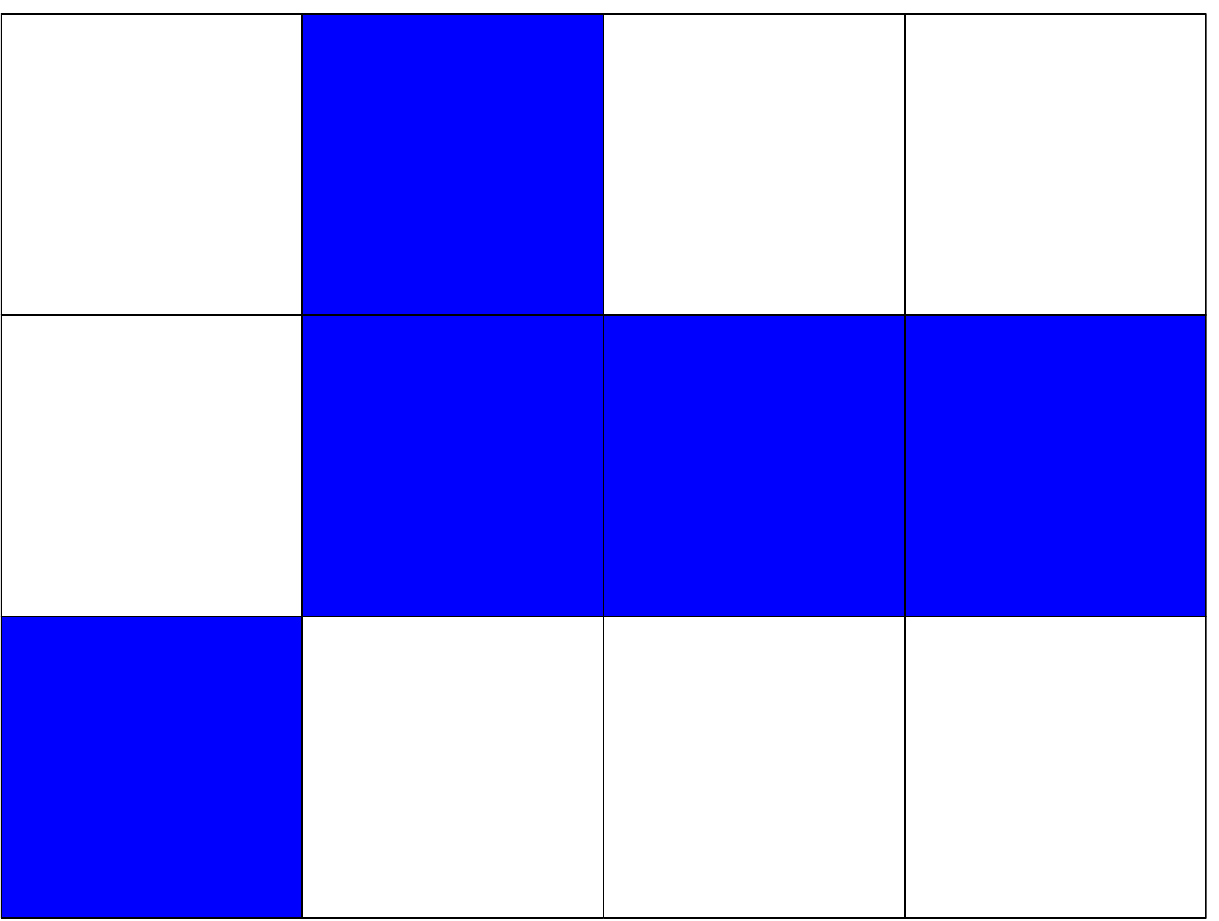} & \includegraphics[width=0.160000\linewidth]{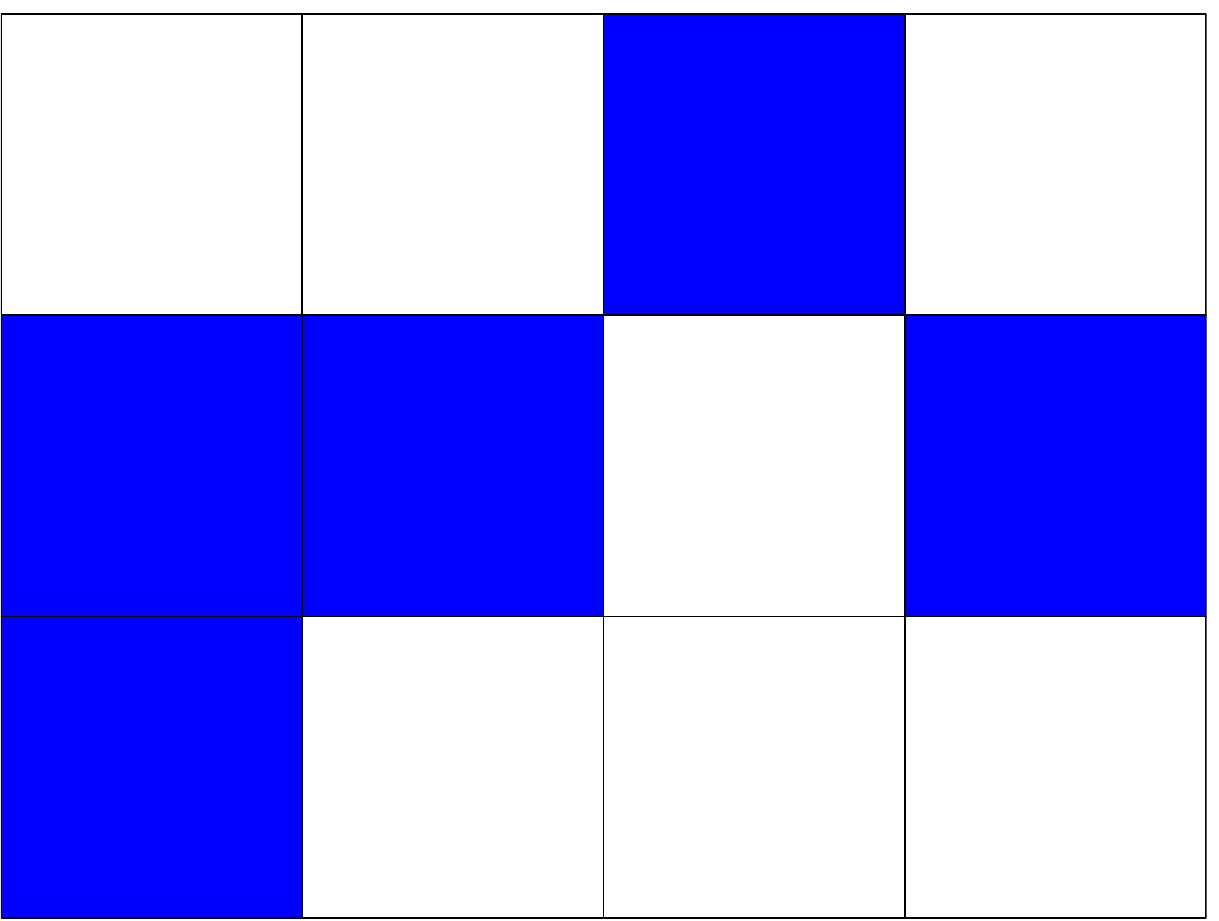} & \includegraphics[width=0.160000\linewidth]{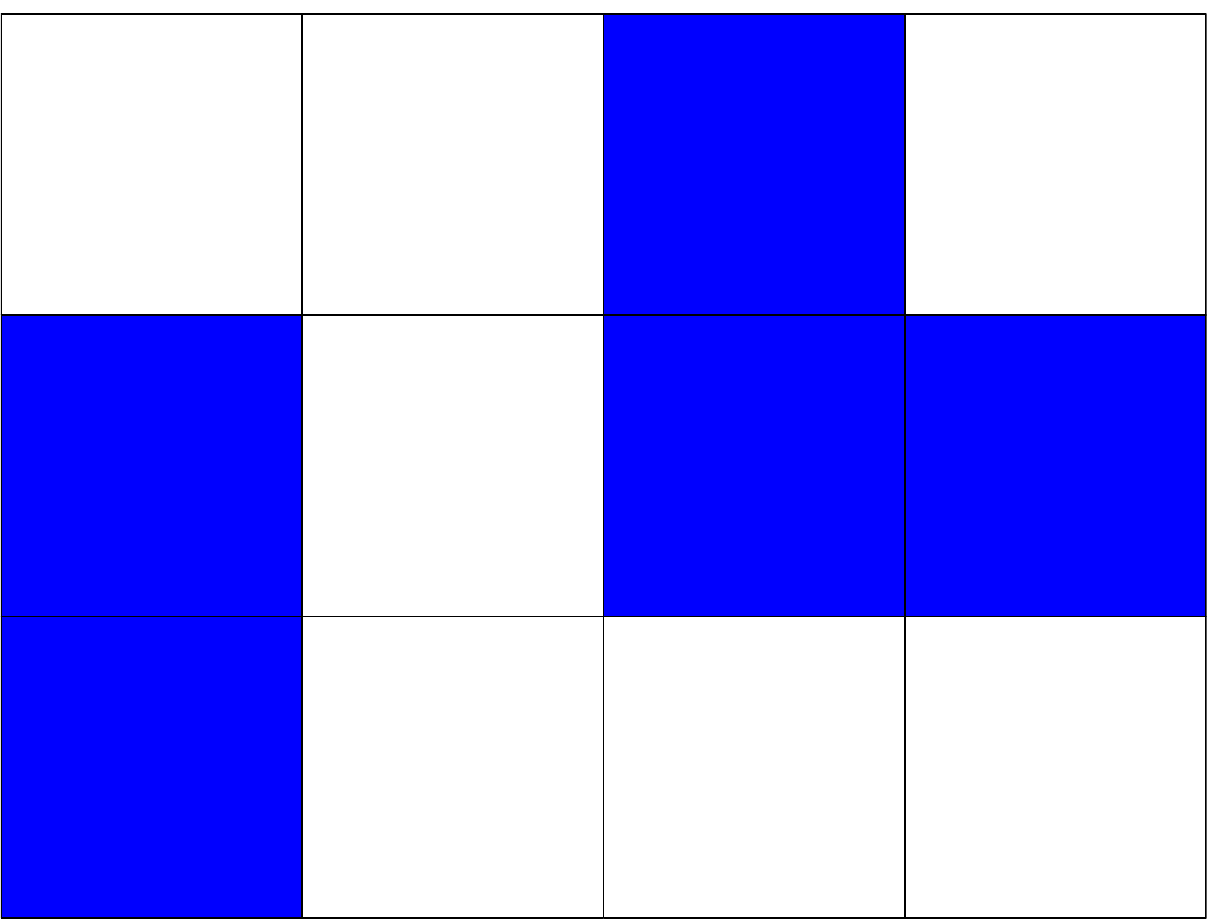} & \includegraphics[width=0.160000\linewidth]{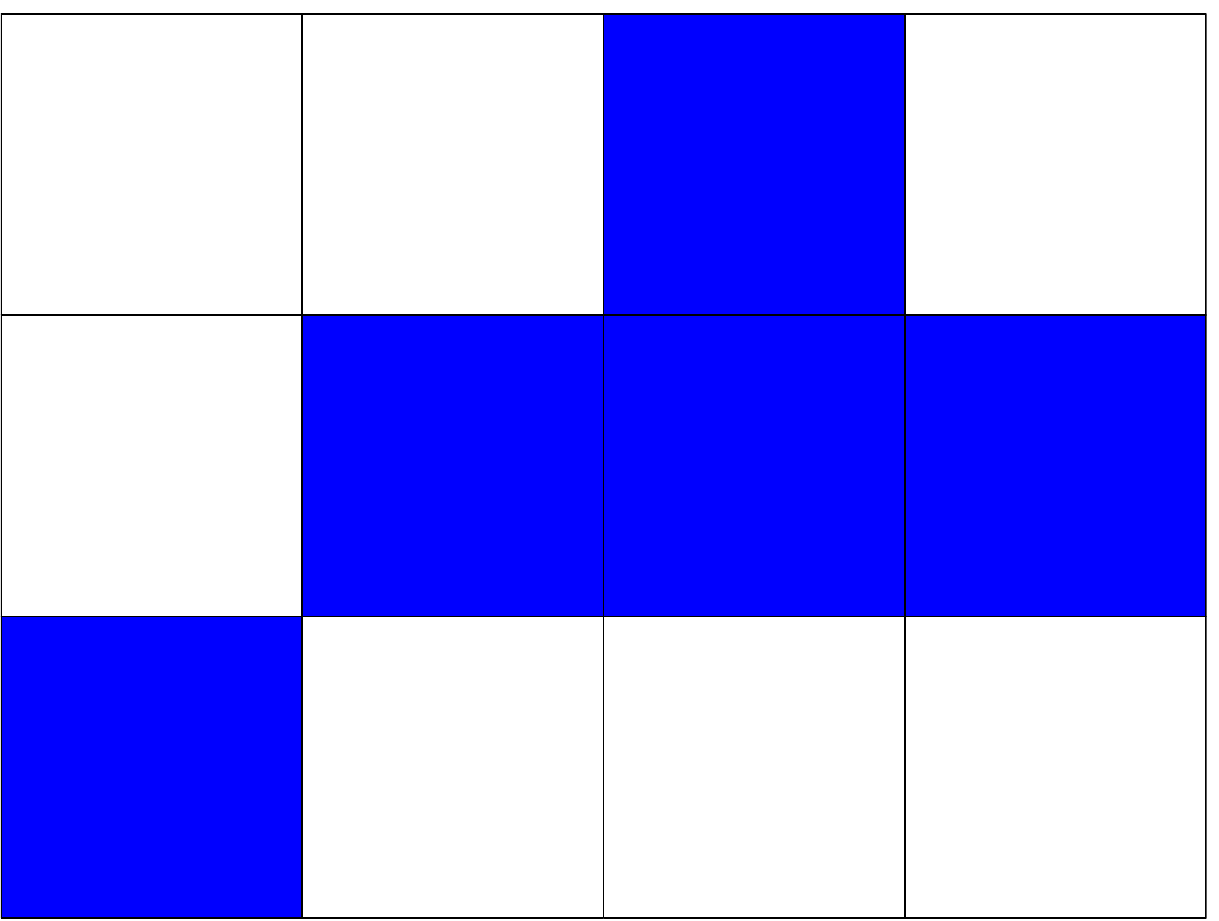} & \includegraphics[width=0.160000\linewidth]{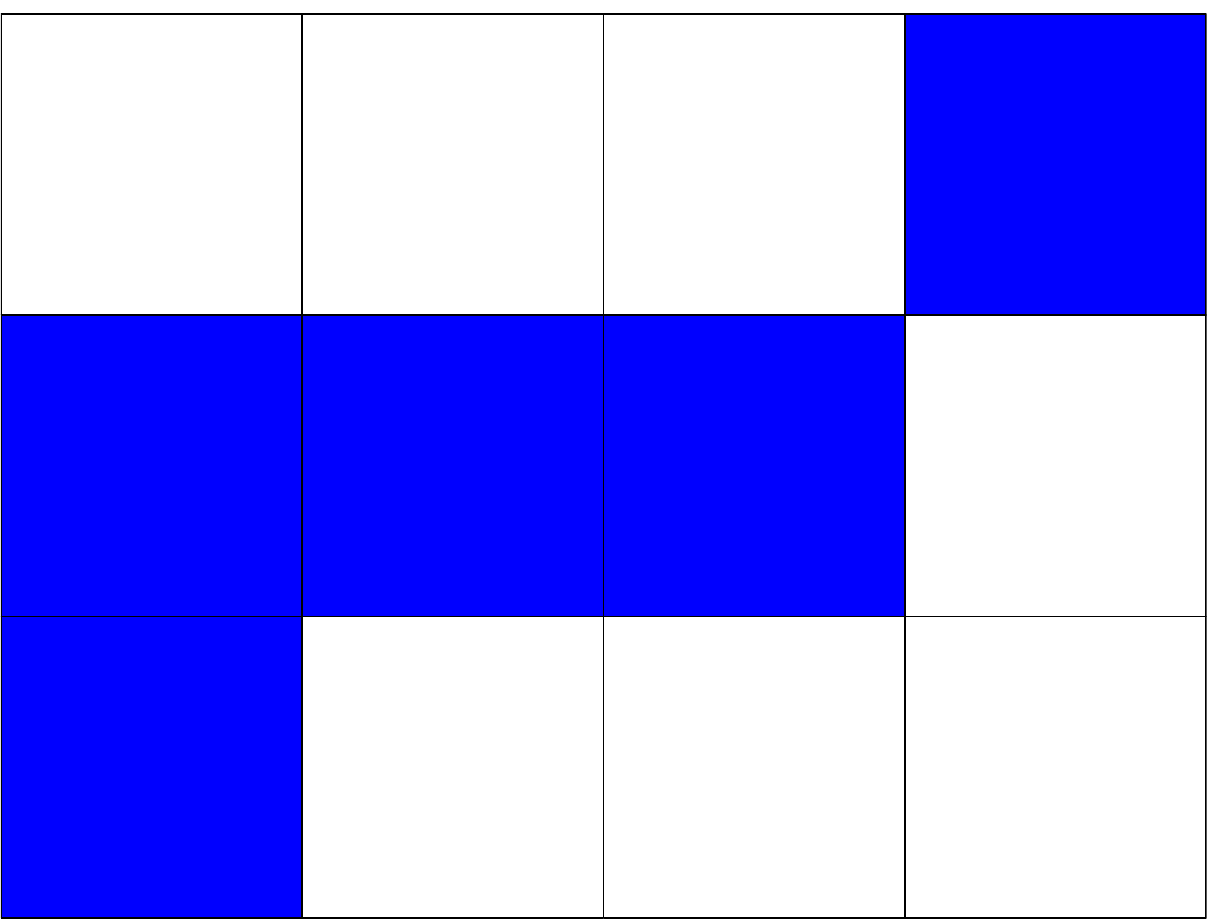} \\ \hline
& & & & \\ 
\includegraphics[width=0.160000\linewidth]{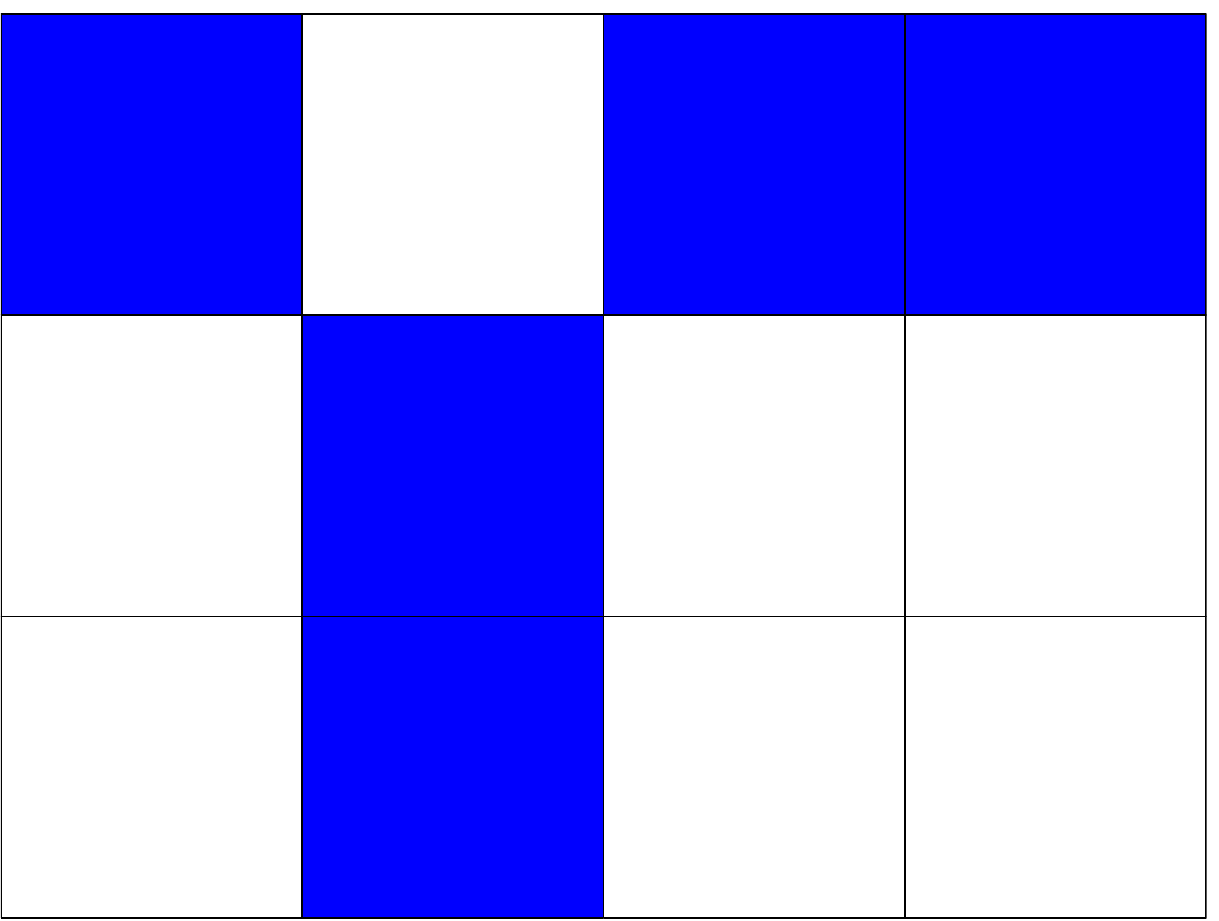} & \includegraphics[width=0.160000\linewidth]{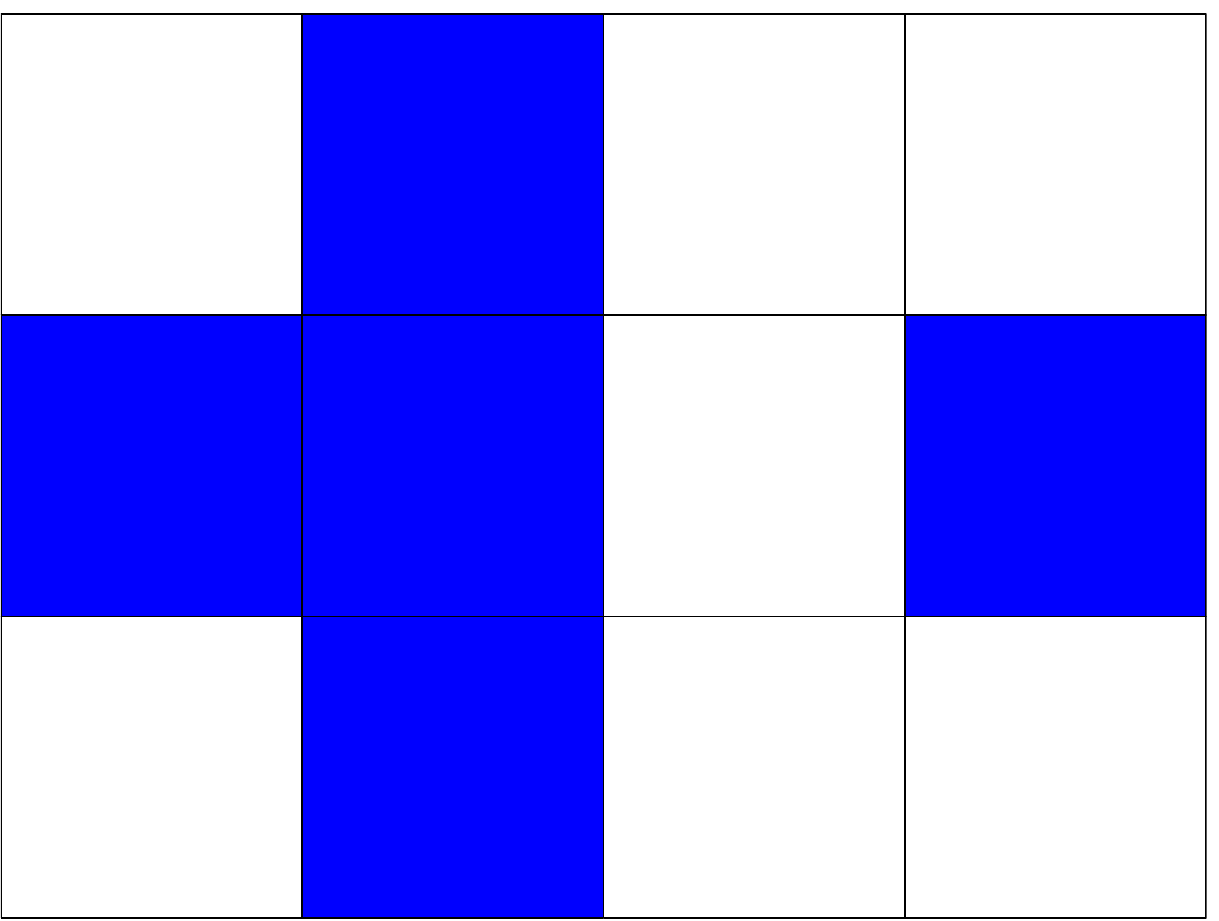} & \includegraphics[width=0.160000\linewidth]{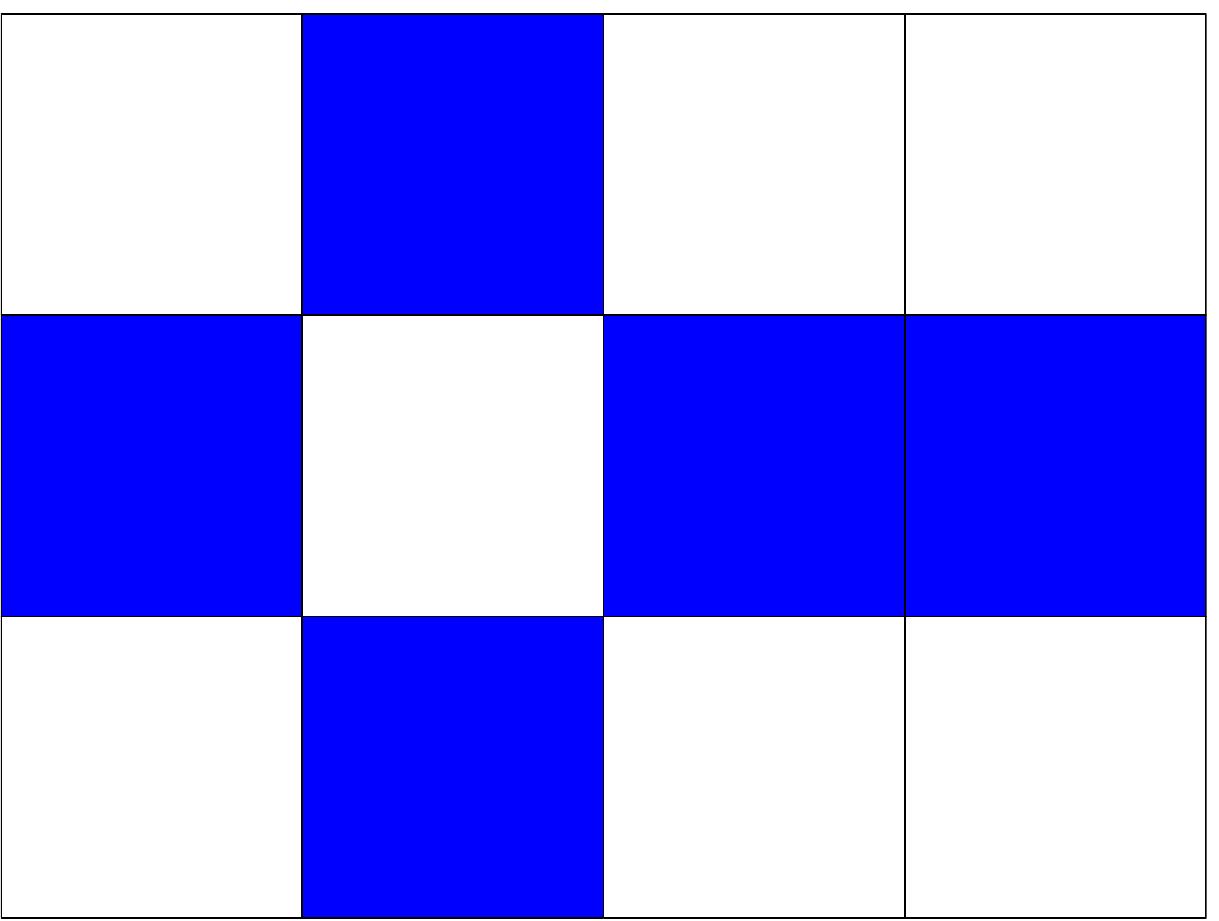} & \includegraphics[width=0.160000\linewidth]{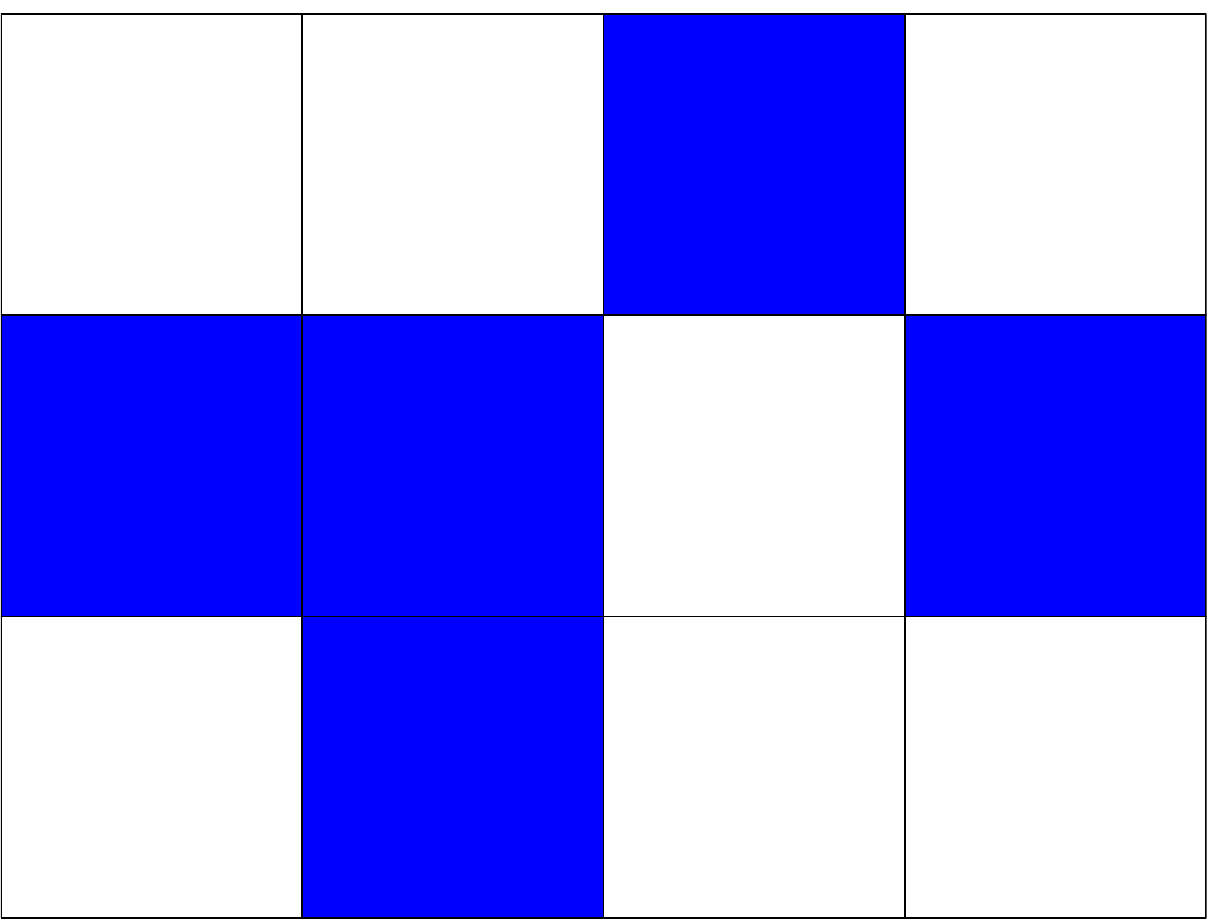} & \includegraphics[width=0.160000\linewidth]{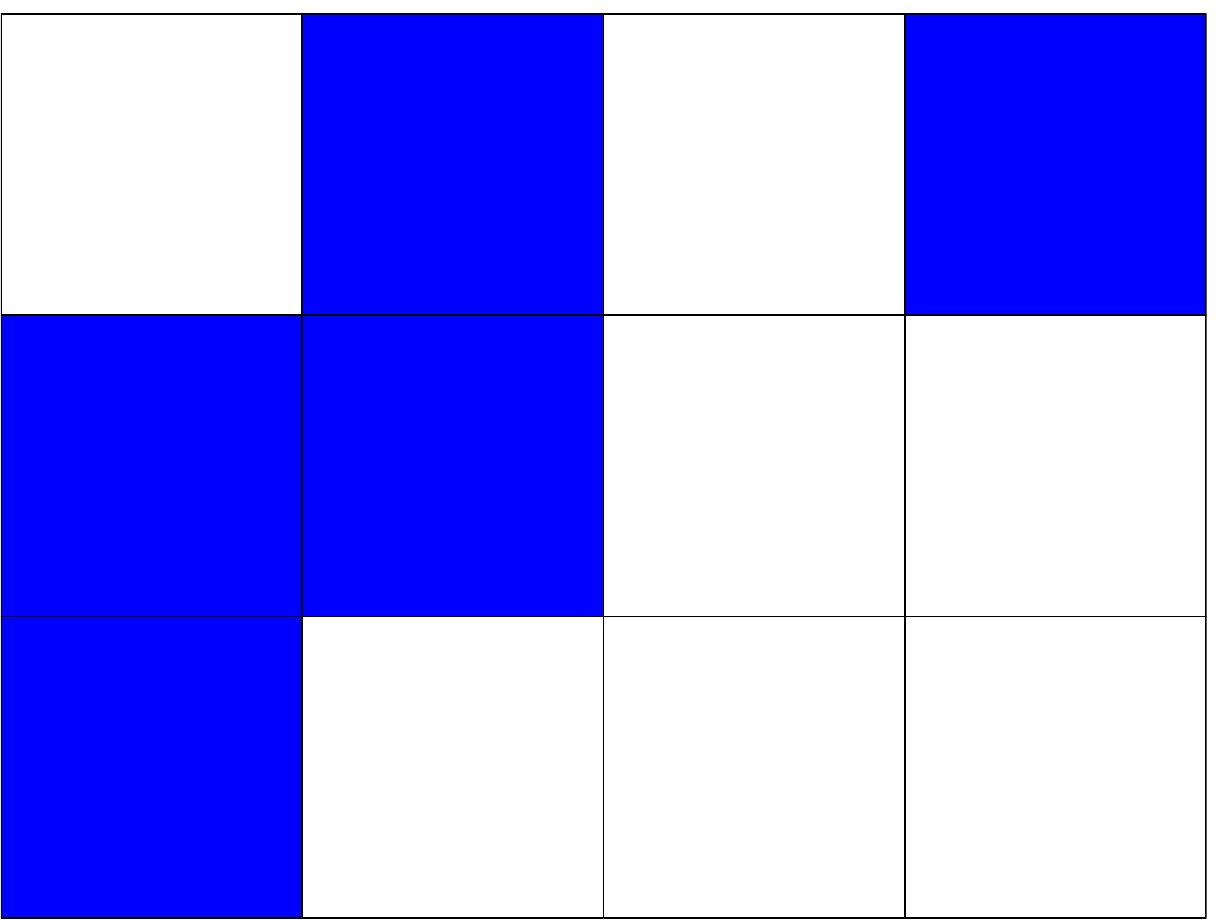} \\ \hline
& & & & \\ 
\includegraphics[width=0.160000\linewidth]{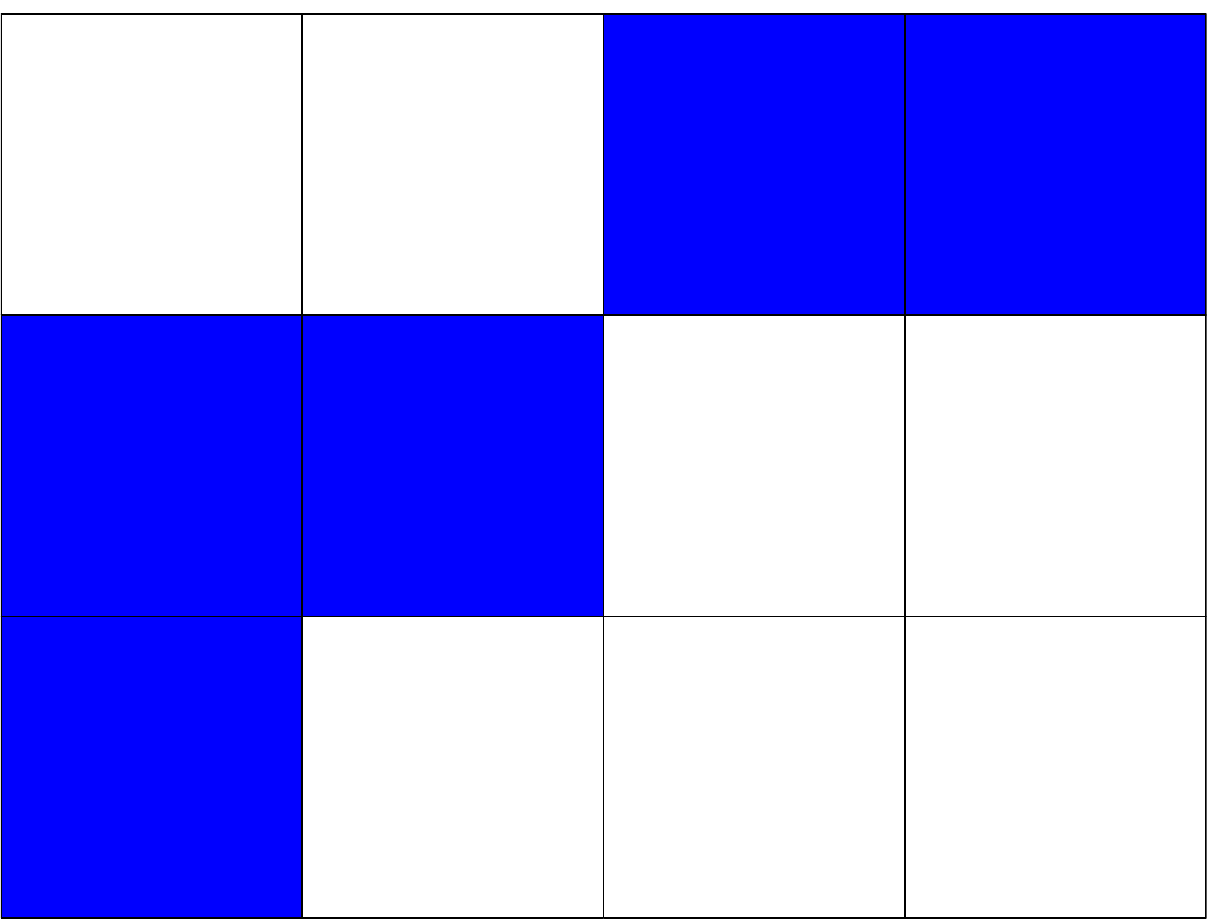} & \includegraphics[width=0.160000\linewidth]{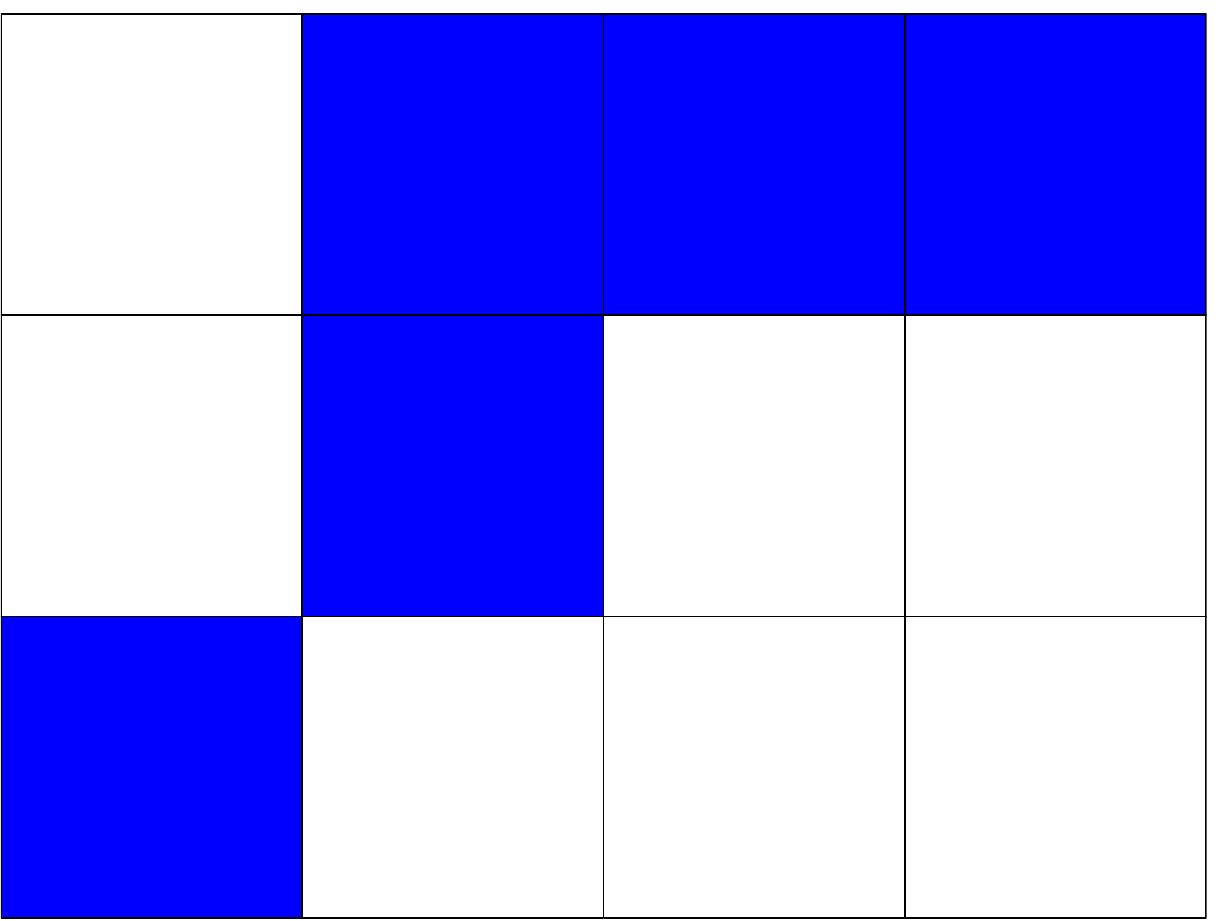} & \includegraphics[width=0.160000\linewidth]{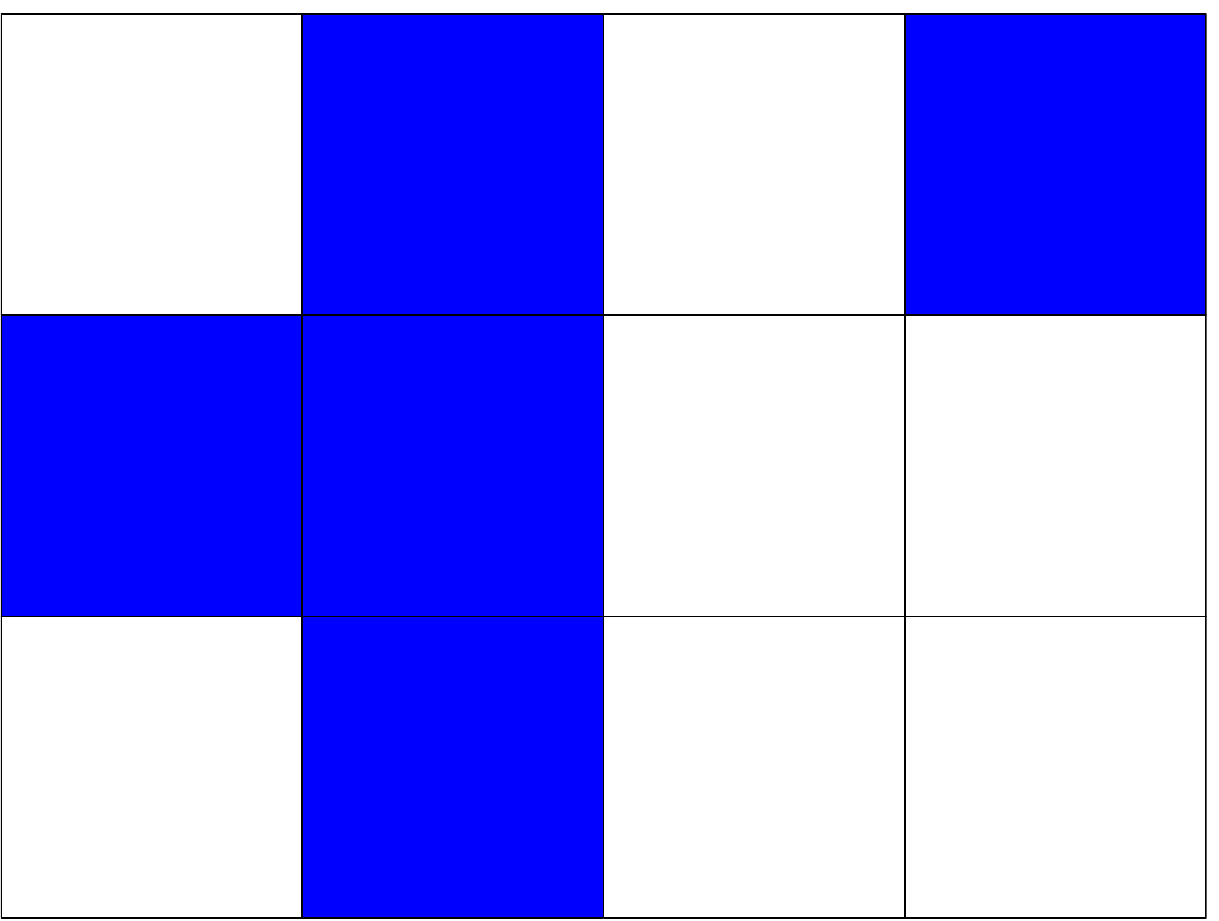} & \includegraphics[width=0.160000\linewidth]{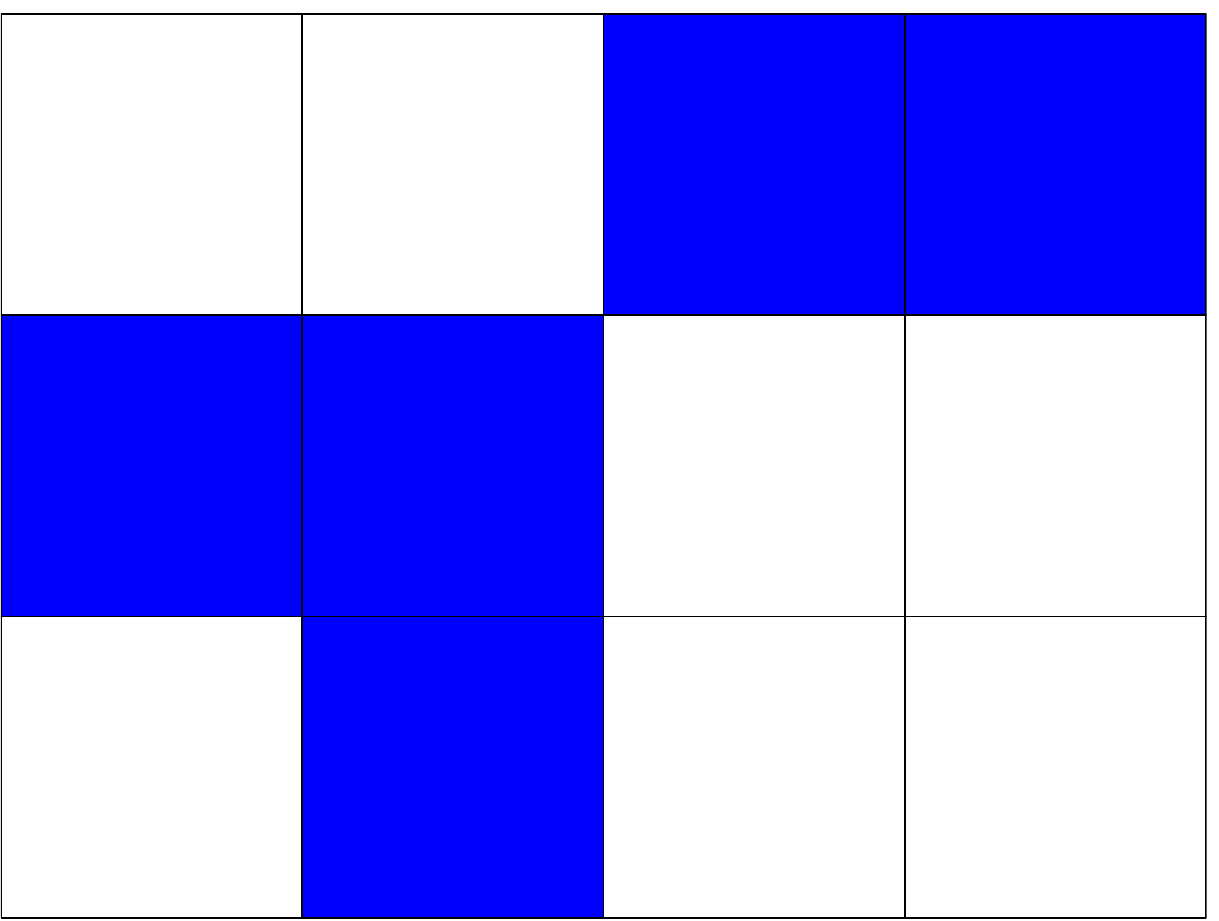} & \includegraphics[width=0.160000\linewidth]{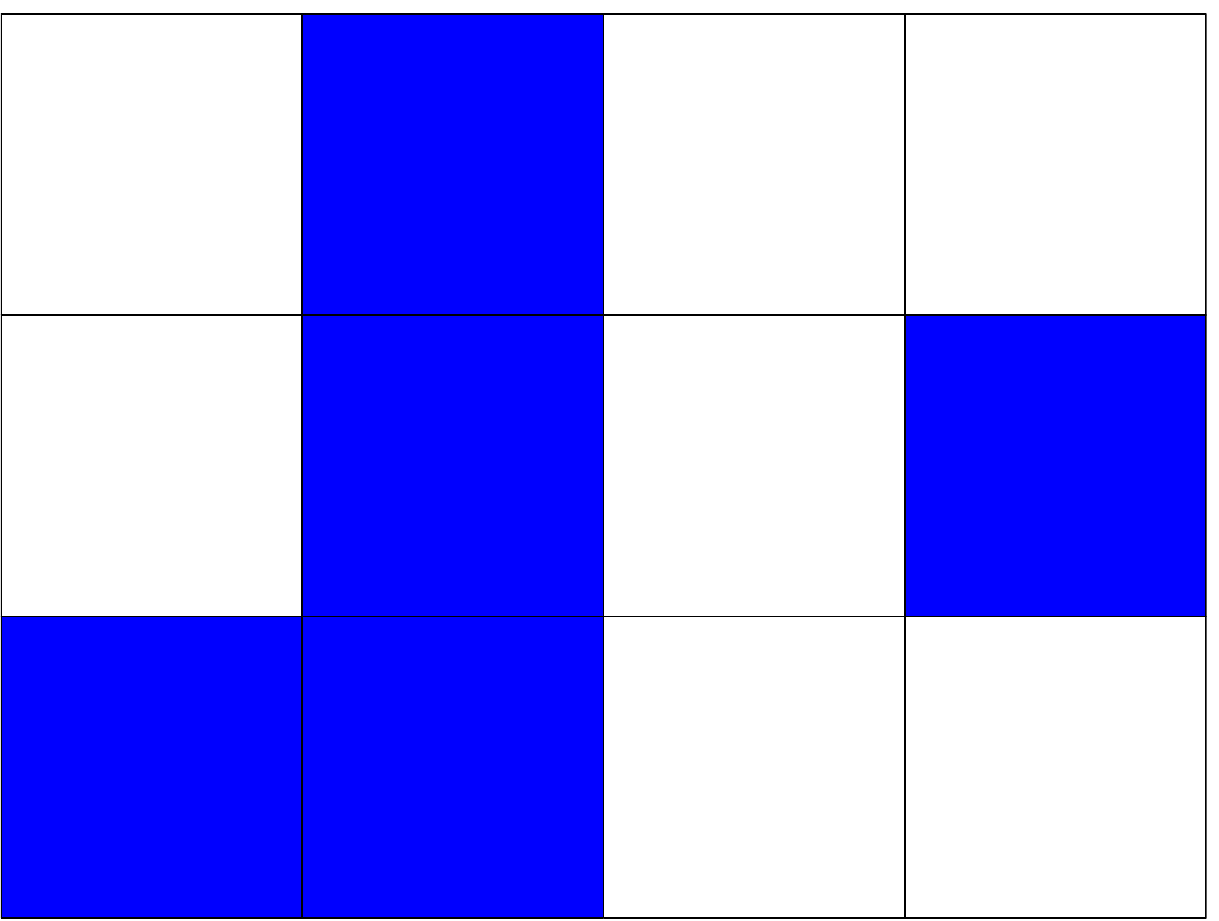} \\ \hline
& & & & \\ 
\includegraphics[width=0.160000\linewidth]{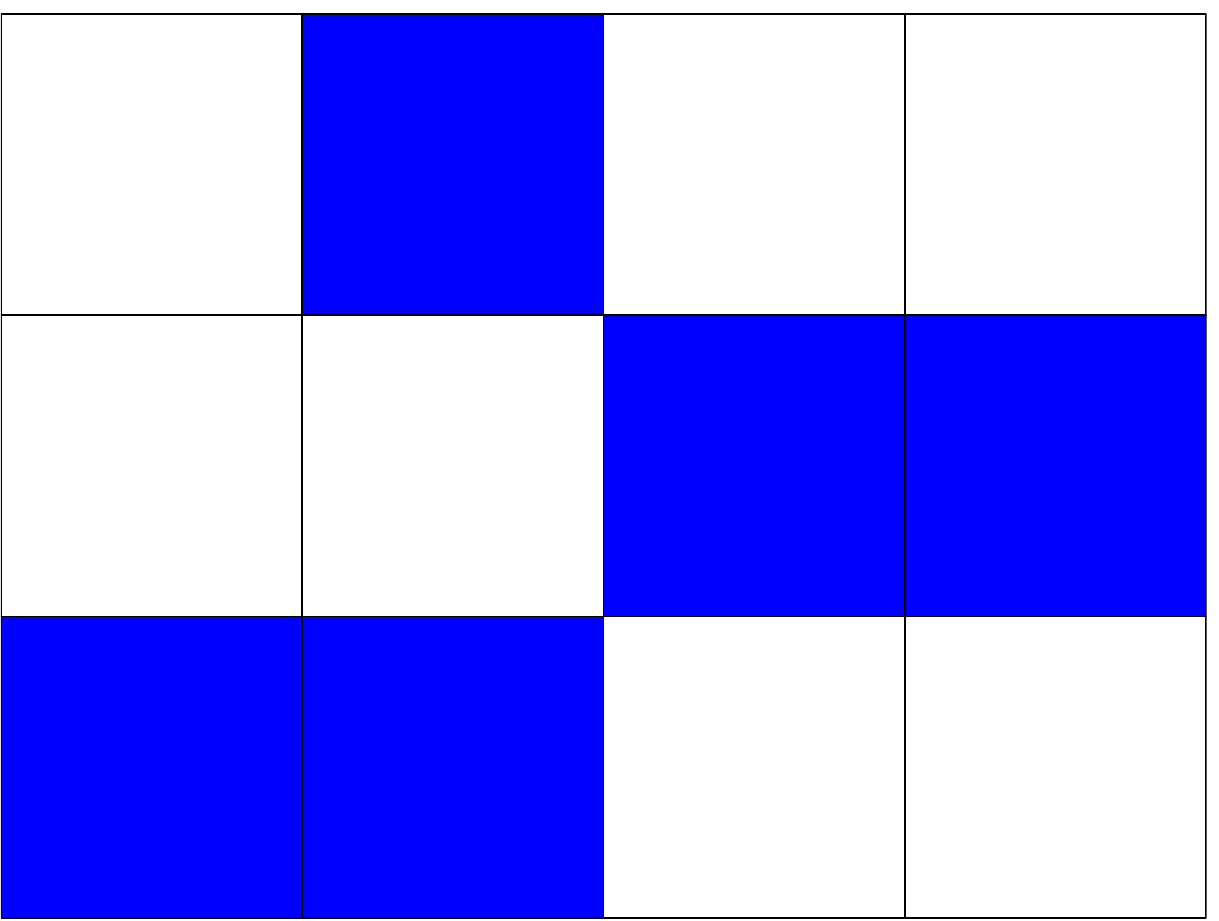} & \includegraphics[width=0.160000\linewidth]{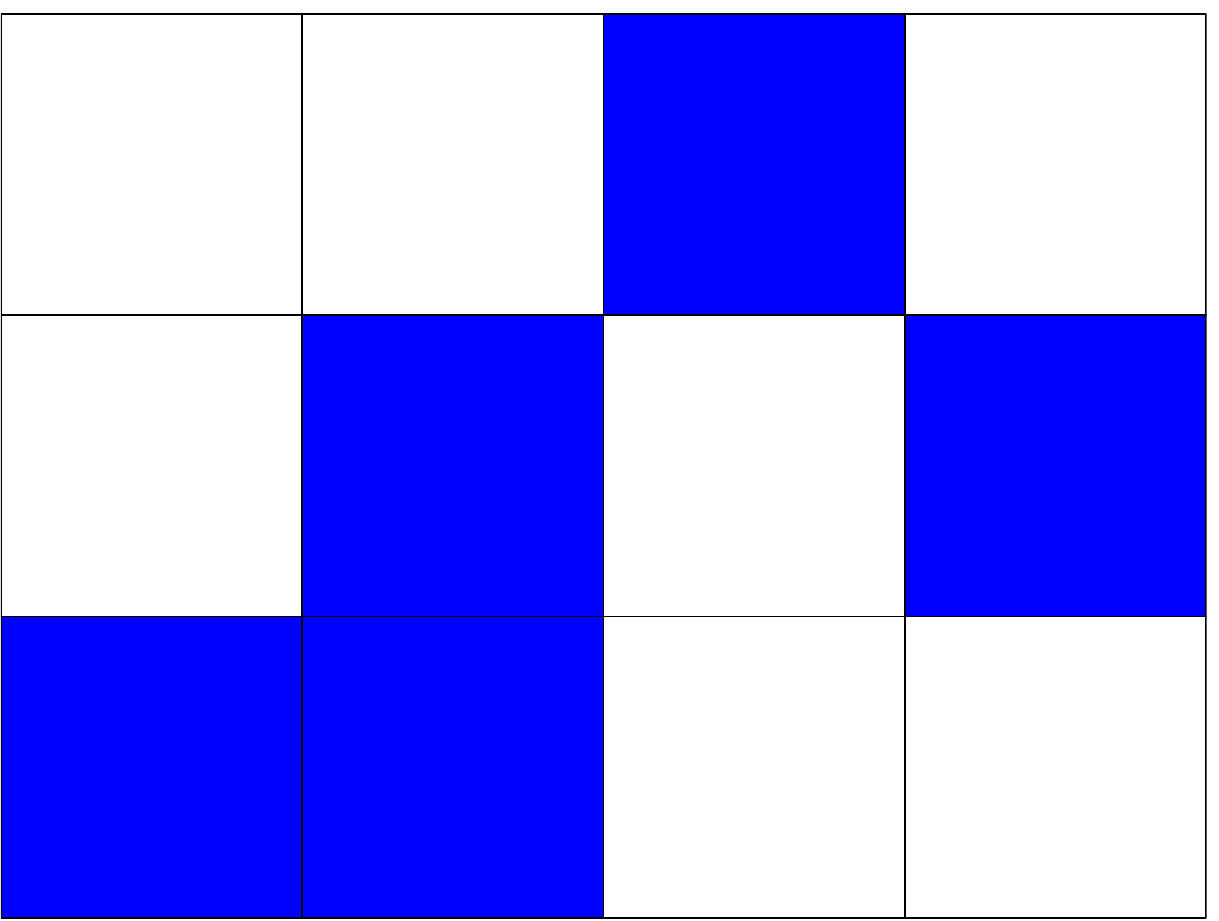} & \includegraphics[width=0.160000\linewidth]{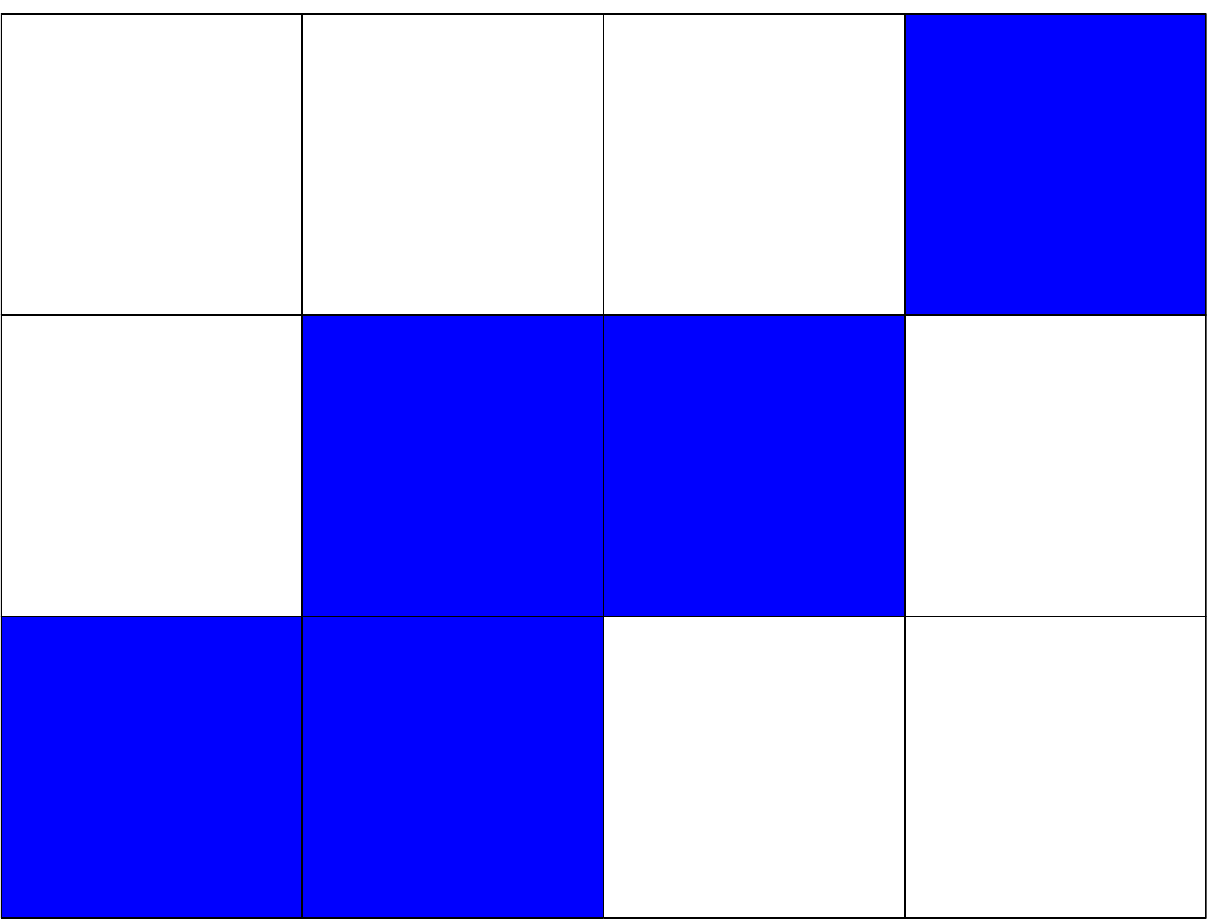} & \includegraphics[width=0.160000\linewidth]{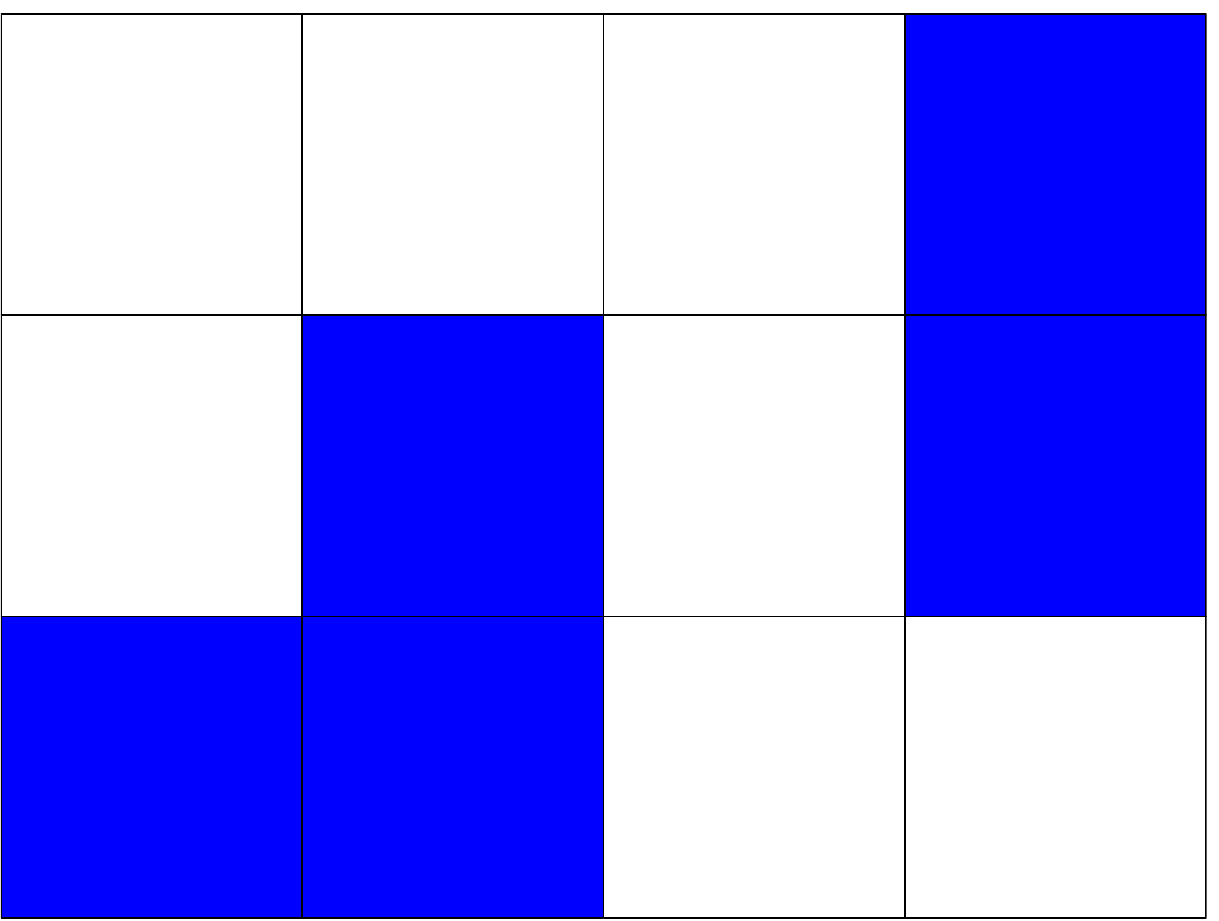} & \includegraphics[width=0.160000\linewidth]{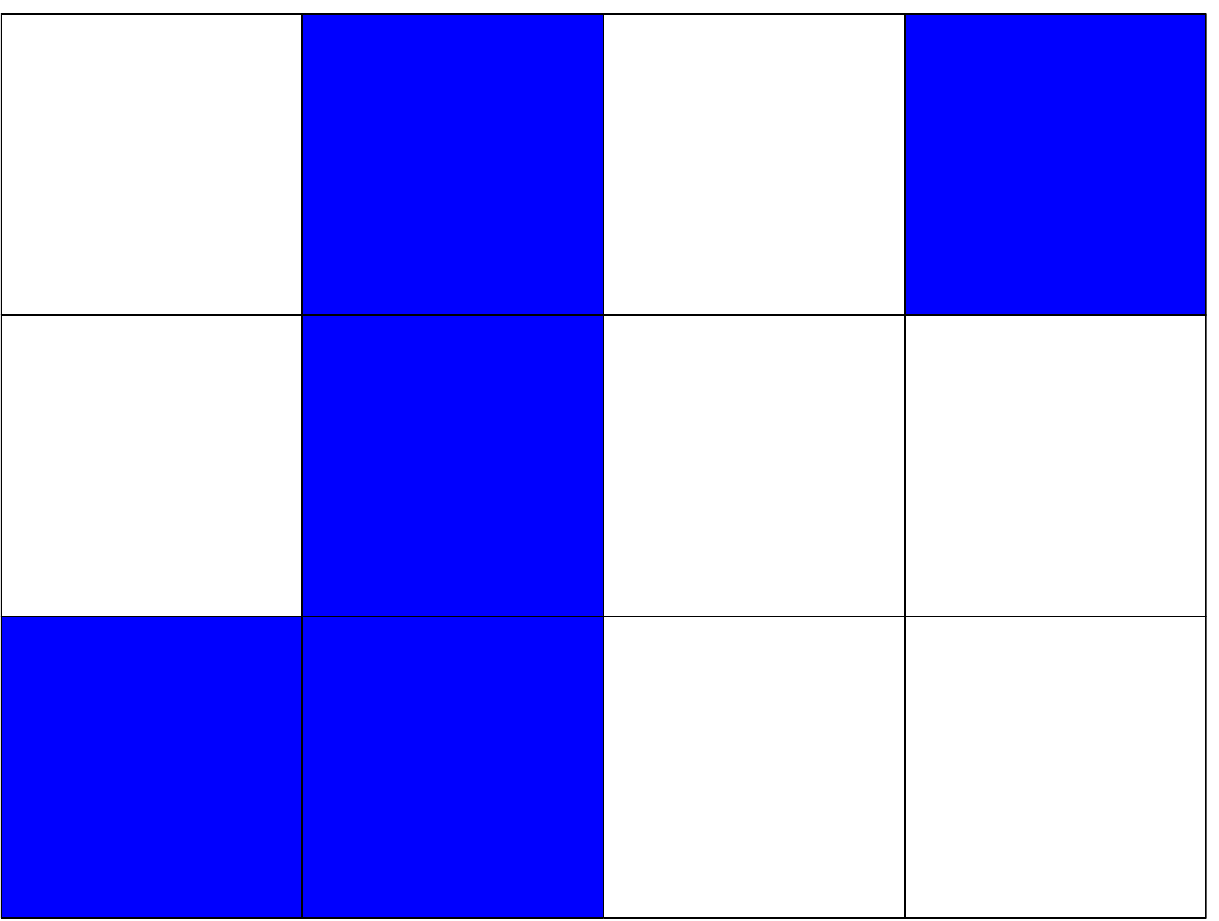} \\ \hline
& & & & \\ 
\includegraphics[width=0.160000\linewidth]{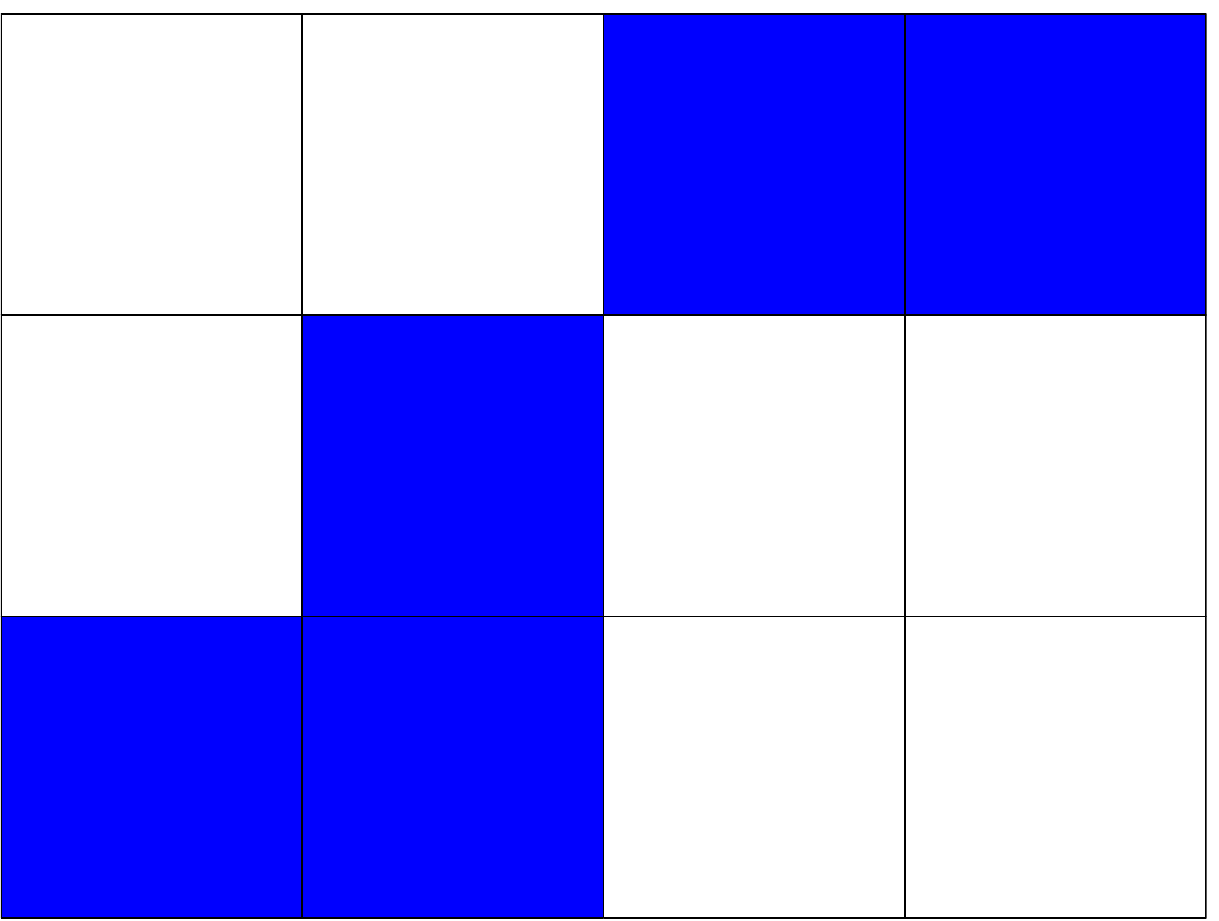} & \includegraphics[width=0.120000\linewidth]{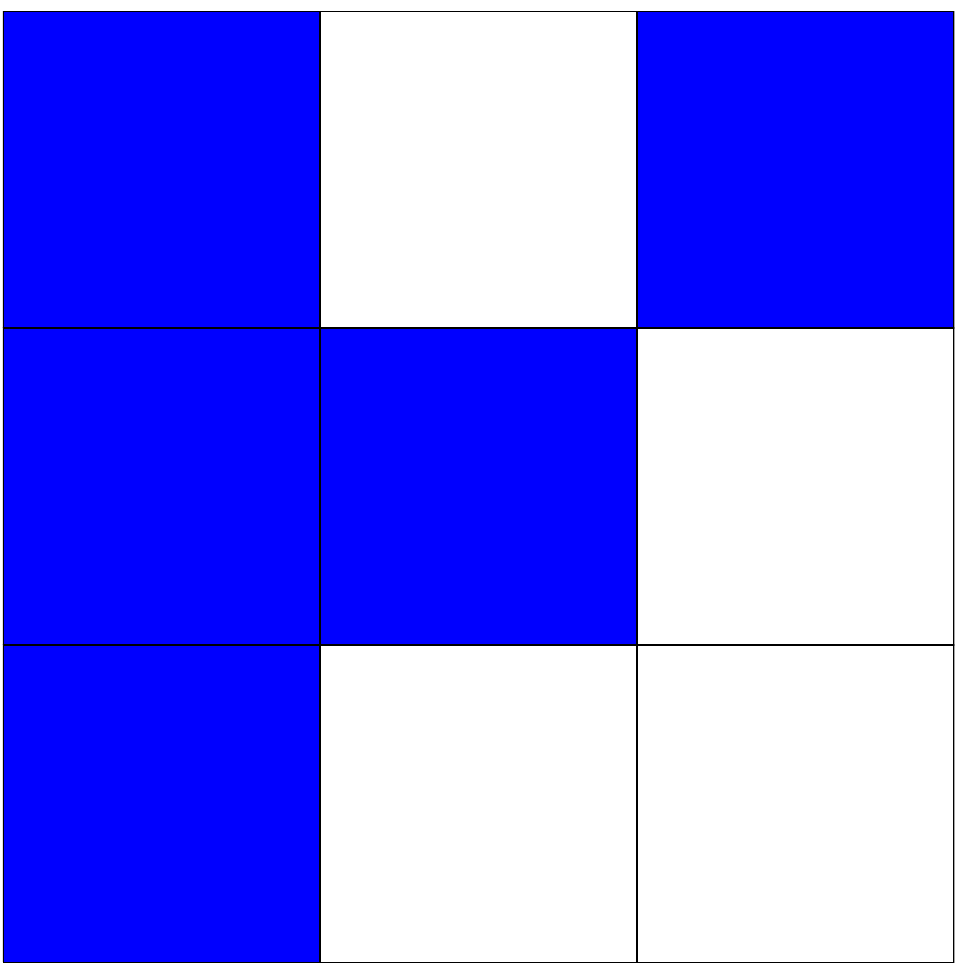} & \includegraphics[width=0.120000\linewidth]{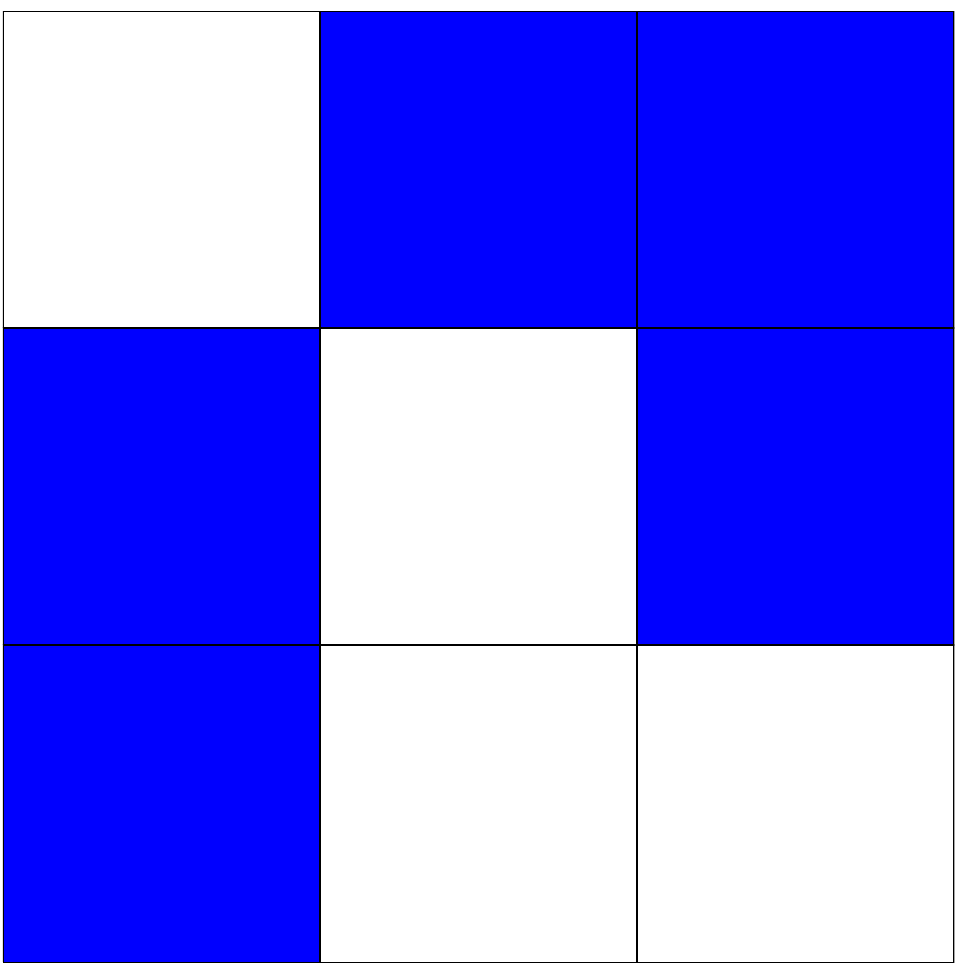} & \includegraphics[width=0.120000\linewidth]{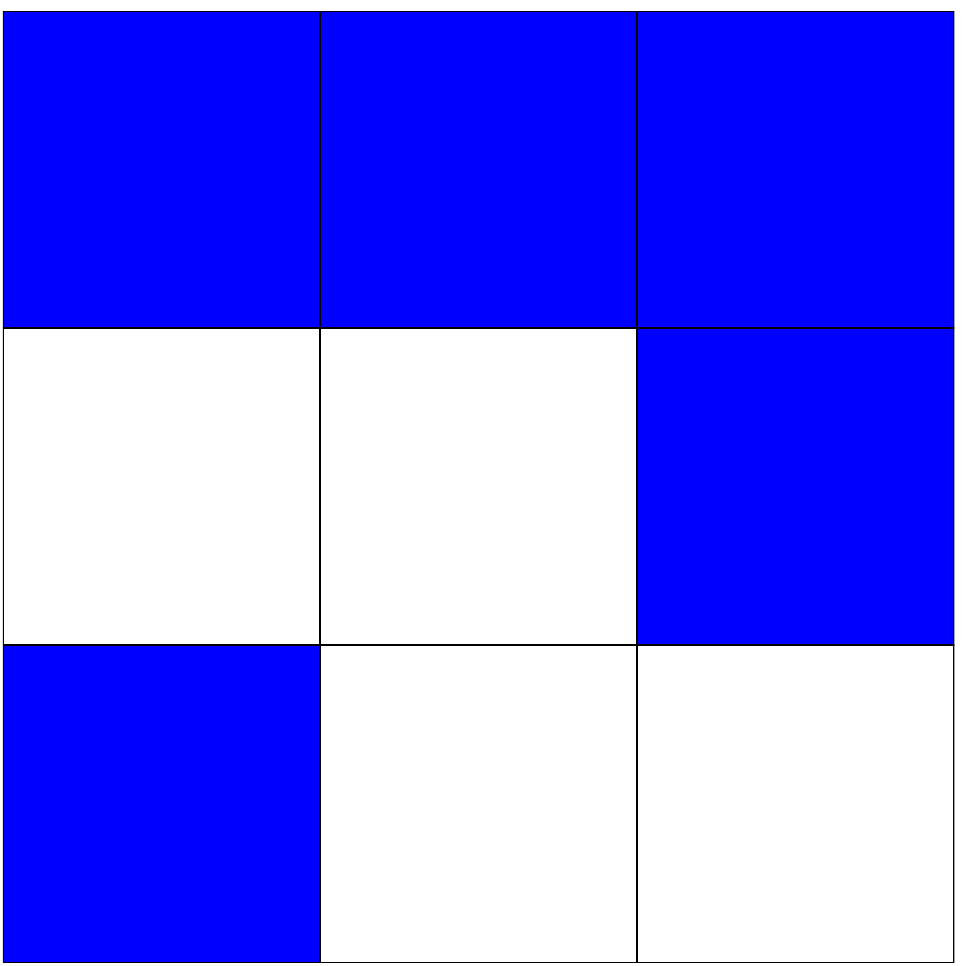} & \includegraphics[width=0.120000\linewidth]{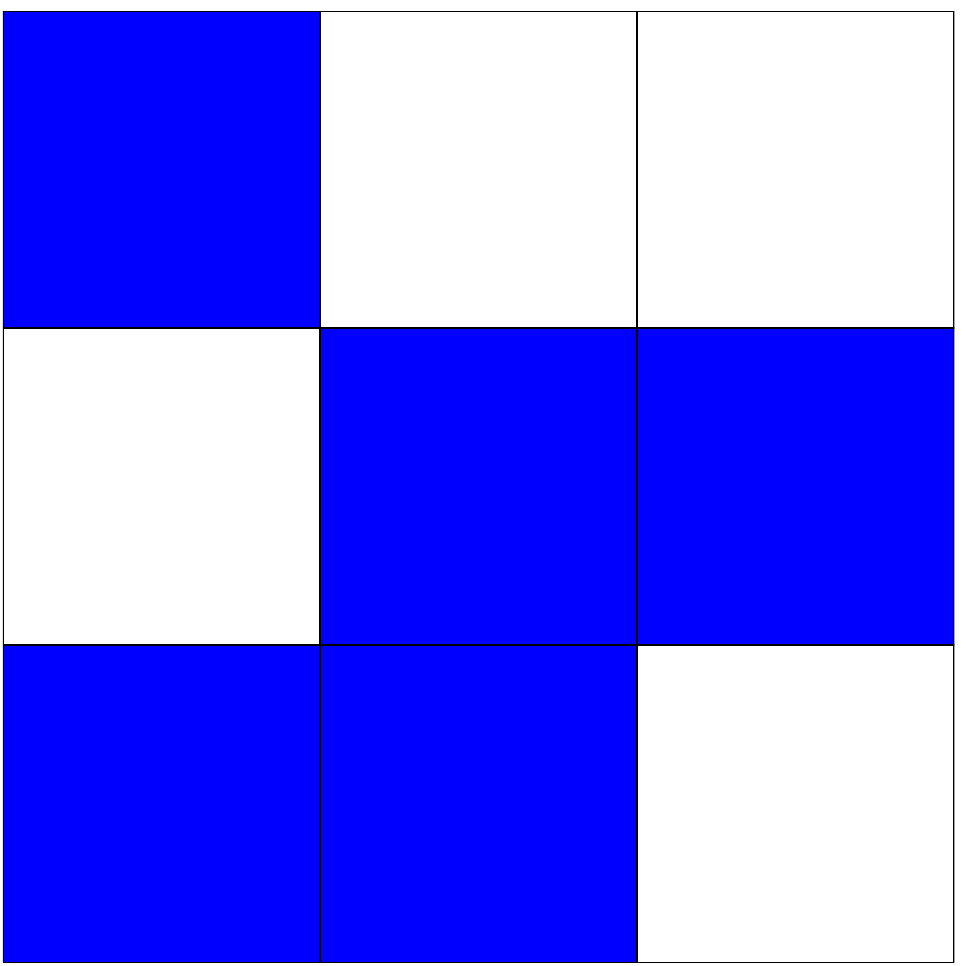} \\ \hline
& & & & \\ 
\includegraphics[width=0.120000\linewidth]{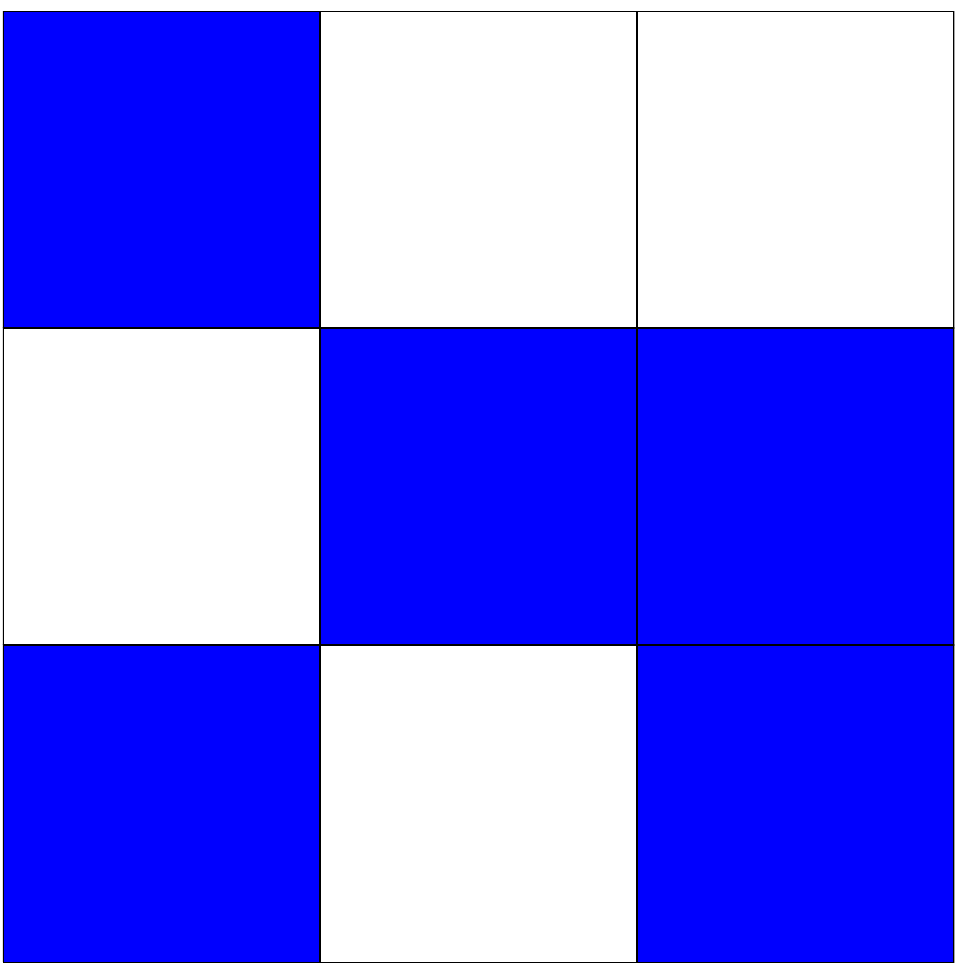} & \includegraphics[width=0.120000\linewidth]{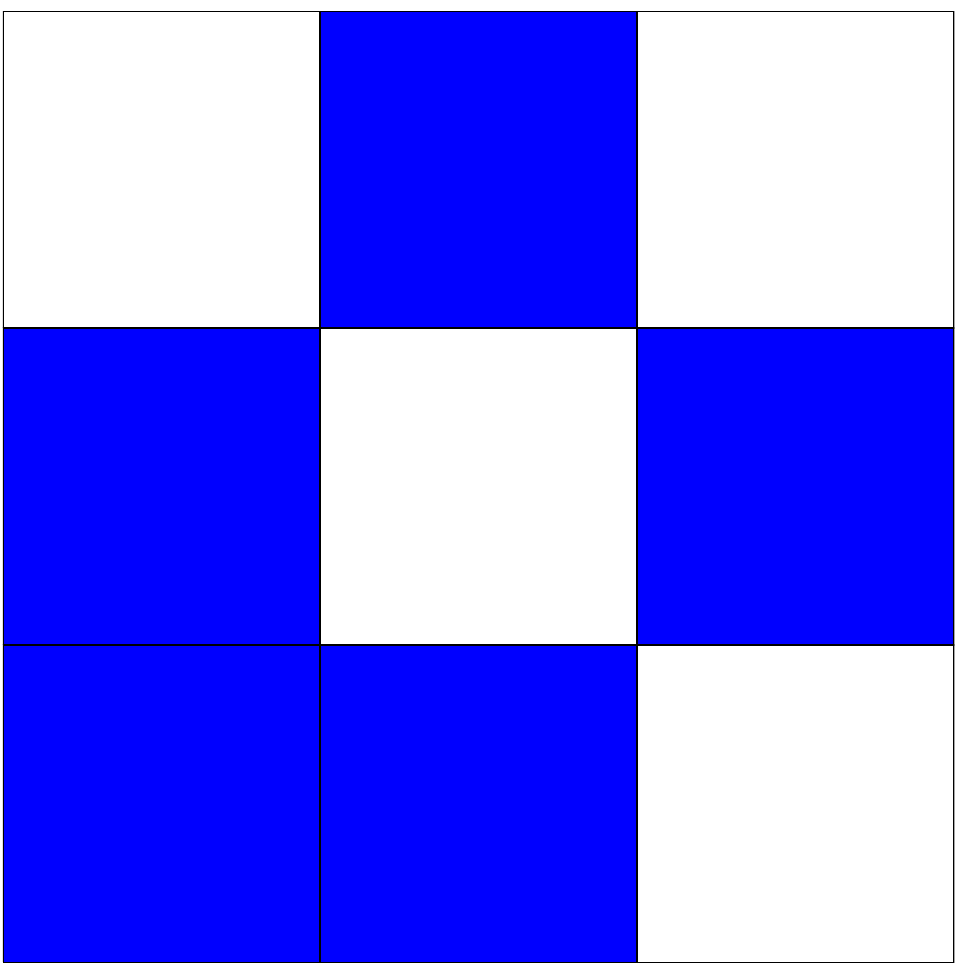} & \includegraphics[width=0.120000\linewidth]{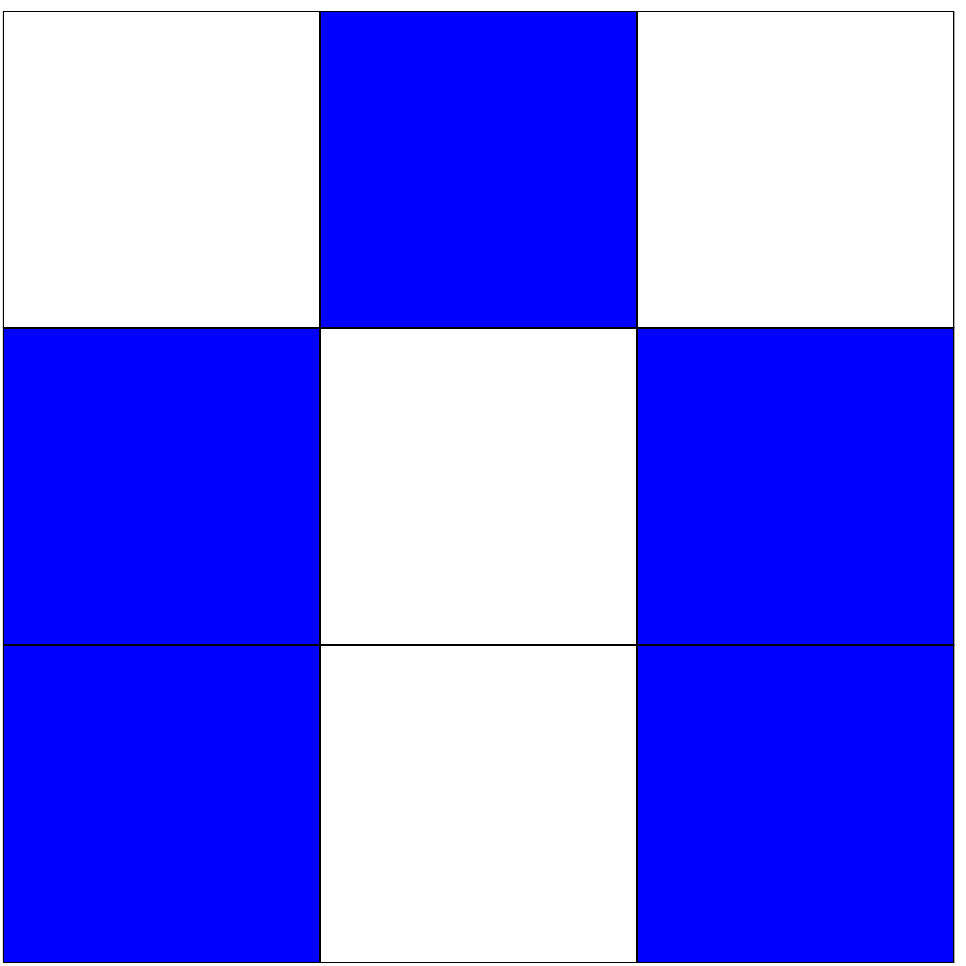} & \includegraphics[width=0.160000\linewidth]{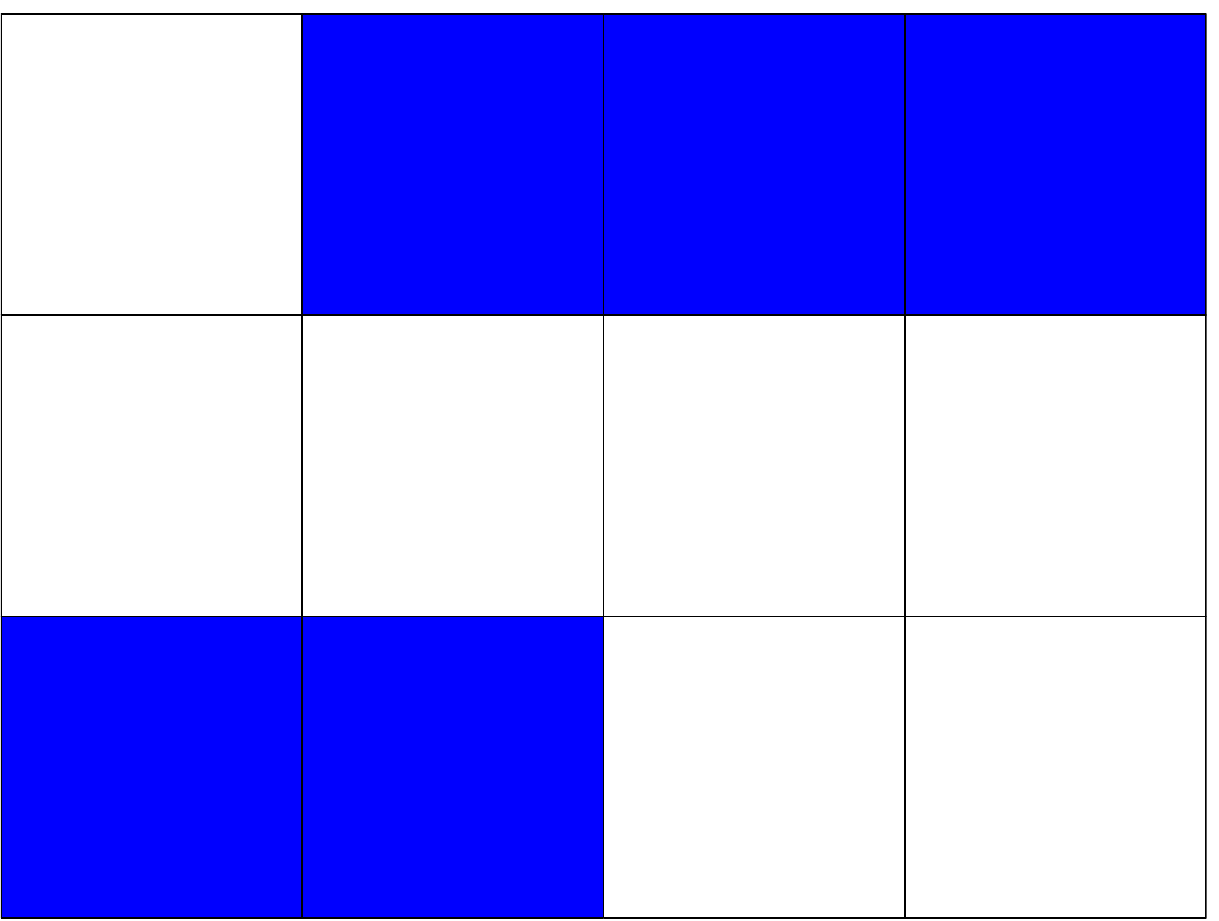} & \\ \hline
\end{tabular}
\caption{Impossible pieces for n=5 and k=1.}
\label{tab:impossible-51}
\end{table}


\pagebreak
\subsection{Unknown polyominoes}

We were not able to determine the status of a small number of polyominoes (Table~\ref{tab:unknown}).
The table lists bounds on the minimum and maximum sides of the rectangle, which
should help future solvers. We obtained the maximum results in column 3 by calling Algorithm~\ref{alg:solver} and
incrementally increasing the dimensions of the grid. For the minimum results we used a different algorithm,
but we do not describe it here.

\begin{table}[h]
\centering
\begin{tabular}{|c|c|c|}
\hline
Piece & min side cannot be $\leq K$ & max side cannot be $\leq K$ \\ \hline
  & & \\ 
\includegraphics[width=0.200000\linewidth]{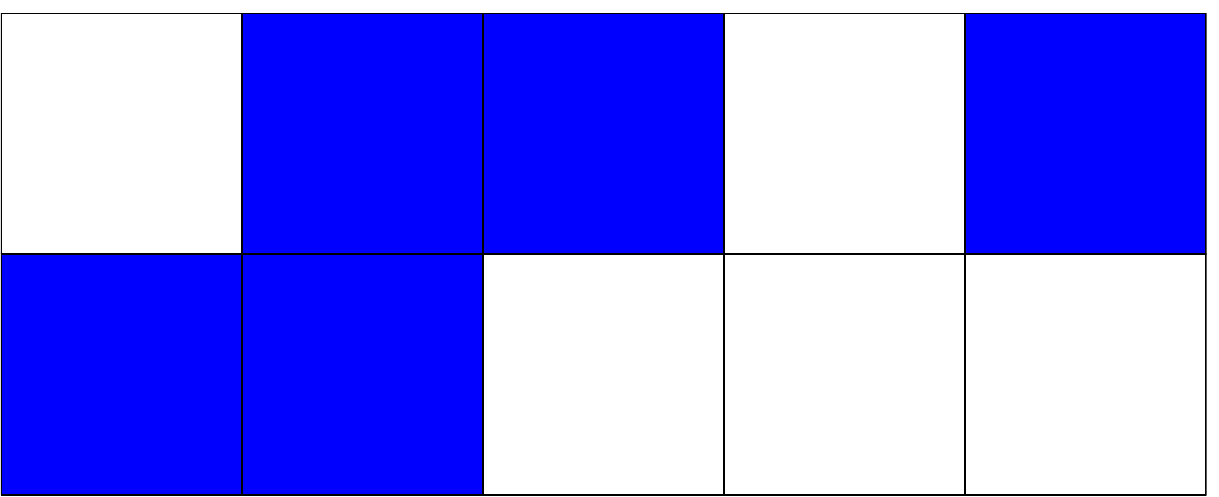} & 18 & 26 \\ \hline
  & & \\ 
\includegraphics[width=0.160000\linewidth]{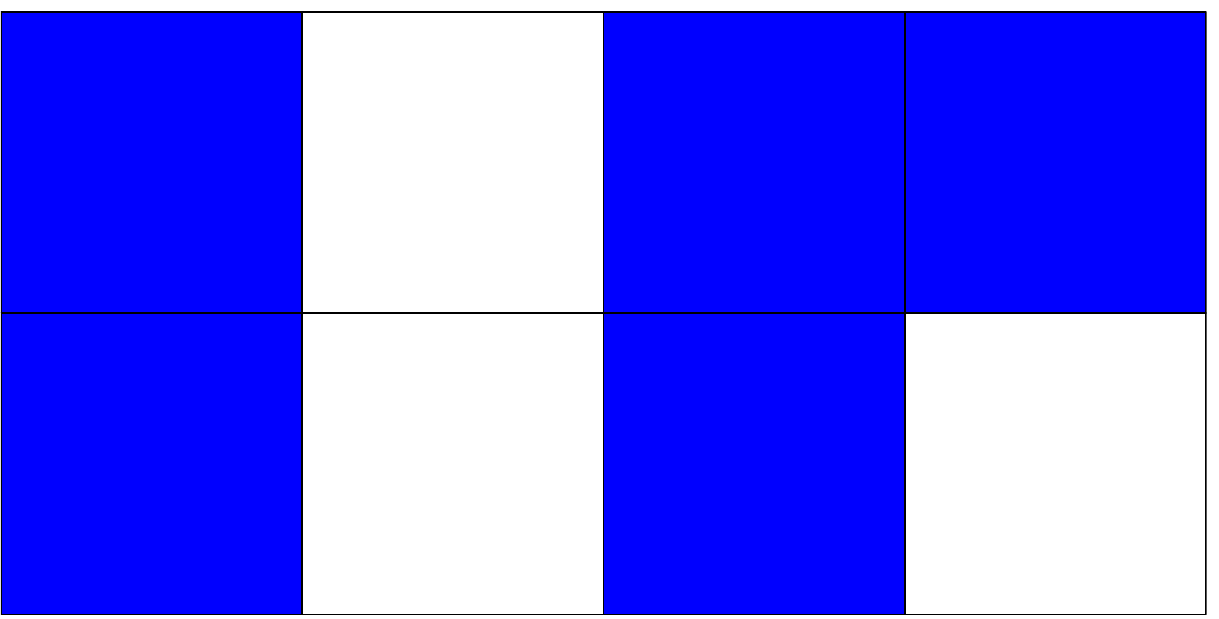} & 17 & 20 \\ \hline
  & & \\ 
\includegraphics[width=0.160000\linewidth]{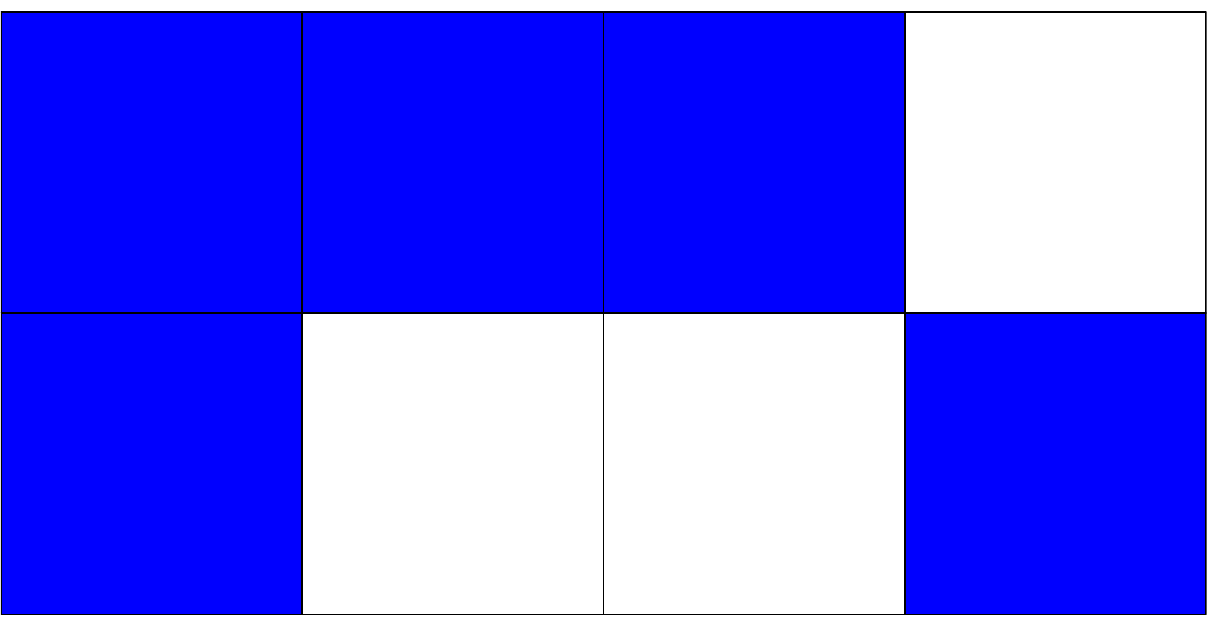} & 20 & 28 \\ \hline
  & & \\ 
\includegraphics[width=0.160000\linewidth]{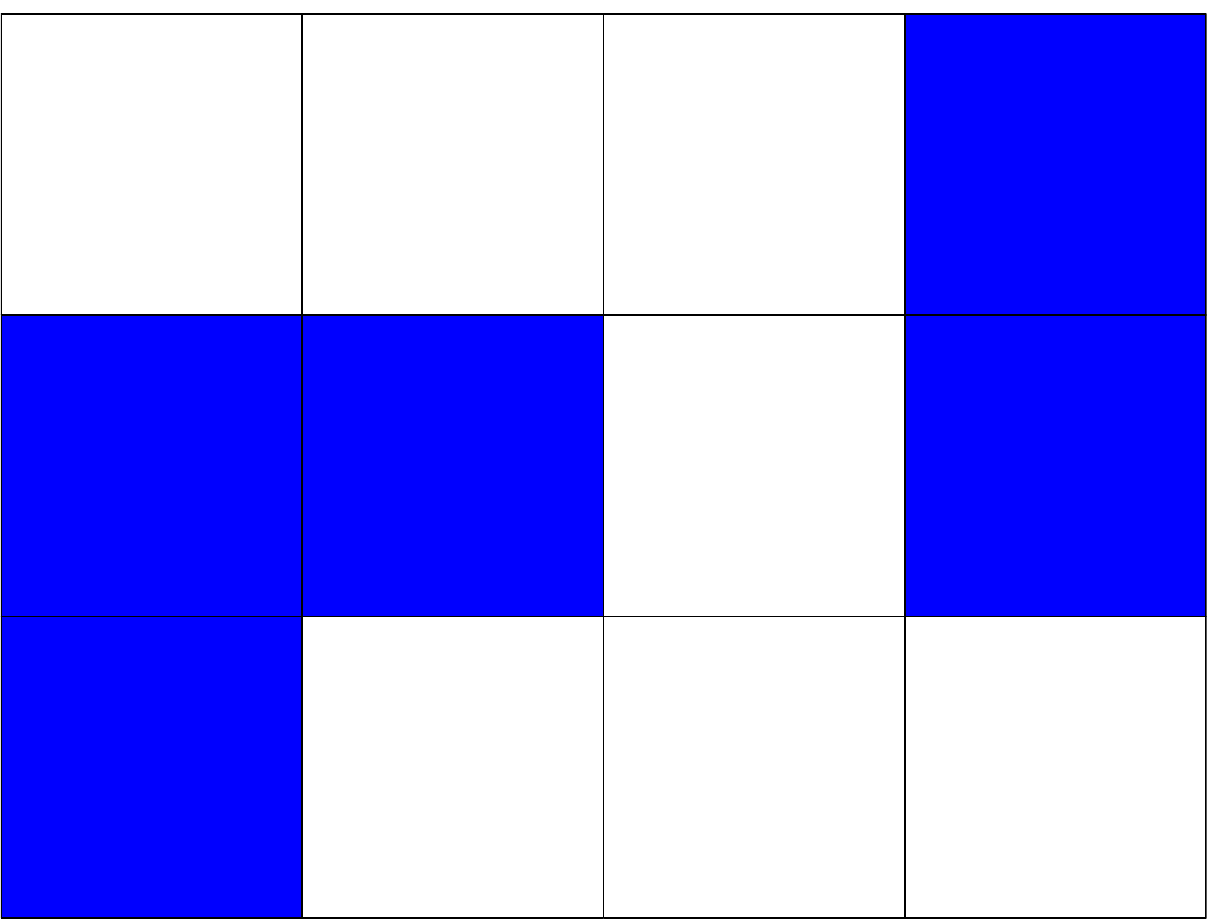} & 20 & 28 \\ \hline
  & & \\ 
\includegraphics[width=0.160000\linewidth]{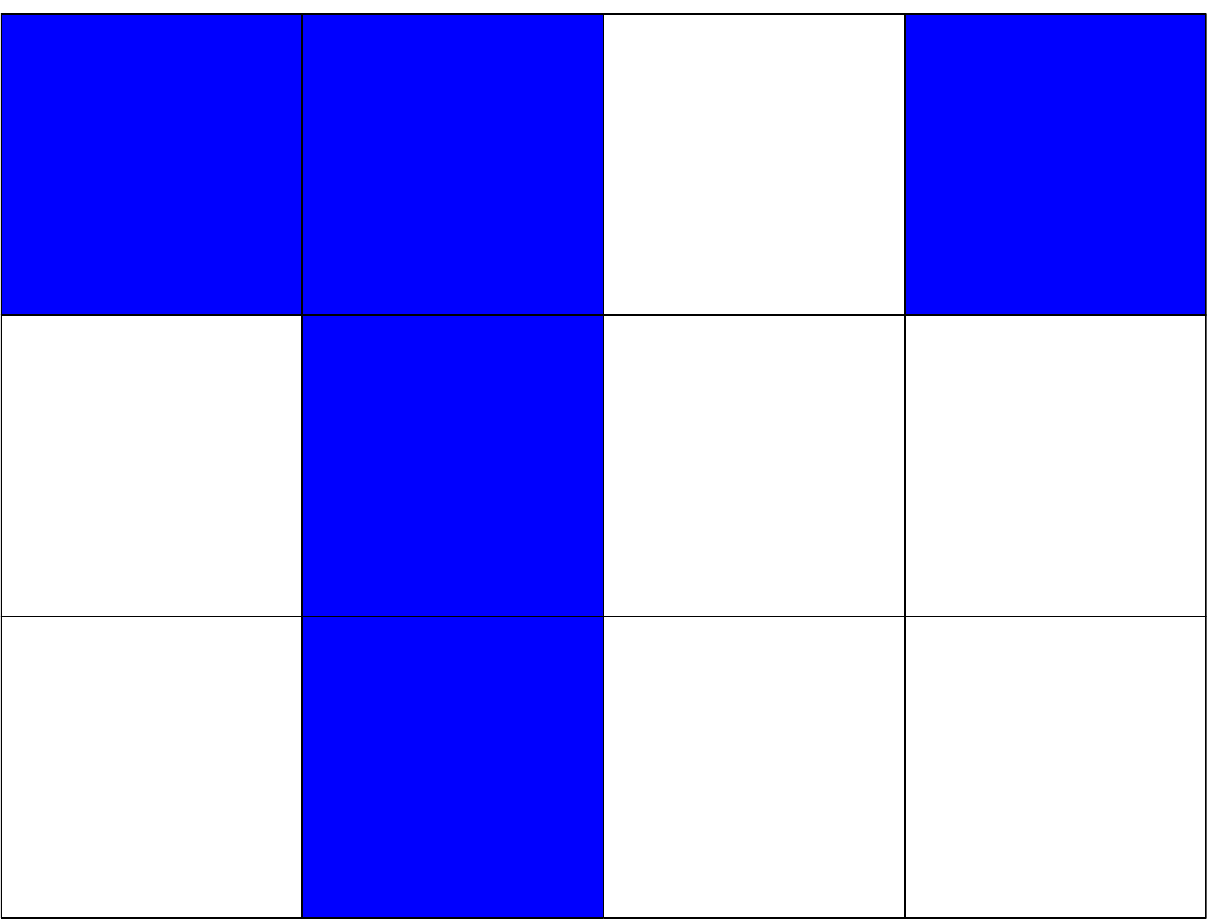} & 20 & 24 \\ \hline
  & & \\ 
\includegraphics[width=0.160000\linewidth]{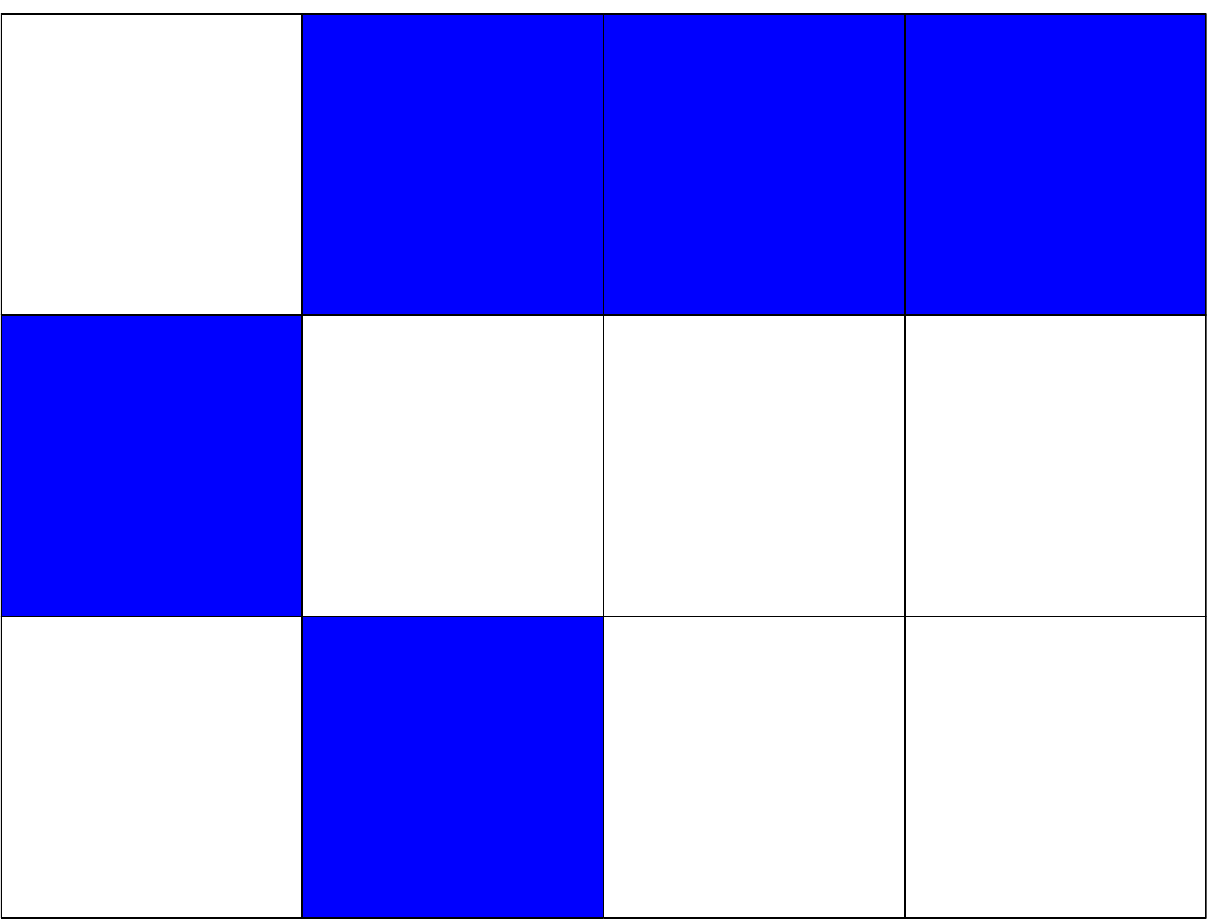} & 19 & 26 \\ \hline
\end{tabular}
\caption{Results for unknown pieces.}
\label{tab:unknown}
\end{table}

\section{Future work}

Future work will focus on creating faster and more powerful solvers for this problem. Our efforts will concentrate on
determining whether the unknown polyominoes listed in Table~\ref{tab:unknown} are rectifiable. We would also like to obtain
results for a larger set of polyominoes, such as those listed in Table~\ref{tab:future}.

\begin{table}[!htpb]
\centering
\begin{tabular}{|c|c|c|}
\hline
n & k & Total \\ \hline
3 & 3 & 17 \\ \hline
3 & 4 & 32 \\ \hline
4 & 2 & 60 \\ \hline
4 & 3 & 151 \\ \hline
5 & 2 & 302 \\ \hline
\end{tabular}
\caption{Unexplored classes of polyominoes.}
\label{tab:future}
\end{table}

There are many more tiling problems that can be explored with holey polyominoes. We can investigate tiling of the plane
or tiling of larger copies of polyominoes. We can also attempt tiling with sets of polyominoes. Finally, we can look at the
\emph{compatibility problem} - given a pair of polyominoes find a figure that can be tiled with each one.

\section{Conclusions}
We introduced a new type of polyominoes, affectionately called the holey polyominoes.
These polyominoes differ from regular ones by containing some invisible squares. We showed how these
polyominoes can tile rectangles. We analysed the rectifiability of all such polyominoes up to 5 visible squares.
We were able to determine the status (rectifiable or not) of all but 6 polyominoes.

\section{Acknowledgements}
We would like to thank Greg O'Keefe for his valuable suggestions. We are grateful to Tom Sirgedas for showing that one of our unknown
polyominoes is in fact unrectifiable.

\vspace{1cm}
\bibliographystyle{abbrv}
\bibliography{shortbib}

\end{document}